\documentclass[sigconf, nonacm]{acmart}
\AtBeginDocument{%
  }

\setcopyright{acmlicensed}
\copyrightyear{2018}
\acmYear{2018}
\acmDOI{XXXXXXX.XXXXXXX}
\acmConference[Conference acronym 'XX]{Make sure to enter the correct
  conference title from your rights confirmation email}{June 03--05,
  2018}{Woodstock, NY}
\acmISBN{978-1-4503-XXXX-X/2018/06}

\usepackage[utf8]{inputenc} 
\usepackage[T1]{fontenc}    
\usepackage{hyperref}       
 
\usepackage{cuted}
\usepackage{url}            
\usepackage{booktabs}       
\usepackage{amsfonts}       
\usepackage{nicefrac}       
\usepackage{microtype}      
\usepackage[dvipsnames]{xcolor}         
\definecolor{mypurple}{HTML}{6663A1}
\definecolor{myred}{HTML}{ED7D8D}
\definecolor{myblue}{HTML}{7CC9ED}
\definecolor{mygray}{HTML}{D9D9D9}
\usepackage{enumitem}
\usepackage{amsthm}
\usepackage{amsmath}
\usepackage{bm}

\usepackage{amssymb}
\usepackage{mathtools}
\usepackage[linesnumbered,ruled]{algorithm2e}
\usepackage{makecell}
\usepackage{multirow}
\usepackage{subfigure} 
\usepackage{booktabs}
\usepackage{wrapfig}
\usepackage{extarrows}
\usepackage{colortbl}
\usepackage{tikz}
\usepackage{pifont}
\usepackage[most]{tcolorbox}
\usepackage{setspace}
\usepackage{adjustbox}
\usepackage{fvextra}
\tcbuselibrary{breakable}
\usepackage{longtable}
\usepackage{wasysym}

\newcommand{\modelname}{\textsc{Gfm}-Retriever}
\newcommand{\ie}{\textit{i.e.}}
\newcommand{\eg}{\textit{e.g.}}

\newcommand{\etc}{\textit{etc.}}

\newcommand{\mycite}[1]{\hyperlink{cite.#1}{$\rhd$}}

\renewcommand\eqref[1]{\textup{(\ref{#1})}}

\theoremstyle{definition}

\renewcommand\footnotetextcopyrightpermission[1]{} 
\newtcolorbox{mathbox}{
  colframe=black,
  colback=white,
  rounded corners,
  boxrule=0.5pt,
  breakable,
  left=0.5pt, right=0.5pt, top=0.5pt, bottom=0.5pt
}

\begin{document}

\title[Retrieving Minimal and Sufficient Reasoning Subgraphs with Graph Foundation Models for Path-aware GraphRAG]{Retrieving Minimal and Sufficient Reasoning Subgraphs\\with Graph Foundation Models for Path-aware GraphRAG}

\author{Haonan Yuan\textsuperscript{1}, Qingyun Sun\textsuperscript{1}, Junhua Shi\textsuperscript{1}, Mingjun Liu\textsuperscript{2}, Jiaqi Yuan\textsuperscript{1}, Ziwei Zhang\textsuperscript{1},\\Xingcheng Fu\textsuperscript{3}, Jianxin Li\textsuperscript{1}}
\affiliation{%
  \institution{\textsuperscript{1}Beihang University \quad \textsuperscript{2}University of Electronic Science and Technology of China \quad \textsuperscript{3}Guangxi Normal University}
}

\renewcommand{\shortauthors}{Preprint. Under review.}

\begin{abstract}
Graph-based retrieval-augmented generation (GraphRAG) exploits structured knowledge to support knowledge-intensive reasoning. 
However, most existing methods treat graphs as intermediate artifacts, and the few subgraph-based retrieval methods depend on heuristic rules coupled with domain-specific distributions. They fail in typical cold-start scenarios where data in target domains is scarce, thus yielding reasoning contexts that are either informationally incomplete or structurally redundant.
In this work, we revisit retrieval from a structural perspective, and propose \textbf{\modelname} that directly responds to user queries with a subgraph, where a pre-trained \underline{\textbf{G}}raph \underline{\textbf{F}}oundation \underline{\textbf{M}}odel acts as a cross-domain \underline{\textbf{Retriever}} for multi-hop path-aware reasoning.
Building on this perspective, we repurpose a pre-trained GFM from an entity ranking function into a generalized retriever to support cross-domain retrieval.
On top of the retrieved graph, we further derive a label-free subgraph selector optimized by a principled Information Bottleneck objective to identify the 
query-conditioned subgraph, which contains informationally sufficient and structurally minimal golden evidence in a self-contained ``core set''.
To connect structure with generation, we explicitly extract and reorganize relational paths as in-context prompts, enabling interpretable reasoning.
Extensive experiments on multi-hop question answering benchmarks demonstrate that \modelname~achieves state-of-the-art performance in both retrieval quality and answer generation, while maintaining efficiency.
\end{abstract}

\maketitle

\section{Introduction}
\label{sec:intro}
Retrieval-augmented generation (RAG) has emerged as a central paradigm for enabling large language models (LLMs) to access and reason over external knowledge in a wide range of real-world applications~\cite{fan2024survey, xu2024ram, wiratunga2024cbr, wang2025omnieval}. To better facilitate multi-hop reasoning in knowledge-intensive question answering, recent works extend the RAG paradigm by incorporating knowledge graphs, giving rise to graph-based retrieval-augmented generation (GraphRAG)~\cite{zhang2025survey}. Despite this motivation, most GraphRAG systems ultimately reduce retrieval outputs to ranked lists of entities and documents~\cite{peng2024graph, han2024retrieval}. Such ranking-based retrieval abstracts away relational dependencies, forcing downstream answer generation to implicitly reconstruct reasoning paths from isolated evidence.

We demonstrate a case study in Figure~\ref{fig:intro}, which presents a user question whose answer depends on tracing multiple relational dependencies across entities~\cite{ghassel2025hierarchical}. While the entity ranking can traverse nodes that are individually relevant, it does not indicate how these entities interact or which relational paths are critical for inference. As a result, the reasoning process remains implicit and fragile. In contrast, a query-specific subgraph explicitly preserves relational structure, making multi-hop dependencies observable and providing the core structural context for generation~\cite{li2025disentangling}. This comparison suggests that responding to a subgraph rather than a flat ranking offers a more informative retrieval interface for complex reasoning.

\begin{figure}[!t]
    \centering
    \includegraphics[width=\linewidth]{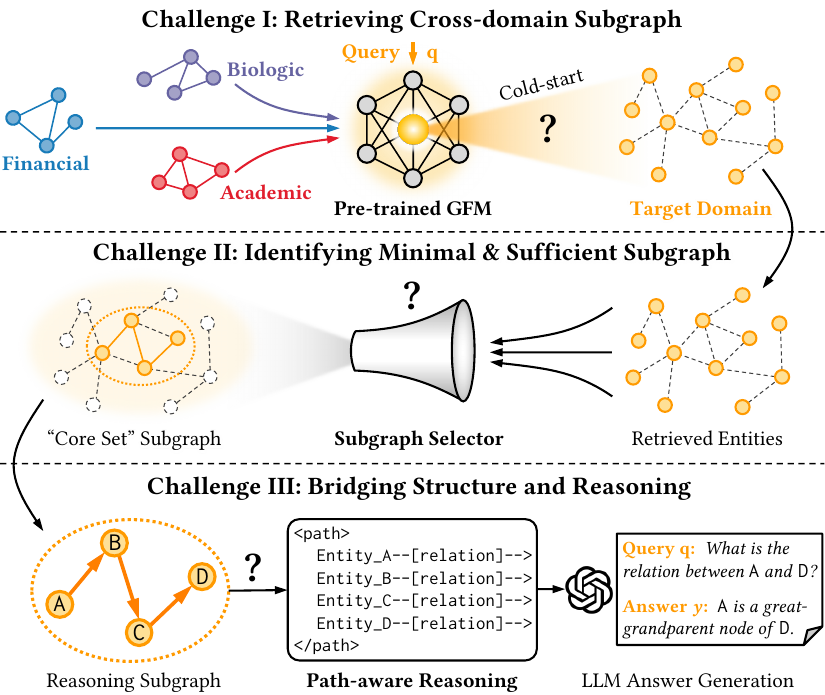}
    \vspace{-0.6cm}
    \caption{Challenges in subgraph-based RAG.}
    \label{fig:intro}
\vspace{-0.35cm}
\end{figure}

Motivated by this key observation, a growing line of GraphRAG methods has started to incorporate subgraph-level retrieval to better support complex reasoning~\cite{li2025simple, gao2025xdgnn}. Instead of returning isolated entities or documents, these works aim to construct query-relevant substructures that expose relational dependencies.
However, they typically rely on heuristic rules that determine which entities and relations to include~\cite{zhang2022grasp}. Crucially, these heuristics are often implicitly coupled to domain-specific distributions, making them fragile in typical cold-start scenarios where target-domain data are scarce. Without sufficient supervision, they lack a principled mechanism to balance informative sufficiency with structural minimality. Consequently, they risk producing reasoning contexts that are either informationally incomplete or structurally redundant.

The aforementioned observations suggest that: although subgraphs offer a more faithful intermediate representation than ranked entities, several fundamental challenges remain unaddressed:

\textbf{Challenge I: Retrieving Cross-domain subgraph.}
In existing GraphRAG systems, retrieval is tightly coupled with query-specific propagation on a fixed knowledge graph~\cite{he2024g, mavromatis2025gnn}. While effective in closed settings, such designs struggle to generalize to unseen domains characterized by distinct distributions. In typical cold-start scenarios, training a domain-specific retriever is often infeasible due to data scarcity, and directly applying a pre-trained retriever tends to fail because of distribution misalignment. A key challenge is therefore to establish a generalized retrieval paradigm that decouples learned retrieval patterns from specific domains.

\textbf{Challenge II: Identifying the minimal yet sufficient subgraph.}
Retrieving subgraphs introduces a fundamental trade-off between sufficiency and redundancy. Larger subgraphs increase the chance of containing all relevant evidence, but also introduce substantial redundancy, whereas smaller subgraphs risk missing critical reasoning paths. Existing methods typically constrain this trade-off through heuristic budgets, fixed hop limits, or manually tuned pruning strategies~\cite{li2025simple}, which are insensitive to query difficulty and hard to generalize. Without explicit supervision on what constitutes a ``ground-truth subgraph'', it remains challenging to learn a principled, query-adaptive notion of sufficiency.

\textbf{Challenge III: Bridging structure and reasoning.}
Even when a query-relevant subgraph is retrieved, its structural information is often underutilized by downstream generators. Standard practices frequently collapse the retrieved subgraph into document sets or flattened triple lists, which inadvertently discard the relational dependencies essential for multi-hop inference~\cite{zhang2023graph}. Consequently, LLM is forced to implicitly reconstruct reasoning chains from a fragmented context, rendering the process uninterpretable. A practical challenge is to preserve and reorganize these relational paths to explicitly guide the reasoning process of the LLM.
To address the challenges identified above, we propose \textbf{\modelname}, a subgraph-based RAG framework that revisits retrieval from a structural perspective by directly responding to user queries with subgraphs. In contrast to existing methods that tightly couple retrieval with query-time graph traversal or heuristic expansion, \modelname~elevates subgraphs to the retrieval interface itself and learns to retrieve them in a principled manner.
At the core of \modelname~is a pre-trained \underline{\textbf{G}}raph \underline{\textbf{F}}oundation \underline{\textbf{M}}odel acting as a generalized \underline{\textbf{Retriever}}, which enables efficient and transferable subgraph retrieval across domains.
To address the lack of principled subgraph selection, we derive a label-free subgraph selector optimized under an Information Bottleneck objective, explicitly balancing informative sufficiency and structural minimality.
Finally, to make retrieved structures actionable for answer generation, \modelname~extracts and reorganizes relational paths from subgraph into structured in-context prompts, allowing the graph structure to directly guide multi-hop reasoning.
\textbf{Our contributions:}
\begin{itemize}[leftmargin=*]
    \item We propose \modelname, a subgraph-based RAG that retrieves query-specific subgraphs. To the best of our knowledge, \modelname~is the first work that leverages a cross-domain GFM to learn minimal yet sufficient subgraphs for answer generation.
    \item We introduce a label-free Information Bottleneck formulation for subgraph selection and derive a tractable optimization objective, enabling principled identification of query-relevant structure that explicitly supports multi-hop, path-aware reasoning.
    \item Extensive experiments demonstrate \modelname~consistently improves both retrieval quality and end-to-end answer generation against state-of-the-art baselines.
\end{itemize}
\section{Related Work}
\label{sec:related_work}
\subsection{Graph-based RAG}
Retrieval-augmented generation (RAG) grounds the large language models in externally retrieved knowledge to improve the factuality in knowledge-intensive tasks~\cite{zhao2023survey, huang2023survey}. Most methods follow a ``retrieve-then-generate'' paradigm that incorporates retrieved text as context~\cite{fan2024survey}, but typically treat evidence as unstructured. Graph-based retrieval-augmented generation (GraphRAG) introduces graphs as knowledge structures that explicitly encode relational information~\cite{zhang2025survey}. By organizing entities as nodes connected by edges, GraphRAG enables relation-aware retrieval and reasoning and has shown benefits for multi-step reasoning~\cite{edge2024local, peng2024graph, han2025reasoning}. However, existing GraphRAG methods rely on large or fixed structures that introduce redundant or task-irrelevant information, highlighting the challenge of constructing task-relevant graphs.

\subsection{Graph Foundation Models}
Graph foundation models (GFMs) aim to learn transferable graph representations via large-scale self-supervised pre-training~\cite{liu2022graph, xie2022self}. While effective within the limited domains, most methods underestimate distribution shifts between pre-training and downstream, leading to degraded transferability~\cite{huang2022few}. Recent work explores cross-domain graph learning and LLM-enhanced alignment to improve generalization~\cite{zhao2024all, he2024unigraph, tang2024higpt}, but often relies on implicit representation similarity or restricted graph assumptions. Though some studies apply graph foundation models as retrievers for RAG systems~\cite{luo2025gfm}, their pre-training remains a semantic ``black box'' without explicit cross-domain alignment, limiting its generalizability.

\subsection{Critical Subgraph Mining}
Although graph structure encodes rich relational semantics, only a small part of subgraphs are truly informative, which can be decisive for downstream tasks~\cite{yang2018node}. Early work relies on pre-defined structural patterns, such as kernels or motifs, which provide strong inductive bias but remain heuristic, and difficult to scale or adapt across tasks~\cite{kriege2020survey, monti2018motifnet}. More recent approaches explore learnable or information-theoretic criteria to identify task-relevant subgraphs~\cite{yu2020graph, sui2022causal}, yet often lack explicit semantic preservation or minimality guarantees. In GraphRAG, subgraphs are rarely used as retrieval units, and existing attempts focus on coverage rather than task sufficiency, leaving the problem of identifying compact and critical subgraphs largely unresolved~\cite{li2025simple}.
\section{Preliminary}
\label{sec:preliminary}
\textbf{Notation.}
A knowledge graph is denoted as $\mathcal{G}=(\mathcal{V},\mathcal{R},\mathcal{E})$,
where $\mathcal{V}$ is entity set, $\mathcal{R}$ denotes relations, and $\mathcal{E}\!\subseteq\! \mathcal{V}\!\times\!\mathcal{R}\!\times\!\mathcal{V}$ denotes the relational triples.
$\mathcal{G}$ is indexed from a corpus of documents $\mathcal{C}=\{c_i\}$ from different knowledge domains $\{\mathcal{D}\}$. Each $c_i$ links to a set of entities $\phi(c_i)\subseteq\mathcal{V}$.
Denote the Graph Foundation Model as $f_{\boldsymbol{\theta}}(\cdot)$, where the parameter $\boldsymbol{\theta}$ is pre-trained on $\mathcal{G}$ with the KG completion task, and then fine-tuned with the subgraph retrieval task.
Given the query $\mathbf{q}$, $f_{\boldsymbol{\theta}}^\star(\cdot)$ acts as a query-dependent retriever for a subgraph $\mathcal{G}_{\mathbf{q}}=f_{\boldsymbol{\theta}}^\star(\mathcal{G}^T,\mathbf{q})$ on the target knowledge graph $\mathcal{G}^T$.




\textbf{Subgraph-based RAG.}
Given a query $\mathbf{q}$, the retrieved subgraph $\mathcal{G}_{\mathbf{q}}$ is denoted as $\mathcal{G}_{\mathbf{q}}=(\mathcal{V}_{\mathbf{q}},\mathcal{R}_{\mathbf{q}},\mathcal{E}_{\mathbf{q}})$, which specifies both the relevant entities and their dependencies. To ground the retrieved entities in text, we derive a supporting document set $\mathcal{C}_{\mathbf{q}}\subseteq\mathcal{C}$ by mapping entities back to their corpus as $\mathcal{C}_{\mathbf{q}}=\bigcup_{e\in\mathcal{V}_{\mathbf{q}}}\phi^{-1}(e)\!=\!\{c\!\in\!\mathcal{C}\mid \phi(c)\cap \mathcal{V}_{\mathbf{q}}\neq \emptyset\}$, where $\phi^{-1}(e)$ denotes the chunks mentioning entity $e$. The downstream answer generator is an LLM that conditions on the query together with the retrieved context to produce an answer $y=\mathrm{LLM}(\mathbf{q},\mathcal{G}_{\mathbf{q}},\mathcal{C}_{\mathbf{q}})$.

\begin{figure*}[!t]
    \centering
        \begin{minipage}{\textwidth}
            \hspace{-0.1cm}
            \begin{adjustbox}{width=1.01\linewidth}
            \input{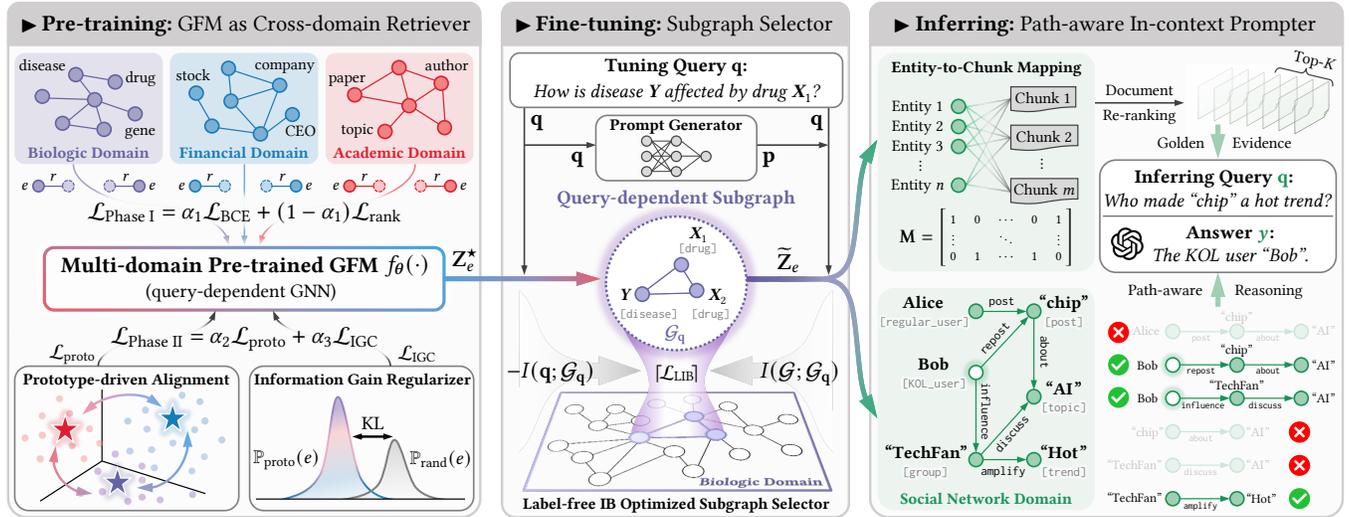}
            \end{adjustbox}
        \end{minipage}
    \vspace{-0.3cm}
    \caption{An overview of the \modelname~framework. \textbf{(1)} Pre-training: A query-conditioned GFM is trained as a cross-domain retriever. \textbf{(2)} Fine-tuning: A label-free IB-optimized selector identifies minimal sufficient query-specific subgraphs. \textbf{(3)} Inferring: Retrieved subgraphs are transformed into path-aware in-context prompts to guide multi-hop reasoning. Click ``$\blacktriangleright$'' to navigate.}
\label{fig:framework}
\end{figure*}
\section{Proposed Framework: \modelname}
\label{sec:method}
In this section, we introduce the proposed \modelname.

\subsection{Generalized GFM as Cross-domain Retriever}
\label{sec:method_gfm}
To address cross-domain retrieval under data-scarce and cold-start scenarios, we establish a GFM-based retriever that is capable of generalizing across different knowledge domains by learning a unified yet discriminative representation space across multi-domain KGs and supports consistent retrieval.

\textbf{Query-dependent GFM Message-Passing.}
Different from traditional GNN message-passing, which aggregates independently of task~\cite{kipf2017semi}, the GFM retriever employs a query-dependent message-passing mechanism that dynamically adapts to input queries~\cite{luo2025gfm}.
Given a query $\mathbf{q}$, the GFM retriever $f_{\boldsymbol{\theta}}(\cdot)$ first initialize entity ($e$) and relation ($r$) embeddings for all $e\in\mathcal{V}$, $r\in\mathcal{R}$:
\begin{equation}
    \mathbf{Z}_{e}^{(0)}=
    \begin{cases}
      \mathbf{q},  & e\in\{\mathcal{V}\mid \mathbf{q}\}, \\
      \textbf{0},  & \text{otherwise},
    \end{cases}
    ~\text{and}~
    ~\mathbf{Z}_{r}^{(0)}=\textsc{TextEmb}(r),
\label{eq:init}
\end{equation}
where $\mathbf{Z}_e$ and $\mathbf{Z}_r$ denote the embeddings of entities and relations, respectively, and $e \in \{\mathcal{V} \mid \mathbf{q}\}$ indicates the set of entities mentioned in the query $\mathbf{q}$. This initialization assigns higher activation to query-relevant entities, providing a contextualized starting point for subsequent propagation.
The message passing at the $(l+1)$-th layer then proceeds through relation-aware triple encoding, entity-wise aggregation, and representation update.
\begin{align}
    &\mathbf{m}_v^{(l+1)}=\textsc{Msg}\big( \mathbf{Z}_{v}^{(l)}, g^{(l+1)}(\mathbf{Z}_r^{(l)}), \mathbf{Z}_{u}^{(l)} \big),~ (u,r,v)\in \mathcal{E},\label{eq:message}\\
    &\mathbf{Z}_{v}^{(l+1)}=\textsc{Update}\big( \mathbf{Z}_{v}^{(l)}, \textsc{Agg}(\{\mathbf{m}_u^{(l+1)}\mid u\in\mathcal{N}(v)\}) \big),\label{eq:update}
\end{align}
where the message $\mathbf{m}$ captures triple-level interactions between the neighboring entities and relations. Specifically, $\textsc{Msg}(\cdot)$, the projection function $g(\cdot)$, $\textsc{Update}(\cdot)$, and $\textsc{Agg}(\cdot)$ are learnable components of the GFM $f_{\boldsymbol{\theta}}(\cdot)$ and are jointly pre-trained over the indexed multi-domain knowledge graphs.

Query-dependent GNNs are theoretically capable of multi-hop reasoning, but they assume the semantic consistency of the KG~\cite{yasunaga2021qa, huang2023theory}. Nevertheless, as the knowledge graph $\mathcal{G}$ in our setting is indexed from the multi-domain documents with varying semantic distributions, it potentially breaks the assumption. We next bridge this alignment gap by deriving from \citet{qiu2024understanding}.

    \begin{proposition}[\textbf{Multi-domain logical expressivity of query-conditioned GFM}]
    \label{prop:expressivity}
    Let $\mathcal{G}$ be a knowledge graph indexed from multiple domains. $\{\operatorname{Dom}_d(e)\}$ denotes unary predicates indicate domain attributes of entity $e$. For a query $q$, the query-conditioned GFM can learn a rule $\textnormal{\textsf{R}}(\textnormal{$\mathbf{q}$},e)$ if and only if $\textnormal{\textsf{R}}(\textnormal{$\mathbf{q}$},e)$ is a formula in graded modal logic over $\mathcal{G}$ with constant \textnormal{$\mathbf{q}$}, \ie, $\exists^{\geqslant N} e^\prime(\textnormal{\textsf{R}}(e^\prime,e)\wedge\operatorname{Dom}_d(e^\prime)\wedge\psi(e^\prime))$, where $\psi(\cdot)$ a subformula. \hfill\hyperref[app:proof_expressivity]{\textnormal{[$\rhd$Proof]}}
    \end{proposition}
In plain terms, Proposition~\ref{prop:expressivity} shows that a query-conditioned GFM over a multi-domain knowledge graph can represent exactly those query-entity relations definable by graded modal logic, with the query $\textsf{q}$ acting as a constant. This captures the multi-hop relational and domain-aware reasoning required for complex queries.

\textbf{Prototype-driven Multi-Domain Pre-training.}
An ideal GFM should exhibit homogenization across tasks and emergence of transferable knowledge~\cite{liu2025graph}, enabling it as a cross-domain retriever under distribution shifts. To achieve this, we pre-train $f_{\boldsymbol{\theta}}(\cdot)$ with KG completion task.
For each masked triple $(e,r,?)$ or $(?,r,e)$ randomly sampled from $\mathcal{G}$ over domains, $f_{\boldsymbol{\theta}}(\cdot)$ receives a pseudo-query $\mathbf{q}$ constructed from the known elements and predicts the missing entity with relevance score~\cite{luo2025gfm}:
\begin{equation}
    \mathbb{P}_\mathbf{q}(e)=\sigma\big(\textsc{MLP}(\mathbf{Z}^{(L)}_e)\big)\in\mathbb{R}^{|\mathcal{V}|\times 1},
\label{eq:relevance}
\end{equation}
which yields a binary cross-entropy (BCE) loss:
\begin{equation}
    \mathcal{L}_\text{BCE}=-\frac{1}{|\mathcal{V}_\mathbf{q}^+|}\sum_{e\in \mathcal{V}_\mathbf{q}^+}\log\mathbb{P}_\mathbf{q}(e)-\frac{1}{|\mathcal{V}_\mathbf{q}^-|}\sum_{e\in \mathcal{V}_\mathbf{q}^-}\log(1-\mathbb{P}_\mathbf{q}(e)),
\label{eq:bce}
\end{equation}
where $\mathcal{V}_\mathbf{q}^+$ denotes the set of entities relevant to the query $\mathbf{q}$, and $\mathcal{V}_\mathbf{q}^-$ denotes the remaining entities. To further alleviate the vanishing-gradient issue caused by the sparsity of positive entities, we additionally incorporate a ranking loss to enforce a margin between positive and negative entities~\cite{bai2023regression,luo2025gfm}:
\begin{equation}
    \mathcal{L}_\text{rank}=-\frac{1}{|\mathcal{V}_\mathbf{q}^+|}\sum \nolimits_{e\in \mathcal{V}_\mathbf{q}^+}\frac{\mathbb{P}_\mathbf{q}(e)}{\sum_{e^\prime\in \mathcal{V}_\mathbf{q}^-}}\mathbb{P}_\mathbf{q}(e^\prime).
\label{eq:rank}
\end{equation}
To improve cross-domain semantic consistency, we introduce a lightweight prototype-based contrastive continuous pre-training mechanism. For each pre-training domain $d \in \{\mathcal{D}\}$, we uniformly sample a subset of query-entity pairs $(\mathbf{q}i, e_i) \in \mathcal{Q}_d$, and define the semantic prototype $\mathbf{c}_d$ of domain $d$ as the mean embedding of the corresponding entities $\{e_i\}$.
The prototype-driven objective is:
\begin{equation}
    \mathcal{L}_{\text{proto}} = \sum_{(\mathbf{q}_i, e_i, d^+) \in \bigcup_d\{\mathcal{Q}_d\}} \log \frac{\exp\big(\mathrm{sim}(\mathbf{Z}_{e_i}^{(L)}, \mathbf{c}_{d^+}) / \tau\big)}{\sum_{d \in \mathcal{D}} \exp\big(\mathrm{sim}(\mathbf{Z}_{e_i}^{(L)}, \mathbf{c}_d) / \tau\big)},
\label{eq:proto}
\end{equation}
where $\{\mathbf{c}_d\}_{d \in \{\mathcal{D}\}}$ denote the prototypes of all domains, $\mathbf{c}_{d^+}$ corresponds to the ground-truth domain of $(\mathbf{q}_i, e_i)$, and $\tau$ is a temperature hyperparameter. This loss encourages retrieved entity representations to align with their corresponding domain prototypes while remaining separated from unrelated domain centers. The optimization is performed over a small sampled subset, ensuring efficiency during continuous pre-training.

\textbf{Regularization via Information Gain.}
While $\mathcal{L}_{\text{proto}}$ encourages semantic discriminability across domains, it does not explicitly distinguish informative semantic alignment from spurious correlations or random perturbations. To further emphasize domain-relevant information, we introduce an information gain-based regularizer~\cite{kent1983information, zhang2022information}.
Given the predicted distribution $\mathbb{P}_\mathbf{q}(e)$, the Information Gain Contrast (IGC) regularizer is defined as:
\begin{equation}
    \mathcal{L}_\text{IGC}=\operatorname{KL}(\mathbb{P}_\mathbf{q}(e)~\|~\mathbb{P}_\text{proto}(e))-\operatorname{KL}(\mathbb{P}_\mathbf{q}(e)~\|~\mathbb{P}_\text{rand}(e)),
\label{eq:igc}
\end{equation}
where $\operatorname{KL}(\cdot)$ denotes the Kullback-Leibler divergence.
$\mathbb{P}_\text{proto}(e)\propto\exp(\operatorname{sim}(\mathbf{Z}_e,\mathbf{c}_{d^+}))$ is a prototype-induced distribution reflecting domain-consistent semantics, while $\mathbb{P}_\text{rand}(e)$ is a randomized reference distribution constructed by shuffling domain prototypes.
By explicitly contrasting structured semantic alignment against randomness, $\mathcal{L}_\text{IGC}$ encourages $\mathbb{P}_\mathbf{q}(e)$ to assign higher confidence to entities grounded in domain-relevant semantics, thereby strengthening cross-domain consistency.

\textbf{Overall Pre-training Objective.}
The full pre-training loss combines retrieval supervision with semantic regularization:
\begin{equation}
    \mathcal{L}_{\text{pre}} = \underbrace{\alpha_1\mathcal{L}_{\text{BCE}} + (1-\alpha_1)\mathcal{L}_{\text{rank}}}_{\text{Phase I}} + \underbrace{\alpha_2 \mathcal{L}_{\text{proto}} + \alpha_3 \mathcal{L}_{\text{IGC}}}_{\text{Phase II}},
\label{eq:pre}
\end{equation}
where $\alpha_1$ is the trade-off hyper-parameter, $\alpha_2$ and $\alpha_3$ are the Lagrange multipliers.
Regarding the diversity of objective terms, we propose a two-phase pre-training strategy to enhance the stability.
During Phase I, we pre-train $f_{\boldsymbol{\theta}}(\cdot)$ for a fixed number of epochs with $\mathcal{L}_{\text{BCE}}$ and $\mathcal{L}_{\text{rank}}$ on KG completion task over mixed-domain triples to establish its general retrieval capability.
During Phase II, we apply continuous pre-training mechanism with $\mathcal{L}_{\text{proto}}$ and $\mathcal{L}_{\text{IGC}}$ on a subset for the fine-grained semantic alignment. This staged training scheme stabilizes optimization and avoids conflicts among heterogeneous loss signals introduced by prototype-level regularization.
We will specify details in Appendix~\ref{app:complexity_pretrain}.

\subsection{Label-free IB Optimized Subgraph Selector}
\label{sec:method_subgraph}
While the pre-trained retriever $f_{\boldsymbol{\theta}}^\star(\cdot)$ can identify relevant entities across domains, directly reasoning over the full graph is computationally expensive and offers limited structural interpretability.
Instead of returning a ranked list of entities, we learn a query-conditioned subgraph selector that identifies a minimal yet sufficient $\mathcal{G}_\mathbf{q}$ for reasoning.

\textbf{Query-conditioned Subgraph Selection.}
To efficiently adapt the pre-trained retriever $f_{\boldsymbol{\theta}}^\star(\cdot)$ to query-specific subgraph selection, we adopt a parameter-efficient fine-tuning (PEFT) strategy that freezes the backbone GFM parameters and introduces lightweight, query-conditioned adaptations for subgraph selection.
Given query $\mathbf{q}$ on the knowledge graph $\mathcal{G}$, we derive a shared query prompt $\mathbf{p} = \operatorname{MLP}(\mathbf{q})$, which is broadcast to construct query-conditioned entity embedding:
\begin{equation}
    \widetilde{\mathbf{Z}}_e = \big[\mathbf{Z}_e^\star~\|~\mathbf{W}_1 \mathbf{p}~\|~\mathbf{W}_2 \mathbf{q}\big],
\label{eq:query_e}
\end{equation}
where $\mathbf{Z}_e^\star$ is the initialized and then freezed entity embedding from $f_{\boldsymbol{\theta}}^\star(\cdot)$. $\mathbf{W}_1$, $\mathbf{W}_2$ are learnable weights. $\|$ denotes the concatenation.

To enable differentiable and optimizable subgraph sampling during fine-tuning, we employ the Gumbel-Sigmoid relaxation to approximate binary selection gates. Specifically, for each entity $e$, the selection probability $S_e(\mathbf{q})\in [0,1]$ is sampled as:
\begin{equation}
    S_e(\mathbf{q})=\left[ 1+\exp\left( -\sigma\big(\big(\operatorname{MLP}(\widetilde{\mathbf{Z}}_e)\big)+g_e\big)/\tau \right) \right]^{-1},
\label{eq:prob}
\end{equation}
where the variable $g_e$ is sampled from the standard Gumbel distribution, and $\tau$ is the temperature parameter.
This formulation provides a continuous and differentiable surrogate for discrete node selection, enabling stochastic yet trainable subgraph sampling within an end-to-end optimization framework.
Given a threshold $\epsilon$, the query-specific subgraph $\mathcal{G}_\mathbf{q}$ is instantiated with entity set $\mathcal{V}_\mathbf{q}=\{e\in\mathcal{V}^T\mid S_e(\mathbf{q}) > \epsilon\}$, from which the corresponding relation set $\mathcal{R}_{\mathbf{q}}$ and edge set $\mathcal{E}_{\mathbf{q}}$ are induced accordingly.
Importantly, the use of Gumbel-Sigmoid ensures that gradients can propagate through the selection process, making the subgraph construction fully differentiable and amenable to optimization under the downstream objectives. We next discuss its optimization.

\textbf{Label-free Information Bottleneck.}
We formulate the selector from an information-theoretic perspective. The Information Bottleneck (IB) principle~\cite{tishby2000information, tishby2015deep} provides an ideal way to optimize the subgraph selector that compresses full KG (\textit{minimal}) while preserving the query-relevant subgraph (\textit{sufficient}):
\begin{equation}
    \mathcal{L}_{\text{IB}} = -I(y;\mathcal{G}_\mathbf{q}) + \beta I(\mathcal{G};\mathcal{G}_\mathbf{q}),
\label{eq:ib}
\end{equation}
where $I(\cdot)$ is mutual information, and $\beta$ is Lagrange multiplier for trading off between compressing (the first term) and preserving (the second term). Minimizing $\mathcal{L}_{\text{IB}}$ encourages the selector to find a minimal yet sufficient ``coreset'' in $\mathcal{G}$.

However, as generating the answer $y$ is non-differentiable, directly optimizing Eq.~\eqref{eq:ib} is intractable, which causes the major challenge in finding $\mathcal{G}_\mathbf{q}$.
To address this dilemma, we establish a surrogate objective for label-free IB optimization by replacing the label $y$ with the query itself:
\begin{equation}
    \mathcal{L}_{\text{LIB}} = -I(\mathbf{q};\mathcal{G}_\mathbf{q}) + \beta I(\mathcal{G};\mathcal{G}_\mathbf{q}).
\label{eq:label_free}
\end{equation}
We next show $\mathcal{L}_{\text{LIB}}$ is a faithful proxy of $\mathcal{L}_{\text{IB}}$.
    \begin{proposition}[\textbf{Error Bound of $\mathcal{L}_{\textnormal{LIB}}$ Approximation}]
    \label{prop:error}
    Let $\mathcal{G}_\mathbf{q}$ be the subgraph, $y$ the ground-truth answer, and $\mathbf{q}$ the query. Assume the Markov property $\langle y \leftrightarrow \mathbf{q} \leftrightarrow \mathcal{G}_\mathbf{q} \rangle$ exists, the approximation error of $\mathcal{L}_{\textnormal{IB}}$ by $\mathcal{L}_{\textnormal{LIB}}$ is upper-bounded by the conditional entropy of $\mathbf{q}$ given $y$:
    \begin{equation}
        \left | I(y;\mathcal{G}_\mathbf{q}) - I(\mathbf{q};\mathcal{G}_\mathbf{q}) \right| \leqslant H(\mathbf{q}\mid y),
    \end{equation}
    where $H(\mathbf{q}\mid y)$ is a data-dependent constant.\hfill\hyperref[app:proof_error]{\textnormal{[$\rhd$Proof]}}
    \end{proposition}
Proposition~\ref{prop:error} declares that the approximation error is upper-bounded by the conditional entropy of the true label $y$ given the query $\mathbf{q}$. Once the data distribution is given and fixed, the approximation gap remains controllable within a constant bound by the informativeness of the query, which provides a guarantee that the quality of approximation is theoretically controlled by the informativeness of the query. 
Since Eq.~\eqref{eq:label_free} provides a theoretical form of fine-tuning objective, both mutual information terms are computationally infeasible. We next demonstrate their tractable derivations.

\begin{proposition}[\textbf{Lower Bound of $I(\mathbf{q};\mathcal{G}_\mathbf{q})$}]
\label{prop:lower}
    Denote the average-pooled subgraph representation as $\mathbf{G}_{\mathbf{q}}=\sum_{e\in\mathcal{V}_\mathbf{q}} S_e(\mathbf{q})\cdot \widetilde{\mathbf{Z}}_e$. Deriving from~\citet{alemi2017deep}, the lower bound is:
    \begin{equation}
        I(\mathbf{q};\mathcal{G}_\mathbf{q}) \geqslant -\mathcal{L}_\textnormal{NCE}=\log\frac{\exp\big( \operatorname{sim}(\mathbf{G}_{\mathbf{q}},\mathbf{q})/\tau \big)}{\sum_{\mathbf{q}^\prime \in\mathcal{B}}\exp\big( \operatorname{sim}(\mathbf{G}_{\mathbf{q}^\prime},\mathbf{q}^\prime)/\tau \big)},
    \label{eq:bound_1}
    \end{equation}
    where $(\mathbf{G}_{\mathbf{q}^\prime},\mathbf{q}^\prime)$ is the negative sample pair.\hfill\hyperref[app:proof_lower]{\textnormal{[$\rhd$Proof]}}
\end{proposition}

\begin{proposition}[\textbf{Upper Bound of $I(\mathcal{G};\mathcal{G}_\mathbf{q})$}]
\label{prop:upper}
    Denote the subgraph size and connectivity penalizers as:
    \begin{equation}
        \mathcal{L}_\textnormal{size}=\sum \nolimits_{e\in\mathcal{V}}S_e(\mathbf{q}), \;\;\mathcal{L}_\textnormal{con}=\operatorname{Tr}\big( \mathbf{S}_e(\mathbf{q})^\top\mathbf{L}\mathbf{S}_e(\mathbf{q}) \big),
    \label{eq:size_con}
    \end{equation}
    where $S_e(\mathbf{q})\in[\textnormal{0},\textnormal{1}]$ is entity selection probability, $\mathbf{S}_e(\mathbf{q}) \in \mathbb{R}^{|\mathcal{V}|}$ is the soft selection vector, $\mathbf{L}$ is the normalized graph Laplacian, and $\operatorname{Tr}(\cdot)$ calculates the trace of a matrix.
    The upper bound is:
    \begin{equation}
        I(\mathcal{G};\mathcal{G}_\mathbf{q}) \leqslant \beta_{\textnormal{1}} \mathcal{L}_\textnormal{size} + \beta_{\textnormal{2}}\mathcal{L}_\textnormal{con},
    \label{eq:bound_2}
    \end{equation}
    where $\beta_{\textnormal{1}}$ and $\beta_{\textnormal{2}}$ are hyper-parameters.\hfill\hyperref[app:proof_upper]{\textnormal{[$\rhd$Proof]}}
\end{proposition}

Proposition~\ref{prop:lower} and Proposition~\ref{prop:upper} provide a principled method for approximating $\mathcal{L}_\text{LIB}$ using tractable optimization objectives, while deriving such bounds under label-free settings is non-trivial. By establishing this connection, we ensure both theoretical soundness and practical guidance for subgraph identification in multi-domain knowledge graph retrieval.

\textbf{Overall Fine-tuning Objective.}
Following the pre-training stage, we fine-tune the GFM retriever on a supervised document retrieval task, where each query $\mathbf{q}$ is associated with a labeled target entity set $\mathcal{V}_{\mathbf{q}}^{+}$. While standard retrieval objectives encourage accurate ranking with respect to labeled entities, they do not explicitly constrain the structural properties of the retrieved subgraphs. To address this limitation, we incorporate subgraph-level regularization into the fine-tuning objective, resulting in the following:
\begin{equation}
    \mathcal{L}_\text{ftn}=\alpha_1\mathcal{L}_\text{BCE} + (1-\alpha_1)\mathcal{L}_\text{rank} + \left \lceil \mathcal{L}_\text{LIB} \right \rceil,
\label{eq:ftn}
\end{equation}
where $\mathcal{L}_\text{BCE}$ and $\mathcal{L}_\text{rank}$ provide complementary supervision for entity-level relevance and relative ordering. The additional term $\lceil \mathcal{L}_\text{LIB} \rceil$ acts as a structural regularizer that promotes the query-dependent minimal and sufficient subgraphs.
$\lceil \mathcal{L}_\text{LIB} \rceil$ is constructed by substituting the RHS of Eq.~\eqref{eq:bound_1} and Eq.~\eqref{eq:bound_2} into Eq.~\eqref{eq:label_free}:
\begin{equation}
    \left \lceil \mathcal{L}_\text{LIB} \right \rceil = \mathcal{L}_\text{NCE} + \beta_{\textnormal{1}} \mathcal{L}_\textnormal{size} + \beta_{\textnormal{2}}\mathcal{L}_\textnormal{con}.
\label{eq:lib_t}
\end{equation}

\newpage
\subsection{Relation Paths Induced In-context Prompter}
\label{sec:method_prompt}
Existing methods mainly retrieve relevant documents while leaving relations implicit. To expose key reasoning paths encoded in $\mathcal{G}_{\mathbf{q}}$, we treat the retrieved subgraph not only as an entity ``coreset'' but also as a structured reasoning scaffold for generation.

\textbf{Entity-to-Document Mapping.}
Following inverted index paradigm~\cite{luo2025gfm}, we construct a binary mapping matrix $\mathbf{M} \in \{0,1\}^{|\mathcal{E}_{\mathbf{q}}| \times |\mathcal{C}_{\mathbf{q}}|}$, where $\mathcal{C}_{\mathbf{q}}$ denotes the set of supporting document chunks. An entry $\mathbf{M}[e, c] = 1$ indicates that entity $e$ appears in document chunk $c$.
At inference time, rather than relying on heuristic top-$K$ entity selection, we restrict retrieval to the subgraph entities $\mathcal{V}_{\mathbf{q}}$ and assign each entity a composite importance weight defined as:
\begin{equation}
    w_e=S_e(\mathbf{q}) \cdot \frac{1}{\sum_c \mathbf{M}[e,c]}\cdot\left|\phi_\text{path}(e)\right|,
\label{eq:weight}
\end{equation}
where $S_e(\mathbf{q})$ denotes the soft selection probability introduced in Eq.~\eqref{eq:prob}, and $|\phi_\text{path}(e)|$ represents the number of shortest reasoning paths that involve entity $e$. This weighting scheme emphasizes entities that are both highly relevant to the query and structurally central within the reasoning subgraph.
Document relevance is subsequently computed by aggregating entity weights through the mapping matrix, given by $R_c = \mathbf{M}^\top w_e$, and the top-$K$ documents are selected accordingly.

\textbf{Reasoning Path Extraction.}
Enumerating all possible paths in a dense subgraph is computationally infeasible due to exponential growth in the number of candidate paths. To explicitly expose interpretable relational structures, we extract a compact set of reasoning paths $\Pi_{\mathbf{q}} = \{\pi_1, \ldots, \pi_l\}$ from the query-specific subgraph $\mathcal{G}_{\mathbf{q}}$ using an $l$-hop truncated depth-first search (DFS). Each extracted path takes the form $\pi = (e_0 \xrightarrow{r_1} e_1 \xrightarrow{r_2} \cdots \xrightarrow{r_k} e_k)$ with $k \leqslant l$.
To quantify the relevance of each reasoning path with respect to relation awareness, we define the following scoring function:
\begin{equation}
    S_r(\pi) = \sum_{i=0}^{k} S_{e_i}(\mathbf{q}) + \sum_{j=1}^{k} \omega(r_j),
\label{eq:sr}
\end{equation}
where $\omega(r_j)$ penalizes long-range dependencies via relation depth. This scoring favors paths traversing high-confidence nodes via semantically concise relations.

\textbf{In-context Prompter Construction.}
To prevent the subgraph from degenerating into an unordered set of documents and entities in the answer generation, we construct an in-context prompter for explicit path-aware LLM reasoning:
\begin{tcolorbox}[
    enhanced,
    colback=white,
    colframe=black,
    boxrule=0.7pt,
    rounded corners,
    fonttitle=\bfseries,
    title=Schema \hspace{0.55cm} Value,
    sidebyside,
    sidebyside align=top seam,
    lefthand width=1.5cm,
    boxsep=2pt,
    left=2pt,
    right=2pt,
    top=3pt,
    bottom=-8pt,
    sidebyside gap=10pt
]

\begin{spacing}{0.91}
{
\small 
\texttt{<system>}\\
\\
\\
\texttt{<doc>}\\
\texttt{<path>}\\
\\
\\
\\
\\
\\
\texttt{<query>}
}
\end{spacing}

\tcblower

\begin{spacing}{0.91}
{
\small
As an advanced reading comprehension assistant, your task is to analyze text passages and corresponding questions meticulously. Your responses start after ...\\
Document: \verb|{title} \n {content} \n|\\
Reasoning Path \verb|{path_id}| (score: \verb|{path_score}|):\verb| \n|\\
\verb|    |Entity \verb|{entity}| (type: \verb|{entity_type}|):\verb| \n|\\
\verb|        ---| (relation: \verb|{relation}|) \verb|---> \n|\\
\verb|    |Entity \verb|{entity}| (type: \verb|{entity_type}|):\verb| \n|\\
\verb|        ---| (relation: \verb|{relation}|) \verb|---> \n|\\
\verb|    |Entity \verb|{entity}| (type: \verb|{entity_type}|):\verb| \n|\\
Question: \verb|{query} \n|
}
\end{spacing}
\end{tcolorbox}
We provide the detailed in-context prompter configurations and case studies in Appendix~\ref{app:exp_detail_setting}.
\newpage
\section{Experiments}
\label{sec:experiment}

In this section, we conduct extensive experiments to evaluate \modelname\footnote{Codes will be made public after the paper is formally published.} on the following research questions (\textit{\textbf{RQ}s}):
\begin{itemize}[leftmargin=*, itemsep=0pt, parsep=0pt]
    \item \textbf{\textit{RQ1:}} Can \modelname~improve retrieval performance?
    \item \textbf{\textit{RQ2:}} How does \modelname~perform in QA task?
    \item \textbf{\textit{RQ3:}} Can \modelname~generalize across different domains?
    \item \textbf{\textit{RQ4:}} Is \modelname~more time-efficient and effective?
    \item \textbf{\textit{RQ5:}} How do the subgraphs support multi-hop QA reasoning?
    \item \textbf{\textit{RQ6:}} How does each component and hyper-parameter affect the overall performance of \modelname?
\end{itemize}


\subsection{Experiment Settings}
We introduce the brief experiment settings of datasets, baselines, and metrics. Detailed settings are provided in Appendix~\ref{app:exp_detail}.

\textbf{Datasets.}
Main experiments are conducted on the \textbf{three} representative multi-hop question answering benchmarks: \texttt{HotpotQA}~\cite{yang2018hotpotqa}, \texttt{MuSiQue}~\cite{trivedi2022musique}, and \texttt{2WikiMultiHopQA}~\cite{ho2020constructing}, which are widely adopted to assess multi-hop retrieving and reasoning capabilities, covering diverse domains and complex inference patterns.
To further evaluate the generalizability of \modelname~across different domains, we further incorporate \textbf{seven} benchmarks, including \texttt{PubMedQA}~\cite{jin2019pubmedqa}, \texttt{DelucionQA}~\cite{sadat2023delucionqa}, \texttt{TechQA}~\cite{castelli2020techqa}, \texttt{ExpertQA}~\cite{malaviya2024expertqa}, \texttt{EManual}~\cite{nandy2021question}, \texttt{MS Marco}~\cite{nguyen2016ms}, and \texttt{HAGRID}~\cite{kamalloo2023hagrid}.
Statistics in Table~\ref{tab:dataset_general_qa} and Table~\ref{tab:dataset_domain_qa}.

\textbf{Baselines.}
We compare \modelname~against \textbf{18} state-of-the-art baselines (and their combinations), categorized into \textbf{four} groups:
\begin{itemize}[leftmargin=*]
    \item \textbf{\textit{Base LLM:}} \texttt{GPT-4o-mini}~\cite{hurst2024gpt}.
    \item \textbf{\textit{Single-step RAGs:}} including \texttt{BM25} \cite{robertson1994some}, \texttt{Contriever} \cite{izacard2022unsupervised}, \texttt{GTR} \cite{ni2022large}, \texttt{ColBERTv2} \cite{santhanam2022colbertv2}, \texttt{RAPTOR} \cite{sarthi2024raptor}, and \texttt{Proposition} \cite{chen2024dense}.
    \item \textbf{\textit{Graph-enhanced RAGs:}} including \texttt{GraphRAG} \cite{edge2024local}, \texttt{G-Retriever} \cite{he2024g}, \texttt{LightRAG} \cite{guo2024lightrag}, \texttt{HippoRAG} \cite{jimenez2024hipporag}, \texttt{HippoRAG} \texttt{2} \cite{gutierrez2025rag}, \texttt{SubgraphRAG} \cite{li2025simple}, \texttt{PropRAG} \cite{wang2025proprag}, and the most relevant \texttt{GFM-RAG} \cite{luo2025gfm}.
    \item \textbf{\textit{Multi-step RAGs:}} \texttt{IRCoT} \cite{trivedi2023interleaving}, \texttt{FLARE} \cite{jiang2023active}, \texttt{Adaptive-RAG} \cite{jeong2024adaptive}.
\end{itemize}
Notably, \texttt{IRCoT}~\cite{trivedi2023interleaving} serves as a general multi-step reasoning wrapper that can be combined with non-iterative retrievers, enabling single-step RAG and graph-based methods to perform multi-hop reasoning via interleaved retrieval and generation. We make a thorough comparison for all baselines and \modelname~in Table~\ref{tab:compare}.

\textbf{Metrics.}
To assess the retrieval quality, we report Recall@2 and Recall@5 for both the retrieved entities and documents (R@2/5$_\textsf{E}$, R@2/5$_\textsf{D}$).
In addition, we report Mean Reciprocal Rank (MRR) in selected retrieval experiments to further reflect ranking quality.
For end-to-end QA, we adopt standard metrics, including Exact Match (EM), F1 score, Precision (P), and Recall (R) in main experiments, which thoroughly analyze the answer coverage and correctness~\cite{jimenez2024hipporag}. We additionally report ROUGE-L in some cases to assess semantic overlap between generated answers and references.

\textbf{Implementation Details.}
The retriever backbone is a 6-layer query-conditioned GNN with a hidden dimension of 512, initialized using \texttt{all-mpnet-v2} sentence encoder~\cite{allmpnet2021sbert}. Pre-training is conducted on a large-scale corpus of 700K documents, from which 60 domain-specific knowledge graphs comprising over 14M triples are constructed~\cite{luo2025gfm}. The model is trained on 8 NVIDIA A100 GPUs (80GB) with a batch size of 4. Additional hyper-parameters and implementation details are reported in Appendix~\ref{app:imp_detail_ours}.

\subsection{\textit{RQ1:} Entity Retrieval Performance}
\label{rq1:retrieve}
We evaluate entity-level and document-level retrieval performance on standard multi-hop QA benchmarks.

\textbf{Results.}
As shown in Table~\ref{tab:res_retrieve}, \modelname~consistently outperforms all baselines across most metrics, achieving the highest average rank. It ranks first on R@2$_\textsf{E}$ and R@5$_\textsf{E}$ scores, which validate the effectiveness of the subgraph selector in extracting the compact yet informative entity sets, and reveal the cross-domain generalizability that the GFM retriever provides. In contrast, dense retrievers such as \texttt{Contriever} and \texttt{GTR} ignore the inter-entity structure, while graph-based methods like \texttt{GraphRAG} and \texttt{G-Retriever} tend to over-select irrelevant neighbors. \texttt{SubgraphRAG} and \texttt{GFM-RAG} leverage graph structure but lack query-specific filtering. \modelname~improves the aforementioned drawbacks by jointly modeling semantic relevance and structural minimality.

\subsection{\textit{RQ2:} Question Answering Performance}
\label{rq2:qa}
We evaluate the end-to-end multi-hop QA by feeding the retrieved evidence into the same base LLM.

\textbf{Results.}
As shown in Table~\ref{tab:res_qa}, our \modelname~already shows strong single-step QA capability, outperforming most single-step baselines and surpassing several multi-step methods, which indicates retrieving a minimal yet sufficient subgraph provides informative evidence for multi-hop answering even without iterative retrieval. When augmented with \texttt{IRCoT}, \modelname~achieves the best overall performance with consistent gains across datasets and metrics, suggesting that multi-step reasoning benefits substantially from structured evidence. The comparison between the \modelname~and \modelname+\texttt{IRCoT} further reveals that subgraph complements iterative reasoning: it provides a compact reasoning scaffold, while \texttt{IRCoT} refines evidence selection step by step, leading to higher answer coverage (Recall) without sacrificing correctness (Precision). Overall, the results confirm that our subgraph-centric retrieval is effective as a standalone single-step QA method and serves as a strong substrate for multi-step reasoning frameworks.

\subsection{\textit{RQ3:} Cross-domain Generalizability}
\label{rq3:domain}
We evaluate cross-domain generalization on \textbf{seven} QA datasets spanning biomedical, customer support, and general knowledge domains, all disjoint from pre-training. \modelname~is compared with \textbf{three} strong baselines identified in \textit{\textbf{RQ1}} and \textit{\textbf{RQ2}}. We report retrieval performance in R@5$_\textsf{D}$ on all datasets and QA performance in ROUGE-L on three datasets that support answer generation, under a zero-shot setting without domain-specific fine-tuning.

\begin{table*}[t]
\renewcommand{\arraystretch}{0.85}
\setlength{\tabcolsep}{3.8pt}
  \centering
  \caption{Results of retrieval performance. We report Recall (\%) at top-2 and top-5 for both entity-level (R@2$_\textsf{E}$, R@5$_\textsf{E}$) and document-level (R@2$_\textsf{D}$, R@5$_\textsf{D}$). Best results are in \textbf{bold} and runner-ups are \underline{underlined}. The darker the cell, the better.}
  \vspace{-0.3cm}
  \resizebox{\textwidth}{!}{
    \begin{tabular}{lccccccccccccr}
    \toprule
    \textbf{Dataset} & \multicolumn{4}{c}{\texttt{HotpotQA}}  & \multicolumn{4}{c}{\texttt{MuSiQue}}   & \multicolumn{4}{c}{\texttt{2WikiMultiHopQA}} & \multirow{2}[2]{*}{\textbf{Avg. Rank}} \\
    \cmidrule(r{1mm}){1-1}  \cmidrule(l{0.5mm}r{1mm}){2-5} \cmidrule(l{0.5mm}r{1mm}){6-9} \cmidrule(l{0.5mm}r{1mm}){10-13} \textbf{Method} & R@2$_\textsf{E}$ & R@5$_\textsf{E}$ & R@2$_\textsf{D}$ & R@5$_\textsf{D}$ & R@2$_\textsf{E}$ & R@5$_\textsf{E}$ & R@2$_\textsf{D}$ & R@5$_\textsf{D}$ & R@2$_\textsf{E}$ & R@5$_\textsf{E}$ & R@2$_\textsf{D}$ & R@5$_\textsf{D}$ &        \\
    \cmidrule(r{1mm}){1-1}  \cmidrule(l{0.5mm}r{1mm}){2-5} \cmidrule(l{0.5mm}r{1mm}){6-9} \cmidrule(l{0.5mm}r{1mm}){10-13} \cmidrule(l{0.5mm}){14-14}
    \texttt{BM25} \scalebox{0.68}{(\mycite{robertson1994some}~\textit{SIGIR'94})} & \cellcolor{mygray!2.78}{40.\scalebox{0.75}{0}} & \cellcolor{mygray!26.33}{53.\scalebox{0.75}{2}} & \cellcolor{mygray!27.31}{55.\scalebox{0.75}{4}} & \cellcolor{mygray!35.24}{72.\scalebox{0.75}{2}} & \cellcolor{mygray!26.17}{19.\scalebox{0.75}{5}} & \cellcolor{mygray!13.86}{23.\scalebox{0.75}{6}} & \cellcolor{mygray!21.25}{32.\scalebox{0.75}{3}} & \cellcolor{mygray!19.61}{41.\scalebox{0.75}{2}} & \cellcolor{mygray!12.81}{46.\scalebox{0.75}{9}} & \cellcolor{mygray!16.03}{57.\scalebox{0.75}{9}} & \cellcolor{mygray!12.36}{51.\scalebox{0.75}{8}} & \cellcolor{mygray!15.74}{61.\scalebox{0.75}{9}} & \cellcolor{myred!9.59}{18.\scalebox{0.75}{3}} \\
 \texttt{Contriever} \scalebox{0.68}{(\mycite{izacard2022unsupervised}~\textit{TMLR'22})} & \cellcolor{mygray!7.19}{41.\scalebox{0.75}{7}} & \cellcolor{mygray!32.7}{56.\scalebox{0.75}{0}} & \cellcolor{mygray!30.05}{57.\scalebox{0.75}{2}} & \cellcolor{mygray!41.51}{75.\scalebox{0.75}{5}} & \cellcolor{mygray!27.14}{19.\scalebox{0.75}{8}} & \cellcolor{mygray!34.42}{28.\scalebox{0.75}{8}} & \cellcolor{mygray!27.66}{34.\scalebox{0.75}{8}} & \cellcolor{mygray!34.24}{46.\scalebox{0.75}{6}} & \cellcolor{mygray!2}{39.\scalebox{0.75}{7}} & \cellcolor{mygray!5.02}{51.\scalebox{0.75}{7}} & \cellcolor{mygray!4.77}{46.\scalebox{0.75}{6}} & \cellcolor{mygray!8.27}{57.\scalebox{0.75}{5}} & \cellcolor{myred!11.92}{17.\scalebox{0.75}{7}} \\
 \texttt{GTR} \scalebox{0.68}{(\mycite{ni2022large}~\textit{EMNLP'22})} & \cellcolor{mygray!19.65}{46.\scalebox{0.75}{5}} & \cellcolor{mygray!35.2}{57.\scalebox{0.75}{1}} & \cellcolor{mygray!33.41}{59.\scalebox{0.75}{4}} & \cellcolor{mygray!37.33}{73.\scalebox{0.75}{3}} & \cellcolor{mygray!28.75}{20.\scalebox{0.75}{3}} & \cellcolor{mygray!30.07}{27.\scalebox{0.75}{7}} & \cellcolor{mygray!34.33}{37.\scalebox{0.75}{4}} & \cellcolor{mygray!41.01}{49.\scalebox{0.75}{1}} & \cellcolor{mygray!27.52}{56.\scalebox{0.75}{7}} & \cellcolor{mygray!29.87}{65.\scalebox{0.75}{7}} & \cellcolor{mygray!24.62}{60.\scalebox{0.75}{2}} & \cellcolor{mygray!25.91}{67.\scalebox{0.75}{9}} & \cellcolor{myred!24.47}{14.\scalebox{0.75}{1}} \\
 \texttt{ColBERTv2} \scalebox{0.68}{(\mycite{santhanam2022colbertv2}~\textit{NAACL'22})} & \cellcolor{mygray!26.92}{49.\scalebox{0.75}{3}} & \cellcolor{mygray!47.03}{62.\scalebox{0.75}{3}} & \cellcolor{mygray!41.49}{64.\scalebox{0.75}{7}} & \cellcolor{mygray!48.73}{79.\scalebox{0.75}{3}} & \cellcolor{mygray!41.96}{24.\scalebox{0.75}{4}} & \cellcolor{mygray!31.65}{28.\scalebox{0.75}{1}} & \cellcolor{mygray!35.62}{37.\scalebox{0.75}{9}} & \cellcolor{mygray!41.28}{49.\scalebox{0.75}{2}} & \cellcolor{mygray!24.52}{54.\scalebox{0.75}{7}} & \cellcolor{mygray!28.1}{64.\scalebox{0.75}{7}} & \cellcolor{mygray!23.16}{59.\scalebox{0.75}{2}} & \cellcolor{mygray!26.42}{68.\scalebox{0.75}{2}} & \cellcolor{myred!33.23}{11.\scalebox{0.75}{6}} \\
 \texttt{RAPTOR} \scalebox{0.68}{(\mycite{sarthi2024raptor}~\textit{ICLR'24})} & \cellcolor{mygray!14.72}{44.\scalebox{0.75}{6}} & \cellcolor{mygray!34.29}{56.\scalebox{0.75}{7}} & \cellcolor{mygray!31.43}{58.\scalebox{0.75}{1}} & \cellcolor{mygray!33.34}{71.\scalebox{0.75}{2}} & \cellcolor{mygray!35.19}{22.\scalebox{0.75}{3}} & \cellcolor{mygray!33.63}{28.\scalebox{0.75}{6}} & \cellcolor{mygray!29.97}{35.\scalebox{0.75}{7}} & \cellcolor{mygray!30.72}{45.\scalebox{0.75}{3}} & \cellcolor{mygray!3.95}{41.\scalebox{0.75}{0}} & \cellcolor{mygray!2}{50.\scalebox{0.75}{0}} & \cellcolor{mygray!4.33}{46.\scalebox{0.75}{3}} & \cellcolor{mygray!2}{53.\scalebox{0.75}{8}} & \cellcolor{myred!13.67}{17.\scalebox{0.75}{2}} \\
 \texttt{Proposition} \scalebox{0.68}{(\mycite{chen2024dense}~\textit{EMNLP'24})} & \cellcolor{mygray!17.05}{45.\scalebox{0.75}{5}} & \cellcolor{mygray!33.61}{56.\scalebox{0.75}{4}} & \cellcolor{mygray!32.34}{58.\scalebox{0.75}{7}} & \cellcolor{mygray!33.15}{71.\scalebox{0.75}{1}} & \cellcolor{mygray!34.87}{22.\scalebox{0.75}{2}} & \cellcolor{mygray!39.56}{30.\scalebox{0.75}{1}} & \cellcolor{mygray!34.85}{37.\scalebox{0.75}{6}} & \cellcolor{mygray!41.55}{49.\scalebox{0.75}{3}} & \cellcolor{mygray!18.51}{50.\scalebox{0.75}{7}} & \cellcolor{mygray!17.45}{58.\scalebox{0.75}{7}} & \cellcolor{mygray!19.07}{56.\scalebox{0.75}{4}} & \cellcolor{mygray!17.77}{63.\scalebox{0.75}{1}} & \cellcolor{myred!22.72}{14.\scalebox{0.75}{6}} \\
    \cmidrule(r{1mm}){1-1}  \cmidrule(l{0.5mm}r{1mm}){2-5} \cmidrule(l{0.5mm}r{1mm}){6-9} \cmidrule(l{0.5mm}r{1mm}){10-13} \cmidrule(l{0.5mm}){14-14}
    \texttt{GraphRAG} \scalebox{0.68}{(\mycite{edge2024local}~\textit{arXiv'24})} & \cellcolor{mygray!19.91}{46.\scalebox{0.75}{6}} & \cellcolor{mygray!43.16}{60.\scalebox{0.75}{6}} & \cellcolor{mygray!31.73}{58.\scalebox{0.75}{3}} & \cellcolor{mygray!43.6}{76.\scalebox{0.75}{6}} & \cellcolor{mygray!31.33}{21.\scalebox{0.75}{1}} & \cellcolor{mygray!39.56}{30.\scalebox{0.75}{1}} & \cellcolor{mygray!29.2}{35.\scalebox{0.75}{4}} & \cellcolor{mygray!41.55}{49.\scalebox{0.75}{3}} & \cellcolor{mygray!27.97}{57.\scalebox{0.75}{0}} & \cellcolor{mygray!41.06}{72.\scalebox{0.75}{0}} & \cellcolor{mygray!26.66}{61.\scalebox{0.75}{6}} & \cellcolor{mygray!41.85}{77.\scalebox{0.75}{3}} & \cellcolor{myred!32.94}{11.\scalebox{0.75}{7}} \\
 \texttt{G-Retriever} \scalebox{0.68}{(\mycite{he2024g}~\textit{NeurIPS'24})} & \cellcolor{mygray!10.31}{42.\scalebox{0.75}{9}} & \cellcolor{mygray!24.97}{52.\scalebox{0.75}{6}} & \cellcolor{mygray!24.11}{53.\scalebox{0.75}{3}} & \cellcolor{mygray!22.51}{65.\scalebox{0.75}{5}} & \cellcolor{mygray!32.94}{21.\scalebox{0.75}{6}} & \cellcolor{mygray!34.42}{28.\scalebox{0.75}{8}} & \cellcolor{mygray!37.92}{38.\scalebox{0.75}{8}} & \cellcolor{mygray!30.18}{45.\scalebox{0.75}{1}} & \cellcolor{mygray!23.92}{54.\scalebox{0.75}{3}} & \cellcolor{mygray!23.84}{62.\scalebox{0.75}{3}} & \cellcolor{mygray!25.49}{60.\scalebox{0.75}{8}} & \cellcolor{mygray!25.74}{67.\scalebox{0.75}{8}} & \cellcolor{myred!19.22}{15.\scalebox{0.75}{6}} \\
 \texttt{LightRAG} \scalebox{0.68}{(\mycite{guo2024lightrag}~\textit{arXiv'24})} & \cellcolor{mygray!2}{39.\scalebox{0.75}{7}} & \cellcolor{mygray!2}{42.\scalebox{0.75}{5}} & \cellcolor{mygray!2}{38.\scalebox{0.75}{8}} & \cellcolor{mygray!2}{54.\scalebox{0.75}{7}} & \cellcolor{mygray!2}{12.\scalebox{0.75}{0}} & \cellcolor{mygray!2}{20.\scalebox{0.75}{6}} & \cellcolor{mygray!2}{24.\scalebox{0.75}{8}} & \cellcolor{mygray!2}{34.\scalebox{0.75}{7}} & \cellcolor{mygray!3.8}{40.\scalebox{0.75}{9}} & \cellcolor{mygray!10.7}{54.\scalebox{0.75}{9}} & \cellcolor{mygray!2.58}{45.\scalebox{0.75}{1}} & \cellcolor{mygray!10.99}{59.\scalebox{0.75}{1}} & \cellcolor{myred!2}{20.\scalebox{0.75}{5}} \\
 \texttt{HippoRAG} \scalebox{0.68}{(\mycite{jimenez2024hipporag}~\textit{NeurIPS'24})} & \cellcolor{mygray!21.73}{47.\scalebox{0.75}{3}} & \cellcolor{mygray!43.62}{60.\scalebox{0.75}{8}} & \cellcolor{mygray!34.48}{60.\scalebox{0.75}{1}} & \cellcolor{mygray!47.21}{78.\scalebox{0.75}{5}} & \cellcolor{mygray!46.15}{25.\scalebox{0.75}{7}} & \cellcolor{mygray!53}{33.\scalebox{0.75}{5}} & \cellcolor{mygray!44.08}{41.\scalebox{0.75}{2}} & \cellcolor{mygray!52.12}{53.\scalebox{0.75}{2}} & \cellcolor{mygray!40.13}{65.\scalebox{0.75}{1}} & \cellcolor{mygray!60.77}{83.\scalebox{0.75}{1}} & \cellcolor{mygray!36.58}{68.\scalebox{0.75}{4}} & \cellcolor{mygray!58.3}{87.\scalebox{0.75}{0}} & \cellcolor{myred!43.73}{8.\scalebox{0.75}{6}} \\
 \texttt{HippoRAG} \texttt{2} \scalebox{0.68}{(\mycite{gutierrez2025rag}~\textit{ICML'25})} & \cellcolor{mygray!63.51}{\underline{~63.\scalebox{0.75}{4}~}} & \cellcolor{mygray!64.54}{70.\scalebox{0.75}{0}} & \cellcolor{mygray!65.58}{80.\scalebox{0.75}{5}} & \cellcolor{mygray!65.44}{88.\scalebox{0.75}{1}} & \cellcolor{mygray!60.98}{30.\scalebox{0.75}{3}} & \cellcolor{mygray!62.49}{35.\scalebox{0.75}{9}} & \cellcolor{mygray!58.97}{47.\scalebox{0.75}{0}} & \cellcolor{mygray!61.6}{56.\scalebox{0.75}{7}} & \cellcolor{mygray!67.9}{\underline{~83.\scalebox{0.75}{6}~}} & \cellcolor{mygray!66.8}{86.\scalebox{0.75}{5}} & \cellcolor{mygray!66.5}{\underline{~88.\scalebox{0.75}{9}~}} & \cellcolor{mygray!63.56}{90.\scalebox{0.75}{1}} & \cellcolor{myred!63}{3.\scalebox{0.75}{1}} \\
 \texttt{SubgraphRAG} \scalebox{0.68}{(\mycite{li2025simple}~\textit{ICLR'25})} & \cellcolor{mygray!24.06}{48.\scalebox{0.75}{2}} & \cellcolor{mygray!33.84}{56.\scalebox{0.75}{5}} & \cellcolor{mygray!36.61}{61.\scalebox{0.75}{5}} & \cellcolor{mygray!36.76}{73.\scalebox{0.75}{0}} & \cellcolor{mygray!51.95}{27.\scalebox{0.75}{5}} & \cellcolor{mygray!47.07}{32.\scalebox{0.75}{0}} & \cellcolor{mygray!46.39}{42.\scalebox{0.75}{1}} & \cellcolor{mygray!41.55}{49.\scalebox{0.75}{3}} & \cellcolor{mygray!38.93}{64.\scalebox{0.75}{3}} & \cellcolor{mygray!55.8}{80.\scalebox{0.75}{3}} & \cellcolor{mygray!39.94}{70.\scalebox{0.75}{7}} & \cellcolor{mygray!55.76}{85.\scalebox{0.75}{5}} & \cellcolor{myred!41.4}{9.\scalebox{0.75}{3}} \\
 \texttt{PropRAG} \scalebox{0.68}{(\mycite{wang2025proprag}~\textit{EMNLP'25})} & \cellcolor{mygray!63.51}{\underline{~63.\scalebox{0.75}{4}~}} & \cellcolor{mygray!66.13}{70.\scalebox{0.75}{7}} & \cellcolor{mygray!67.71}{\underline{~81.\scalebox{0.75}{9}~}} & \cellcolor{mygray!65.25}{88.\scalebox{0.75}{0}} & \cellcolor{mygray!62.91}{\underline{~30.\scalebox{0.75}{9}~}} & \cellcolor{mygray!66.44}{\underline{~36.\scalebox{0.75}{9}~}} & \cellcolor{mygray!60.76}{\underline{~47.\scalebox{0.75}{7}~}} & \cellcolor{mygray!64.85}{\underline{~57.\scalebox{0.75}{9}~}} & \cellcolor{mygray!65.5}{82.\scalebox{0.75}{0}} & \cellcolor{mygray!65.92}{86.\scalebox{0.75}{0}} & \cellcolor{mygray!65.04}{87.\scalebox{0.75}{9}} & \cellcolor{mygray!63.56}{90.\scalebox{0.75}{1}} & \cellcolor{myred!63.87}{\underline{2.\scalebox{0.75}{8}}} \\
 \texttt{GFM-RAG} \scalebox{0.68}{(\mycite{luo2025gfm}~\textit{NeurIPS'25})} & \cellcolor{mygray!56.5}{60.\scalebox{0.75}{7}} & \cellcolor{mygray!66.59}{\underline{~70.\scalebox{0.75}{9}~}} & \cellcolor{mygray!58.11}{75.\scalebox{0.75}{6}} & \cellcolor{mygray!68.29}{\underline{~89.\scalebox{0.75}{6}~}} & \cellcolor{mygray!59.36}{29.\scalebox{0.75}{8}} & \cellcolor{mygray!58.93}{35.\scalebox{0.75}{0}} & \cellcolor{mygray!49.98}{43.\scalebox{0.75}{5}} & \cellcolor{mygray!64.04}{57.\scalebox{0.75}{6}} & \cellcolor{mygray!51.39}{72.\scalebox{0.75}{6}} & \cellcolor{mygray!68.76}{\underline{~87.\scalebox{0.75}{6}~}} & \cellcolor{mygray!52.2}{79.\scalebox{0.75}{1}} & \cellcolor{mygray!67.46}{92.\scalebox{0.75}{4}} & \cellcolor{myred!60.95}{3.\scalebox{0.75}{7}} \\
    \cmidrule(r{1mm}){1-1}  \cmidrule(l{0.5mm}r{1mm}){2-5} \cmidrule(l{0.5mm}r{1mm}){6-9} \cmidrule(l{0.5mm}r{1mm}){10-13} \cmidrule(l{0.5mm}){14-14}
    \texttt{FLARE} \scalebox{0.68}{(\mycite{jiang2023active}~\textit{EMNLP'23})} & \cellcolor{mygray!45.34}{56.\scalebox{0.75}{4}} & \cellcolor{mygray!53.17}{65.\scalebox{0.75}{0}} & \cellcolor{mygray!54.3}{73.\scalebox{0.75}{1}} & \cellcolor{mygray!52.53}{81.\scalebox{0.75}{3}} & \cellcolor{mygray!58.72}{29.\scalebox{0.75}{6}} & \cellcolor{mygray!62.49}{35.\scalebox{0.75}{9}} & \cellcolor{mygray!52.04}{44.\scalebox{0.75}{3}} & \cellcolor{mygray!57.27}{55.\scalebox{0.75}{1}} & \cellcolor{mygray!36.98}{63.\scalebox{0.75}{0}} & \cellcolor{mygray!35.73}{69.\scalebox{0.75}{0}} & \cellcolor{mygray!34.69}{67.\scalebox{0.75}{1}} & \cellcolor{mygray!34.73}{73.\scalebox{0.75}{1}} & \cellcolor{myred!20.39}{6.\scalebox{0.75}{8}} \\
 \texttt{Adaptive-RAG} \scalebox{0.68}{(\mycite{jeong2024adaptive}~\textit{NAACL'24})} & \cellcolor{mygray!23.28}{47.\scalebox{0.75}{9}} & \cellcolor{mygray!44.07}{61.\scalebox{0.75}{0}} & \cellcolor{mygray!35.85}{61.\scalebox{0.75}{0}} & \cellcolor{mygray!43.22}{76.\scalebox{0.75}{4}} & \cellcolor{mygray!35.84}{22.\scalebox{0.75}{5}} & \cellcolor{mygray!29.28}{27.\scalebox{0.75}{5}} & \cellcolor{mygray!28.43}{35.\scalebox{0.75}{1}} & \cellcolor{mygray!29.09}{44.\scalebox{0.75}{7}} & \cellcolor{mygray!5}{41.\scalebox{0.75}{7}} & \cellcolor{mygray!17.45}{58.\scalebox{0.75}{7}} & \cellcolor{mygray!2}{44.\scalebox{0.75}{7}} & \cellcolor{mygray!14.89}{61.\scalebox{0.75}{4}} & \cellcolor{myred!50.15}{15.\scalebox{0.75}{3}} \\
 \texttt{BM25} + \texttt{IRCoT} \scalebox{0.68}{(\mycite{trivedi2023interleaving}~\textit{ACL'23})} & \cellcolor{mygray!20.69}{46.\scalebox{0.75}{9}} & \cellcolor{mygray!40.66}{59.\scalebox{0.75}{5}} & \cellcolor{mygray!42.86}{65.\scalebox{0.75}{6}} & \cellcolor{mygray!48.16}{79.\scalebox{0.75}{0}} & \cellcolor{mygray!25.2}{19.\scalebox{0.75}{2}} & \cellcolor{mygray!33.63}{28.\scalebox{0.75}{6}} & \cellcolor{mygray!26.12}{34.\scalebox{0.75}{2}} & \cellcolor{mygray!29.09}{44.\scalebox{0.75}{7}} & \cellcolor{mygray!23.17}{53.\scalebox{0.75}{8}} & \cellcolor{mygray!35.38}{68.\scalebox{0.75}{8}} & \cellcolor{mygray!26.08}{61.\scalebox{0.75}{2}} & \cellcolor{mygray!38.97}{75.\scalebox{0.75}{6}} & \cellcolor{myred!27.1}{13.\scalebox{0.75}{3}} \\
 \texttt{Contriever} + \texttt{IRCoT} \scalebox{0.68}{(\mycite{trivedi2023interleaving}~\textit{ACL'23})} & \cellcolor{mygray!17.83}{45.\scalebox{0.75}{8}} & \cellcolor{mygray!50.21}{63.\scalebox{0.75}{7}} & \cellcolor{mygray!43.32}{65.\scalebox{0.75}{9}} & \cellcolor{mygray!53.09}{81.\scalebox{0.75}{6}} & \cellcolor{mygray!40.35}{23.\scalebox{0.75}{9}} & \cellcolor{mygray!37.58}{29.\scalebox{0.75}{6}} & \cellcolor{mygray!38.69}{39.\scalebox{0.75}{1}} & \cellcolor{mygray!49.41}{52.\scalebox{0.75}{2}} & \cellcolor{mygray!9.66}{44.\scalebox{0.75}{8}} & \cellcolor{mygray!14.07}{56.\scalebox{0.75}{8}} & \cellcolor{mygray!12.07}{51.\scalebox{0.75}{6}} & \cellcolor{mygray!18.96}{63.\scalebox{0.75}{8}} & \cellcolor{myred!30.89}{12.\scalebox{0.75}{3}} \\
 \texttt{ColBERTv2} + \texttt{IRCoT} \scalebox{0.68}{(\mycite{trivedi2023interleaving}~\textit{ACL'23})} & \cellcolor{mygray!34.96}{52.\scalebox{0.75}{4}} & \cellcolor{mygray!56.81}{66.\scalebox{0.75}{6}} & \cellcolor{mygray!46.37}{67.\scalebox{0.75}{9}} & \cellcolor{mygray!53.85}{82.\scalebox{0.75}{0}} & \cellcolor{mygray!43.25}{24.\scalebox{0.75}{8}} & \cellcolor{mygray!50.63}{32.\scalebox{0.75}{9}} & \cellcolor{mygray!45.37}{41.\scalebox{0.75}{7}} & \cellcolor{mygray!53.47}{53.\scalebox{0.75}{7}} & \cellcolor{mygray!27.82}{56.\scalebox{0.75}{9}} & \cellcolor{mygray!33.78}{67.\scalebox{0.75}{9}} & \cellcolor{mygray!30.31}{64.\scalebox{0.75}{1}} & \cellcolor{mygray!36.93}{74.\scalebox{0.75}{4}} & \cellcolor{myred!45.48}{8.\scalebox{0.75}{1}} \\
 \texttt{HippoRAG} + \texttt{IRCoT} \scalebox{0.68}{(\mycite{trivedi2023interleaving}~\textit{ACL'23})} & \cellcolor{mygray!46.64}{56.\scalebox{0.75}{9}} & \cellcolor{mygray!58.4}{67.\scalebox{0.75}{3}} & \cellcolor{mygray!45}{67.\scalebox{0.75}{0}} & \cellcolor{mygray!55.75}{83.\scalebox{0.75}{0}} & \cellcolor{mygray!50.99}{27.\scalebox{0.75}{2}} & \cellcolor{mygray!61.7}{35.\scalebox{0.75}{7}} & \cellcolor{mygray!54.6}{45.\scalebox{0.75}{3}} & \cellcolor{mygray!64.04}{57.\scalebox{0.75}{6}} & \cellcolor{mygray!46.43}{69.\scalebox{0.75}{3}} & \cellcolor{mygray!66.45}{86.\scalebox{0.75}{3}} & \cellcolor{mygray!47.38}{75.\scalebox{0.75}{8}} & \cellcolor{mygray!70}{\textbf{93.\scalebox{0.75}{9}}} & \cellcolor{myred!57.45}{4.\scalebox{0.75}{7}} \\
    \cmidrule(r{1mm}){1-1}  \cmidrule(l{0.5mm}r{1mm}){2-5} \cmidrule(l{0.5mm}r{1mm}){6-9} \cmidrule(l{0.5mm}r{1mm}){10-13} \cmidrule(l{0.5mm}){14-14}
    \textbf{\modelname} (ours) & \cellcolor{mygray!70}{\textbf{65.\scalebox{0.75}{9}}} & \cellcolor{mygray!70}{\textbf{72.\scalebox{0.75}{4}}} & \cellcolor{mygray!70}{\textbf{83.\scalebox{0.75}{4}}} & \cellcolor{mygray!70}{\textbf{90.\scalebox{0.75}{5}}} & \cellcolor{mygray!70}{\textbf{33.\scalebox{0.75}{1}}} & \cellcolor{mygray!70}{\textbf{37.\scalebox{0.75}{8}}} & \cellcolor{mygray!70}{\textbf{51.\scalebox{0.75}{3}}} & \cellcolor{mygray!70}{\textbf{59.\scalebox{0.75}{8}}} & \cellcolor{mygray!70}{\textbf{85.\scalebox{0.75}{0}}} & \cellcolor{mygray!70}{\textbf{88.\scalebox{0.75}{3}}} & \cellcolor{mygray!70}{\textbf{91.\scalebox{0.75}{3}}} & \cellcolor{mygray!69.83}{\underline{~93.\scalebox{0.75}{8}~}} & \cellcolor{myred!70}{\textbf{1.\scalebox{0.75}{1}}} \\
    \bottomrule
    \end{tabular}%
  }
  \label{tab:res_retrieve}%
\end{table*}%

\begin{table*}[!t]
\renewcommand{\arraystretch}{0.85}
\setlength{\tabcolsep}{4pt}
  \centering
  \caption{Results of multi-hop question answering (QA) performance. We report Exact Match (EM), F1 score, Precision (P), and Recall (R), all reported as percentages (\%). Best results are in \textbf{bold} and runner-ups are \underline{underlined}. The darker the cell, the better.}
  \vspace{-0.3cm}
  \resizebox{\textwidth}{!}{
    \begin{tabular}{lccccccccccccr}
    \toprule
    \textbf{Dataset} & \multicolumn{4}{c}{\texttt{HotpotQA}}  & \multicolumn{4}{c}{\texttt{MuSiQue}}   & \multicolumn{4}{c}{\texttt{2WikiMultiHopQA}} & \multirow{2}[2]{*}{\textbf{Avg. Rank}} \\
\cmidrule(r{1mm}){1-1}  \cmidrule(l{0.5mm}r{1mm}){2-5} \cmidrule(l{0.5mm}r{1mm}){6-9} \cmidrule(l{0.5mm}r{1mm}){10-13} \textbf{Method} & ~~~EM~~ & ~~F1~~ & ~~P~~ & ~~R~~ & ~~~EM~~ & ~~F1~~ & ~~P~~ & ~~R~~ & ~~~EM~~ & ~~F1~~ & ~~P~~ & ~~R~~ & \\
    \cmidrule(r{1mm}){1-1}  \cmidrule(l{0.5mm}r{1mm}){2-5} \cmidrule(l{0.5mm}r{1mm}){6-9} \cmidrule(l{0.5mm}r{1mm}){10-13} \cmidrule(l{0.5mm}){14-14}
    \texttt{GPT-4o-mini} \scalebox{0.68}{(\mycite{robertson1994some}~\textit{arXiv'24})} & \cellcolor{mygray!2}{40.\scalebox{0.75}{0}} & \cellcolor{mygray!2}{53.\scalebox{0.75}{2}} & \cellcolor{mygray!2}{55.\scalebox{0.75}{4}} & \cellcolor{mygray!2}{72.\scalebox{0.75}{2}} & \cellcolor{mygray!2}{19.\scalebox{0.75}{5}} & \cellcolor{mygray!8.16}{23.\scalebox{0.75}{6}} & \cellcolor{mygray!8.3}{32.\scalebox{0.75}{3}} & \cellcolor{mygray!8.23}{41.\scalebox{0.75}{2}} & \cellcolor{mygray!10.36}{46.\scalebox{0.75}{9}} & \cellcolor{mygray!9.39}{57.\scalebox{0.75}{9}} & \cellcolor{mygray!9.81}{51.\scalebox{0.75}{8}} & \cellcolor{mygray!9.44}{61.\scalebox{0.75}{9}} & \cellcolor{myred!2}{19.\scalebox{0.75}{3}} \\
    \cmidrule(r{1mm}){1-1}  \cmidrule(l{0.5mm}r{1mm}){2-5} \cmidrule(l{0.5mm}r{1mm}){6-9} \cmidrule(l{0.5mm}r{1mm}){10-13} \cmidrule(l{0.5mm}){14-14}
    \texttt{Contriever} \scalebox{0.68}{(\mycite{izacard2022unsupervised}~\textit{TMLR'22})} & \cellcolor{mygray!12.17}{34.\scalebox{0.75}{9}} & \cellcolor{mygray!20.98}{51.\scalebox{0.75}{2}} & \cellcolor{mygray!19.95}{51.\scalebox{0.75}{8}} & \cellcolor{mygray!20.93}{52.\scalebox{0.75}{2}} & \cellcolor{mygray!10.53}{16.\scalebox{0.75}{3}} & \cellcolor{mygray!2}{21.\scalebox{0.75}{6}} & \cellcolor{mygray!2}{21.\scalebox{0.75}{9}} & \cellcolor{mygray!2}{22.\scalebox{0.75}{0}} & \cellcolor{mygray!2}{24.\scalebox{0.75}{3}} & \cellcolor{mygray!2}{33.\scalebox{0.75}{9}} & \cellcolor{mygray!2}{34.\scalebox{0.75}{1}} & \cellcolor{mygray!2}{34.\scalebox{0.75}{8}} & \cellcolor{myred!5.58}{18.\scalebox{0.75}{5}} \\
 \texttt{GTR} \scalebox{0.68}{(\mycite{ni2022large}~\textit{EMNLP'22})} & \cellcolor{mygray!9.68}{33.\scalebox{0.75}{8}} & \cellcolor{mygray!22.56}{51.\scalebox{0.75}{9}} & \cellcolor{mygray!21.43}{52.\scalebox{0.75}{5}} & \cellcolor{mygray!22.47}{52.\scalebox{0.75}{9}} & \cellcolor{mygray!7.83}{15.\scalebox{0.75}{1}} & \cellcolor{mygray!10.87}{25.\scalebox{0.75}{2}} & \cellcolor{mygray!11.07}{25.\scalebox{0.75}{5}} & \cellcolor{mygray!10.86}{25.\scalebox{0.75}{7}} & \cellcolor{mygray!13.73}{33.\scalebox{0.75}{7}} & \cellcolor{mygray!14.47}{42.\scalebox{0.75}{5}} & \cellcolor{mygray!15.06}{42.\scalebox{0.75}{8}} & \cellcolor{mygray!14.6}{43.\scalebox{0.75}{6}} & \cellcolor{myred!9.16}{16.\scalebox{0.75}{9}} \\
 \texttt{ColBERTv2} \scalebox{0.68}{(\mycite{santhanam2022colbertv2}~\textit{NAACL'22})} & \cellcolor{mygray!31.37}{43.\scalebox{0.75}{4}} & \cellcolor{mygray!35.66}{57.\scalebox{0.75}{7}} & \cellcolor{mygray!33.89}{58.\scalebox{0.75}{4}} & \cellcolor{mygray!35.45}{58.\scalebox{0.75}{8}} & \cellcolor{mygray!8.73}{15.\scalebox{0.75}{5}} & \cellcolor{mygray!13.83}{26.\scalebox{0.75}{4}} & \cellcolor{mygray!14.09}{26.\scalebox{0.75}{7}} & \cellcolor{mygray!13.97}{27.\scalebox{0.75}{0}} & \cellcolor{mygray!13.35}{33.\scalebox{0.75}{4}} & \cellcolor{mygray!15.63}{43.\scalebox{0.75}{3}} & \cellcolor{mygray!16.26}{43.\scalebox{0.75}{6}} & \cellcolor{mygray!15.6}{44.\scalebox{0.75}{3}} & \cellcolor{myred!19.89}{14.\scalebox{0.75}{8}} \\
 \texttt{RAPTOR} \scalebox{0.68}{(\mycite{sarthi2024raptor}~\textit{ICLR'24})} & \cellcolor{mygray!42.21}{48.\scalebox{0.75}{2}} & \cellcolor{mygray!39.05}{59.\scalebox{0.75}{2}} & \cellcolor{mygray!37.06}{59.\scalebox{0.75}{9}} & \cellcolor{mygray!38.75}{60.\scalebox{0.75}{3}} & \cellcolor{mygray!13.45}{17.\scalebox{0.75}{6}} & \cellcolor{mygray!19.99}{28.\scalebox{0.75}{9}} & \cellcolor{mygray!20.39}{29.\scalebox{0.75}{2}} & \cellcolor{mygray!19.96}{29.\scalebox{0.75}{5}} & \cellcolor{mygray!9.86}{30.\scalebox{0.75}{6}} & \cellcolor{mygray!13.74}{42.\scalebox{0.75}{0}} & \cellcolor{mygray!14.31}{42.\scalebox{0.75}{3}} & \cellcolor{mygray!13.88}{43.\scalebox{0.75}{1}} & \cellcolor{myred!27.05}{13.\scalebox{0.75}{3}} \\
    \cmidrule(r{1mm}){1-1}  \cmidrule(l{0.5mm}r{1mm}){2-5} \cmidrule(l{0.5mm}r{1mm}){6-9} \cmidrule(l{0.5mm}r{1mm}){10-13} \cmidrule(l{0.5mm}){14-14}
    \texttt{GraphRAG} \scalebox{0.68}{(\mycite{edge2024local}~\textit{arXiv'24})} & \cellcolor{mygray!13.07}{35.\scalebox{0.75}{3}} & \cellcolor{mygray!28.66}{54.\scalebox{0.75}{6}} & \cellcolor{mygray!27.34}{55.\scalebox{0.75}{3}} & \cellcolor{mygray!28.41}{55.\scalebox{0.75}{6}} & \cellcolor{mygray!4.02}{13.\scalebox{0.75}{4}} & \cellcolor{mygray!21.46}{29.\scalebox{0.75}{5}} & \cellcolor{mygray!21.9}{29.\scalebox{0.75}{8}} & \cellcolor{mygray!21.39}{30.\scalebox{0.75}{1}} & \cellcolor{mygray!6.99}{28.\scalebox{0.75}{3}} & \cellcolor{mygray!20.85}{46.\scalebox{0.75}{9}} & \cellcolor{mygray!21.66}{47.\scalebox{0.75}{2}} & \cellcolor{mygray!21.04}{48.\scalebox{0.75}{1}} & \cellcolor{myred!16.32}{14.\scalebox{0.75}{3}} \\
 \texttt{G-Retriever} \scalebox{0.68}{(\mycite{he2024g}~\textit{NeurIPS'24})} & \cellcolor{mygray!8.33}{33.\scalebox{0.75}{2}} & \cellcolor{mygray!18.94}{50.\scalebox{0.75}{3}} & \cellcolor{mygray!17.42}{50.\scalebox{0.75}{6}} & \cellcolor{mygray!19.39}{51.\scalebox{0.75}{5}} & \cellcolor{mygray!14.34}{18.\scalebox{0.75}{0}} & \cellcolor{mygray!12.59}{25.\scalebox{0.75}{9}} & \cellcolor{mygray!12.58}{26.\scalebox{0.75}{1}} & \cellcolor{mygray!13.01}{26.\scalebox{0.75}{6}} & \cellcolor{mygray!24.46}{42.\scalebox{0.75}{3}} & \cellcolor{mygray!18.96}{45.\scalebox{0.75}{6}} & \cellcolor{mygray!19.71}{45.\scalebox{0.75}{9}} & \cellcolor{mygray!19.18}{46.\scalebox{0.75}{8}} & \cellcolor{myred!12.74}{15.\scalebox{0.75}{7}} \\
 \texttt{LightRAG} \scalebox{0.68}{(\mycite{guo2024lightrag}~\textit{arXiv'24})} & \cellcolor{mygray!16.46}{36.\scalebox{0.75}{8}} & \cellcolor{mygray!14.43}{48.\scalebox{0.75}{3}} & \cellcolor{mygray!13.83}{48.\scalebox{0.75}{9}} & \cellcolor{mygray!14.32}{49.\scalebox{0.75}{2}} & \cellcolor{mygray!14.57}{18.\scalebox{0.75}{1}} & \cellcolor{mygray!16.54}{27.\scalebox{0.75}{5}} & \cellcolor{mygray!16.86}{27.\scalebox{0.75}{8}} & \cellcolor{mygray!16.61}{28.\scalebox{0.75}{1}} & \cellcolor{mygray!27.95}{45.\scalebox{0.75}{1}} & \cellcolor{mygray!24.62}{49.\scalebox{0.75}{5}} & \cellcolor{mygray!25.57}{49.\scalebox{0.75}{8}} & \cellcolor{mygray!24.91}{50.\scalebox{0.75}{8}} & \cellcolor{myred!23.47}{14.\scalebox{0.75}{4}} \\
 \texttt{HippoRAG} \scalebox{0.68}{(\mycite{jimenez2024hipporag}~\textit{NeurIPS'24})} & \cellcolor{mygray!27.75}{41.\scalebox{0.75}{8}} & \cellcolor{mygray!29.56}{55.\scalebox{0.75}{0}} & \cellcolor{mygray!28.19}{55.\scalebox{0.75}{7}} & \cellcolor{mygray!29.29}{56.\scalebox{0.75}{0}} & \cellcolor{mygray!17.04}{19.\scalebox{0.75}{2}} & \cellcolor{mygray!22.2}{29.\scalebox{0.75}{8}} & \cellcolor{mygray!22.65}{30.\scalebox{0.75}{1}} & \cellcolor{mygray!22.35}{30.\scalebox{0.75}{5}} & \cellcolor{mygray!29.82}{46.\scalebox{0.75}{6}} & \cellcolor{mygray!39.12}{59.\scalebox{0.75}{5}} & \cellcolor{mygray!40.88}{60.\scalebox{0.75}{0}} & \cellcolor{mygray!39.65}{61.\scalebox{0.75}{1}} & \cellcolor{myred!34.21}{11.\scalebox{0.75}{6}} \\
 \texttt{HippoRAG} \texttt{2} \scalebox{0.68}{(\mycite{gutierrez2025rag}~\textit{ICML'25})} & \cellcolor{mygray!62.77}{\underline{~57.\scalebox{0.75}{3}~}} & \cellcolor{mygray!62.54}{69.\scalebox{0.75}{6}} & \cellcolor{mygray!58.6}{70.\scalebox{0.75}{1}} & \cellcolor{mygray!63.4}{71.\scalebox{0.75}{5}} & \cellcolor{mygray!50.03}{33.\scalebox{0.75}{9}} & \cellcolor{mygray!48.57}{40.\scalebox{0.75}{5}} & \cellcolor{mygray!49.6}{40.\scalebox{0.75}{8}} & \cellcolor{mygray!48.93}{41.\scalebox{0.75}{6}} & \cellcolor{mygray!65.76}{\underline{~75.\scalebox{0.75}{4}~}} & \cellcolor{mygray!66.23}{78.\scalebox{0.75}{2}} & \cellcolor{mygray!68.95}{78.\scalebox{0.75}{7}} & \cellcolor{mygray!67.14}{80.\scalebox{0.75}{3}} & \cellcolor{myred!55.68}{3.\scalebox{0.75}{8}} \\
 \texttt{SubgraphRAG} \scalebox{0.68}{(\mycite{li2025simple}~\textit{ICLR'25})} & \cellcolor{mygray!52.15}{52.\scalebox{0.75}{6}} & \cellcolor{mygray!53.96}{65.\scalebox{0.75}{8}} & \cellcolor{mygray!50.36}{66.\scalebox{0.75}{2}} & \cellcolor{mygray!54.6}{67.\scalebox{0.75}{5}} & \cellcolor{mygray!48.01}{33.\scalebox{0.75}{0}} & \cellcolor{mygray!45.36}{39.\scalebox{0.75}{2}} & \cellcolor{mygray!46.33}{39.\scalebox{0.75}{5}} & \cellcolor{mygray!45.58}{40.\scalebox{0.75}{2}} & \cellcolor{mygray!60.27}{71.\scalebox{0.75}{0}} & \cellcolor{mygray!63.33}{76.\scalebox{0.75}{2}} & \cellcolor{mygray!65.95}{76.\scalebox{0.75}{7}} & \cellcolor{mygray!64.27}{78.\scalebox{0.75}{3}} & \cellcolor{myred!48.53}{6.\scalebox{0.75}{7}} \\
 \texttt{PropRAG} \scalebox{0.68}{(\mycite{wang2025proprag}~\textit{EMNLP'25})} & \cellcolor{mygray!62.54}{57.\scalebox{0.75}{2}} & \cellcolor{mygray!63.67}{70.\scalebox{0.75}{1}} & \cellcolor{mygray!57.54}{69.\scalebox{0.75}{6}} & \cellcolor{mygray!62.3}{71.\scalebox{0.75}{0}} & \cellcolor{mygray!61.25}{\underline{~38.\scalebox{0.75}{9}~}} & \cellcolor{mygray!50.54}{41.\scalebox{0.75}{3}} & \cellcolor{mygray!49.1}{40.\scalebox{0.75}{6}} & \cellcolor{mygray!46.06}{40.\scalebox{0.75}{4}} & \cellcolor{mygray!62.76}{73.\scalebox{0.75}{0}} & \cellcolor{mygray!68.7}{\underline{~79.\scalebox{0.75}{9}~}} & \cellcolor{mygray!68.5}{78.\scalebox{0.75}{4}} & \cellcolor{mygray!66.42}{79.\scalebox{0.75}{8}} & \cellcolor{myred!59.26}{4.\scalebox{0.75}{0}} \\
 \texttt{GFM-RAG} \scalebox{0.68}{(\mycite{luo2025gfm}~\textit{NeurIPS'25})} & \cellcolor{mygray!49.89}{51.\scalebox{0.75}{6}} & \cellcolor{mygray!56.45}{66.\scalebox{0.75}{9}} & \cellcolor{mygray!53.53}{67.\scalebox{0.75}{7}} & \cellcolor{mygray!56.14}{68.\scalebox{0.75}{2}} & \cellcolor{mygray!41.72}{30.\scalebox{0.75}{2}} & \cellcolor{mygray!48.32}{40.\scalebox{0.75}{4}} & \cellcolor{mygray!49.6}{40.\scalebox{0.75}{8}} & \cellcolor{mygray!48.21}{41.\scalebox{0.75}{3}} & \cellcolor{mygray!58.77}{69.\scalebox{0.75}{8}} & \cellcolor{mygray!65.51}{77.\scalebox{0.75}{7}} & \cellcolor{mygray!68.2}{78.\scalebox{0.75}{2}} & \cellcolor{mygray!66.42}{79.\scalebox{0.75}{8}} & \cellcolor{myred!52.11}{5.\scalebox{0.75}{8}} \\
    \cmidrule(r{1mm}){1-1}  \cmidrule(l{0.5mm}r{1mm}){2-5} \cmidrule(l{0.5mm}r{1mm}){6-9} \cmidrule(l{0.5mm}r{1mm}){10-13} \cmidrule(l{0.5mm}){14-14}
    \texttt{FLARE} \scalebox{0.68}{(\mycite{jiang2023active}~\textit{EMNLP'23})} & \cellcolor{mygray!43.34}{48.\scalebox{0.75}{7}} & \cellcolor{mygray!42.21}{60.\scalebox{0.75}{6}} & \cellcolor{mygray!40.01}{61.\scalebox{0.75}{3}} & \cellcolor{mygray!42.05}{61.\scalebox{0.75}{8}} & \cellcolor{mygray!10.3}{16.\scalebox{0.75}{2}} & \cellcolor{mygray!18.75}{28.\scalebox{0.75}{4}} & \cellcolor{mygray!19.13}{28.\scalebox{0.75}{7}} & \cellcolor{mygray!19}{29.\scalebox{0.75}{1}} & \cellcolor{mygray!29.95}{46.\scalebox{0.75}{7}} & \cellcolor{mygray!47.67}{65.\scalebox{0.75}{4}} & \cellcolor{mygray!49.58}{65.\scalebox{0.75}{8}} & \cellcolor{mygray!48.38}{67.\scalebox{0.75}{2}} & \cellcolor{myred!41.37}{10.\scalebox{0.75}{0}} \\
 \texttt{Adaptive-RAG} \scalebox{0.68}{(\mycite{jeong2024adaptive}~\textit{NAACL'24})} & \cellcolor{mygray!36.11}{45.\scalebox{0.75}{5}} & \cellcolor{mygray!39.95}{59.\scalebox{0.75}{6}} & \cellcolor{mygray!37.9}{60.\scalebox{0.75}{3}} & \cellcolor{mygray!39.63}{60.\scalebox{0.75}{7}} & \cellcolor{mygray!4.92}{13.\scalebox{0.75}{8}} & \cellcolor{mygray!11.86}{25.\scalebox{0.75}{6}} & \cellcolor{mygray!12.07}{25.\scalebox{0.75}{9}} & \cellcolor{mygray!11.82}{26.\scalebox{0.75}{1}} & \cellcolor{mygray!32.69}{48.\scalebox{0.75}{9}} & \cellcolor{mygray!43.9}{62.\scalebox{0.75}{8}} & \cellcolor{mygray!45.68}{63.\scalebox{0.75}{2}} & \cellcolor{mygray!44.52}{64.\scalebox{0.75}{5}} & \cellcolor{myred!37.79}{11.\scalebox{0.75}{8}} \\
 \texttt{ColBERTv2} + \texttt{IRCoT} \scalebox{0.68}{(\mycite{trivedi2023interleaving}~\textit{ACL'23})} & \cellcolor{mygray!36.11}{45.\scalebox{0.75}{5}} & \cellcolor{mygray!37.24}{58.\scalebox{0.75}{4}} & \cellcolor{mygray!35.37}{59.\scalebox{0.75}{1}} & \cellcolor{mygray!36.99}{59.\scalebox{0.75}{5}} & \cellcolor{mygray!16.81}{19.\scalebox{0.75}{1}} & \cellcolor{mygray!23.93}{30.\scalebox{0.75}{5}} & \cellcolor{mygray!24.41}{30.\scalebox{0.75}{8}} & \cellcolor{mygray!24.03}{31.\scalebox{0.75}{2}} & \cellcolor{mygray!15.85}{35.\scalebox{0.75}{4}} & \cellcolor{mygray!18.24}{45.\scalebox{0.75}{1}} & \cellcolor{mygray!18.96}{45.\scalebox{0.75}{4}} & \cellcolor{mygray!18.46}{46.\scalebox{0.75}{3}} & \cellcolor{myred!30.63}{11.\scalebox{0.75}{9}} \\
 \texttt{HippoRAG} + \texttt{IRCoT} \scalebox{0.68}{(\mycite{trivedi2023interleaving}~\textit{ACL'23})} & \cellcolor{mygray!36.56}{45.\scalebox{0.75}{7}} & \cellcolor{mygray!39.05}{59.\scalebox{0.75}{2}} & \cellcolor{mygray!37.06}{59.\scalebox{0.75}{9}} & \cellcolor{mygray!38.75}{60.\scalebox{0.75}{3}} & \cellcolor{mygray!23.1}{21.\scalebox{0.75}{9}} & \cellcolor{mygray!30.83}{33.\scalebox{0.75}{3}} & \cellcolor{mygray!31.72}{33.\scalebox{0.75}{7}} & \cellcolor{mygray!30.73}{34.\scalebox{0.75}{0}} & \cellcolor{mygray!31.2}{47.\scalebox{0.75}{7}} & \cellcolor{mygray!43.76}{62.\scalebox{0.75}{7}} & \cellcolor{mygray!45.53}{63.\scalebox{0.75}{1}} & \cellcolor{mygray!44.23}{64.\scalebox{0.75}{3}} & \cellcolor{myred!44.95}{9.\scalebox{0.75}{3}} \\
 \texttt{GFM-RAG} + \texttt{IRCoT} \scalebox{0.68}{(\mycite{trivedi2023interleaving}~\textit{ACL'23})} & \cellcolor{mygray!59.83}{56.\scalebox{0.75}{0}} & \cellcolor{mygray!67.51}{\underline{~71.\scalebox{0.75}{8}~}} & \cellcolor{mygray!59.86}{70.\scalebox{0.75}{7}} & \cellcolor{mygray!60.54}{70.\scalebox{0.75}{2}} & \cellcolor{mygray!56.09}{36.\scalebox{0.75}{6}} & \cellcolor{mygray!70}{\textbf{49.\scalebox{0.75}{2}}} & \cellcolor{mygray!66.98}{\underline{~47.\scalebox{0.75}{7}~}} & \cellcolor{mygray!62.58}{\underline{~47.\scalebox{0.75}{3}~}} & \cellcolor{mygray!62.14}{72.\scalebox{0.75}{5}} & \cellcolor{mygray!70}{\textbf{80.\scalebox{0.75}{8}}} & \cellcolor{mygray!66.55}{77.\scalebox{0.75}{1}} & \cellcolor{mygray!63.84}{78.\scalebox{0.75}{0}} & \cellcolor{myred!66.42}{3.\scalebox{0.75}{4}} \\
    \cmidrule(r{1mm}){1-1}  \cmidrule(l{0.5mm}r{1mm}){2-5} \cmidrule(l{0.5mm}r{1mm}){6-9} \cmidrule(l{0.5mm}r{1mm}){10-13} \cmidrule(l{0.5mm}){14-14}
    \textbf{\modelname} (ours) & \cellcolor{mygray!58.48}{55.\scalebox{0.75}{4}} & \cellcolor{mygray!64.35}{70.\scalebox{0.75}{4}} & \cellcolor{mygray!63.66}{\underline{~72.\scalebox{0.75}{5}~}} & \cellcolor{mygray!65.6}{\underline{~72.\scalebox{0.75}{5}~}} & \cellcolor{mygray!52.05}{34.\scalebox{0.75}{8}} & \cellcolor{mygray!51.03}{41.\scalebox{0.75}{5}} & \cellcolor{mygray!53.38}{42.\scalebox{0.75}{3}} & \cellcolor{mygray!51.08}{42.\scalebox{0.75}{5}} & \cellcolor{mygray!64.39}{74.\scalebox{0.75}{3}} & \cellcolor{mygray!67.1}{78.\scalebox{0.75}{8}} & \cellcolor{mygray!69.55}{\underline{~79.\scalebox{0.75}{1}~}} & \cellcolor{mygray!67.28}{\underline{~80.\scalebox{0.75}{4}~}} & \cellcolor{myred!62.84}{\underline{3.\scalebox{0.75}{0}}} \\
 \textbf{\modelname} (ours) + \texttt{IRCoT} \scalebox{0.68}{(\mycite{trivedi2023interleaving}~\textit{ACL'23})} & \cellcolor{mygray!70}{\textbf{60.\scalebox{0.75}{5}}} & \cellcolor{mygray!70}{\textbf{72.\scalebox{0.75}{9}}} & \cellcolor{mygray!70}{\textbf{75.\scalebox{0.75}{5}}} & \cellcolor{mygray!70}{\textbf{74.\scalebox{0.75}{5}}} & \cellcolor{mygray!70}{\textbf{42.\scalebox{0.75}{8}}} & \cellcolor{mygray!68.77}{\underline{~48.\scalebox{0.75}{7}~}} & \cellcolor{mygray!70}{\textbf{48.\scalebox{0.75}{9}}} & \cellcolor{mygray!70}{\textbf{50.\scalebox{0.75}{4}}} & \cellcolor{mygray!70}{\textbf{78.\scalebox{0.75}{8}}} & \cellcolor{mygray!68.41}{79.\scalebox{0.75}{7}} & \cellcolor{mygray!70}{\textbf{79.\scalebox{0.75}{4}}} & \cellcolor{mygray!70}{\textbf{82.\scalebox{0.75}{3}}} & \cellcolor{myred!70}{\textbf{1.\scalebox{0.75}{3}}} \\
    \bottomrule
    \end{tabular}%
  }
  \label{tab:res_qa}%
\end{table*}%

\textbf{Results.}
As shown in Figure~\ref{fig:radar}, \modelname~consistently delivers the best cross-domain performance on both retrieval and QA. It outperforms all baselines on document retrieval across seven datasets, with particularly strong gains on biomedical and customer-support benchmarks such as \texttt{PubMedQA}, \texttt{DelucionQA}, and \texttt{HAGRID}, indicating robust generalization under domain shift. In contrast, \texttt{PropRAG} and \texttt{HippoRAG} \texttt{2} exhibit clear performance drops, reflecting sensitivity to domain-specific indexing. On generative QA, \modelname~also achieves the highest ROUGE-L scores, showing that its cross-domain retrieval advantages directly translate to better answer generation. Additional domain-specific fine-tuning results are provided in Appendix~\ref{app:additional_res_1}.

\begin{figure*}[!t]
\centering
    \begin{minipage}{0.49\textwidth}
        \centering
        \includegraphics[height=4.3cm]{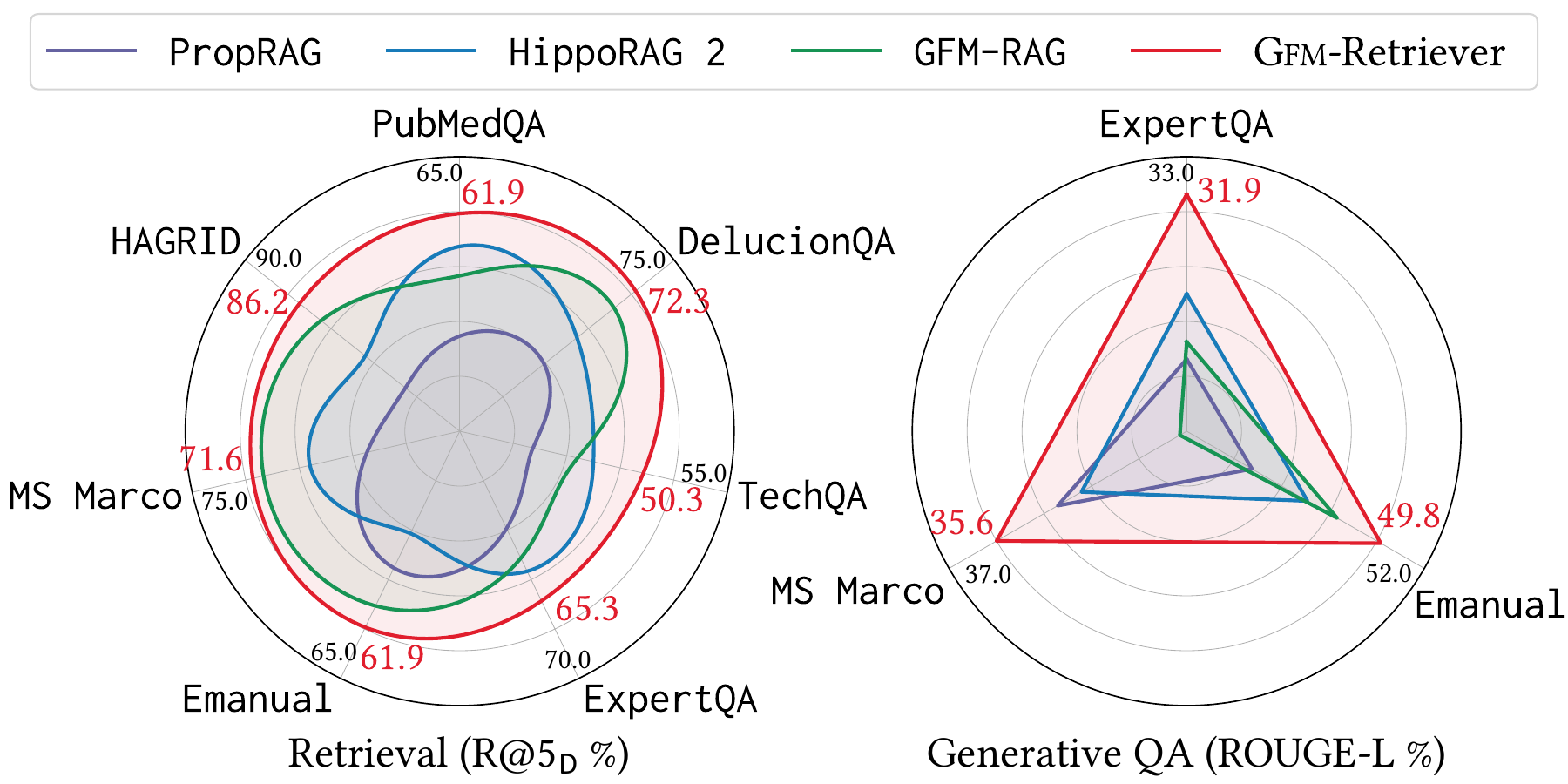}
        \vspace{-0.3cm}
        \caption{Cross-domain generalizability analysis.}
        \label{fig:radar}
    \end{minipage}
    \hfill
    \begin{minipage}{0.49\textwidth}
        \centering
        \includegraphics[height=4.3cm]{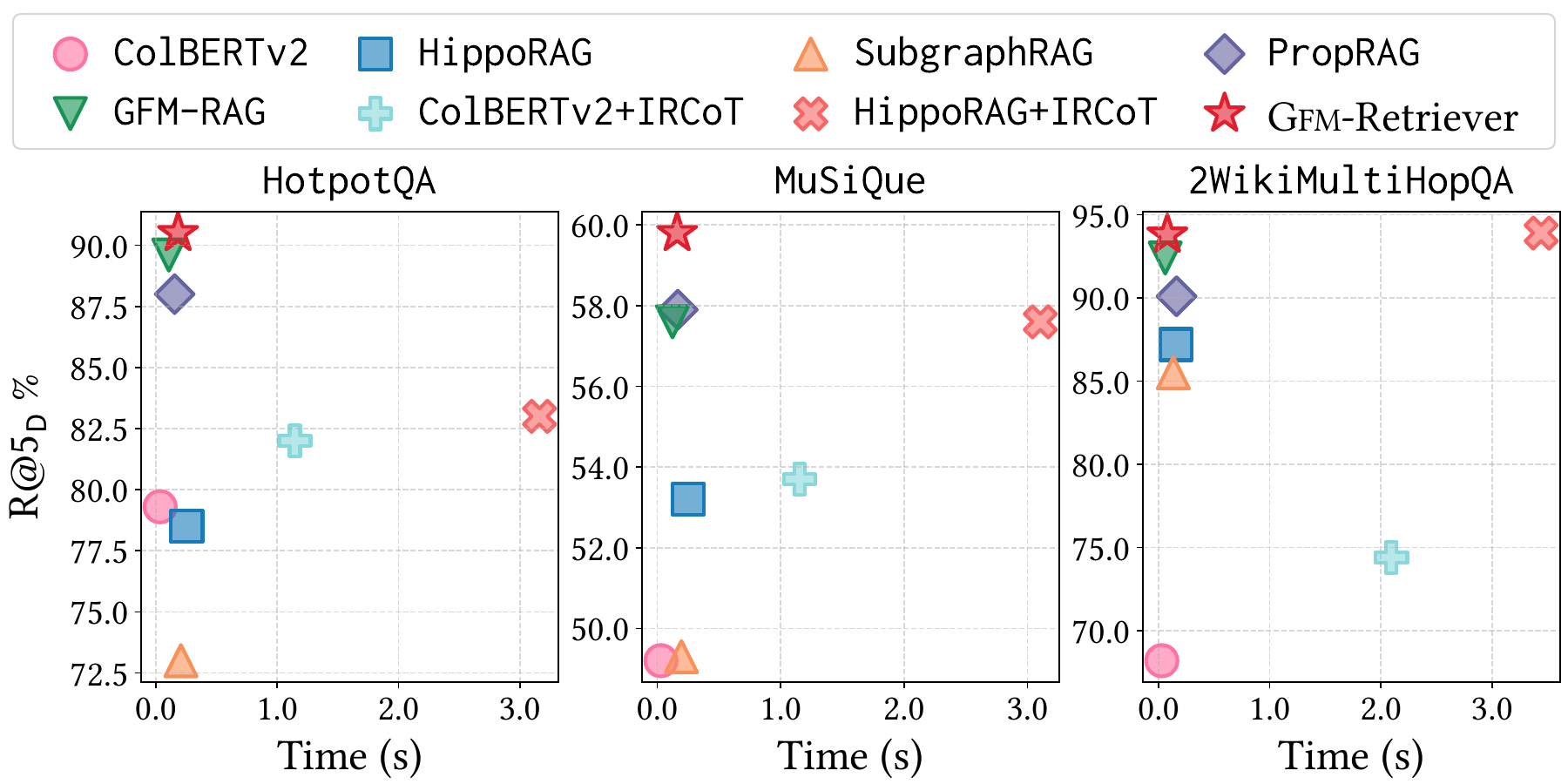}
        \vspace{-0.3cm}
        \caption{Retrieving time efficiency and effectiveness.}
        \label{fig:efficiency}
    \end{minipage}
\end{figure*}
\begin{figure*}[!t]
\centering
    \begin{minipage}{0.49\textwidth}
        \centering
        \includegraphics[height=4.3cm]{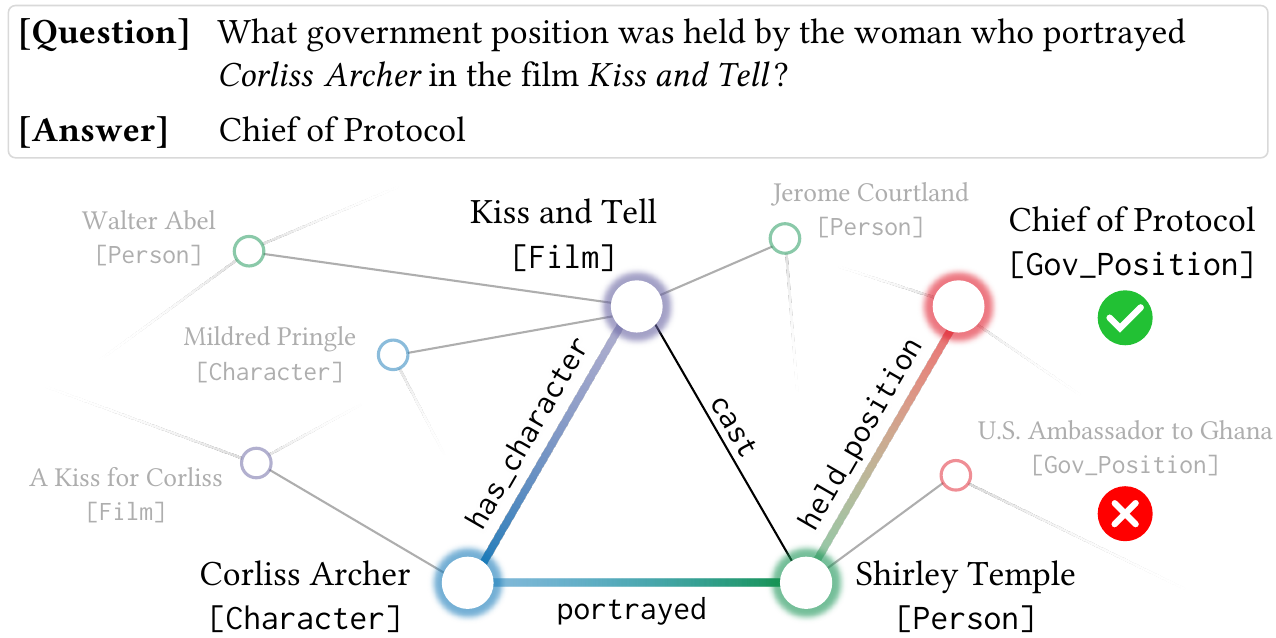}
        \vspace{-0.3cm}
        \caption{Visualizations of the retrieved subgraph.}
        \label{fig:case}
    \end{minipage}
    \hfill
    \begin{minipage}{0.49\textwidth}
        \centering
        \includegraphics[height=4.3cm]{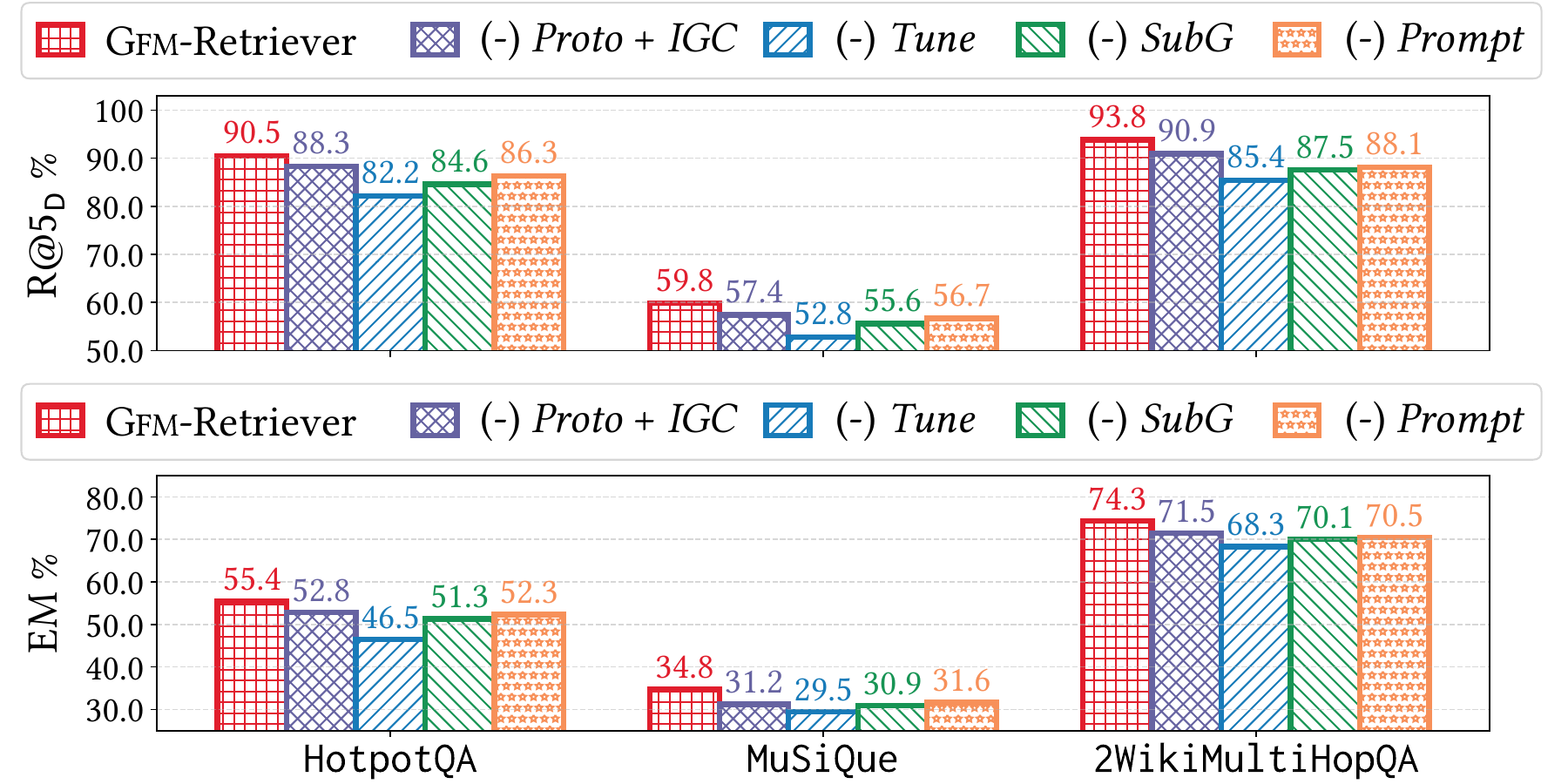}
        \vspace{-0.3cm}
        \caption{Ablation studies on retrieving and multi-hop QA.}
        \label{fig:ablation}
    \end{minipage}
\end{figure*}
\begin{figure*}[!t]
    \centering
    \includegraphics[width=\linewidth]{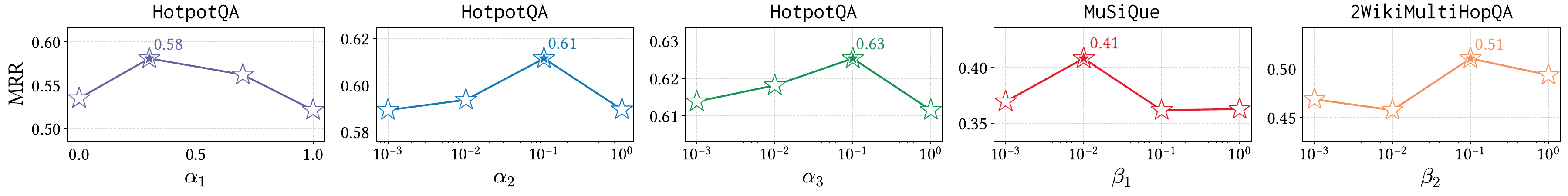}
    \vspace{-0.75cm}
    \caption{Sensitivity analysis of hyper-parameters. Results are reported with the MRR metric.}
    \label{fig:hyper}
\vspace{-0.2cm}
\end{figure*}

\subsection{\textit{RQ4:} Efficiency and Effectiveness Analysis}
\label{rq4:efficiency}
We compare retrieval latency and document-level recall (R@5$_\textsf{D}$) across representative baselines.

\textbf{Results.}
Figure~\ref{fig:efficiency} illustrates a clear efficiency-effectiveness trade-off across retrieval paradigms. Single-step dense retrievers (\eg, \texttt{ColBERTv2}) achieve low latency but suffer from limited recall, while multi-step methods such as \texttt{HippoRAG}+\texttt{IRCoT} improve effectiveness at the cost of substantial latency. In contrast, \modelname~consistently lies in the upper-left region across all datasets, achieving the highest or near-highest R@5$_\textsf{D}$ with sub-second retrieval time. Compared to graph-based baselines like \texttt{SubgraphRAG} and \texttt{GFM-RAG}, it delivers stronger retrieval quality without iterative graph traversal, benefiting from the IB-optimized subgraph selector that extracts a compact query-conditioned subgraph in a single forward pass.

\subsection{\textit{RQ5:} Subgraph Visualization and Case Study}
\label{rq5:case}
Figure~\ref{fig:case} illustrates a representative multi-hop QA example from \texttt{HotpotQA}. 
\modelname~retrieves a compact subgraph that forms a clear reasoning chain from ``Corliss Archer'' to ``Kiss and Tell'', ``Shirley Temple'', and the correct government role ``Chief of Protocol''.
Importantly, entities that are semantically related but irrelevant to the target reasoning path are excluded. 
This case illustrates that \modelname~identifies a minimal, query-aligned subgraph whose explicit relational paths directly support multi-hop inference, rather than aggregating loosely related evidence.

\subsection{\textit{RQ6:} Ablation and Sensitivity Analysis}
\label{rq6:ablation}

\paragraph{\textnormal{\textbf{Ablation.}}}
We ablate \textbf{five} core modules of \modelname: prototype alignment (\textit{Proto}, Eq.~\eqref{eq:proto}), information gain regularizer (\textit{IGC}, Eq.~\eqref{eq:igc}, fine-tuning (\textit{Tune}, Eq.~\eqref{eq:ftn}, the subgraph selector (\textit{SubG}, Section~\ref{sec:method_subgraph}), and the in-context relational prompting strategy (\textit{Prompt}, Section~\ref{sec:method_prompt}). As shown in Figure~\ref{fig:ablation}, removing any component consistently degrades performance on both retrieval and QA.
Specifically, removing \textit{Proto+IGC} leads to clear drops in R@5$_\textsf{D}$, degrading into the retriever in \texttt{GFM-RAG}, sacrificing cross-domain alignment. Disabling \textit{Tune} causes degradation, especially on QA, highlighting the importance of adapting to query-conditioned subgraph selection. Removing \textit{SubG} also results in substantial declines, corresponding to replacing the learned subgraph with a general retrieval strategy. Removing \textit{Prompt} further reduces QA performance, whereas retrieval remains less affected, confirming that explicitly organizing relational paths into in-context prompts is important for translating structural evidence into effective multi-hop reasoning.

\textbf{Sensitivity.}
We study \textbf{five} weighting hyper-parameters: $\alpha_1$, $\alpha_2$, and $\alpha_3$ in Eq.~\eqref{eq:pre}, and $\beta_1$, $\beta_2$ in Eq.~\eqref{eq:bound_2}.
Figure~\ref{fig:hyper} reports the results measured by MRR. Overall, performance is sensitive to different values of the hyper-parameters and contains a reasonable range. 
\vspace{0.33cm}
\section{Conclusion}
\label{sec:conclusion}
In this work, we rethink graph-based retrieval from a structural perspective and propose \modelname, which directly learns query-specific subgraphs as reasoning evidence. By repurposing a pre-trained GFM as a cross-domain retriever and introducing a label-free information bottleneck for optimizing minimal yet sufficient subgraphs, our framework enables efficient and interpretable multi-hop reasoning. Extensive experiments demonstrate strong retrieval quality, QA performance, generalization, and efficiency.

\newpage
\bibliographystyle{ACM-Reference-Format}
\bibliography{ref}

\newpage
\appendix
\counterwithin{table}{section}
\counterwithin{figure}{section}
\counterwithin{equation}{section}

\section{Notations}
\label{app:notatios}
\vspace{0.5cm}

\begin{minipage}{\textwidth}
\label{tab:notation}%
  \centering
  \resizebox{0.85\textwidth}{!}{
    \begin{tabular}{ll}
    \toprule
    \textbf{Notations} & \textbf{Descriptions} \\
    \midrule
    \multicolumn{2}{l}{\textit{Knowledge Graph and Corpus}}\\
    $\mathcal{G}=(\mathcal{V},\mathcal{R},\mathcal{E})$ & Knowledge graph with entity set $\mathcal{V}$, relation set $\mathcal{R}$, and triple set $\mathcal{E}$.\\
    $(u,r,v)\in\mathcal{E}$ & Relational triple from entity $u$ to $v$ via relation $r$.\\
    $\{\mathcal{D}\}$, $d\in{\mathcal{D}}$ & Set of knowledge domains and a domain identifier.\\
    $\operatorname{Dom}_d(e)$ & Predicate indicating entity $e$ belongs to domain $d$.\\
    $\mathcal{C}=\{c_i\}$ & Corpus of document chunks.\\
    $\phi(c_i),\ \phi^{-1}(e)$ & Mapping between document chunks and mentioned entities.\\
    \midrule
    \multicolumn{2}{l}{\textit{Query, Answer, and RAG Interface}}\\
    $\mathbf{q}$, ${\mathcal{V}\mid \mathbf{q}}$ & Input query, and entities explicitly mentioned in the query.\\
    $\mathrm{LLM}(\mathbf{q},\mathcal{G}_{\mathbf{q}},\mathcal{C}_{\mathbf{q}})$, $y$ & The answer generator, and the generated answer.\\
    $\mathcal{C}_{\mathbf{q}}$ & Supporting documents induced by retrieved entities.\\
    \midrule
    \multicolumn{2}{l}{\textit{Subgraph Retrieval}}\\
    $\mathcal{G}_{\mathbf{q}}=(\mathcal{V}_{\mathbf{q}},\mathcal{R}_{\mathbf{q}},\mathcal{E}_{\mathbf{q}})$ & Query-specific retrieved subgraph.\\
    $\mathcal{V}_{\mathbf{q}}$, $\mathcal{R}_{\mathbf{q}}$, $\mathcal{E}_{\mathbf{q}}$ & Entity, relation, and edge sets of the retrieved subgraph.\\
    \midrule
    \multicolumn{2}{l}{\textit{Graph Foundation Model (GFM)}}\\
    $f_{\boldsymbol{\theta}}(\cdot)$, $f_{\boldsymbol{\theta}}^\star(\cdot)$, $\boldsymbol{\theta}$ & The graph foundation model and its pre-trained variant with parameters.\\
    $\mathbf{Z}_e^{(l)}$, $\mathbf{Z}_r^{(l)}$ & Entity and relation embeddings at the GFM layer $l$.\\
    $\mathcal{N}(v)$, $\mathbf{m}_v^{(l+1)}$ & Neighborhood of entity $v$ and aggregated message.\\
    $\textsc{Msg}(\cdot)$, $g(\cdot)$, $\textsc{Agg}(\cdot)$, $\textsc{Update}(\cdot)$ & Learnable GFM message passing and update functions.\\
    $l$, $L$ & GNN layer index and total number of layers.\\
    \midrule
    \multicolumn{2}{l}{\textit{Pre-training and Prototype Learning}}\\
    $(e,r,?)$, $(?,r,e)$ & Masked triples used in KG completion pre-training.\\
    $\mathbb{P}_{\mathbf{q}}(e)$ & Relevance probability of entity $e$ for query $\mathbf{q}$.\\
    $\mathcal{V}_{\mathbf{q}}^+$, $\mathcal{V}_{\mathbf{q}}^-$ & Positive and negative entity sets.\\
    $\mathcal{L}_{\text{BCE}}$, $\mathcal{L}_{\text{rank}}$ & Binary cross-entropy and ranking loss.\\
    $\mathcal{Q}_d$, $(\mathbf{q}_i,e_i)$ & Query-entity pairs sampled from domain $d$.\\
    $\mathbf{c}_d$, $\{\mathbf{c}_d\}$ & Domain prototype embedding and all prototypes.\\
    $\mathcal{L}_{\text{proto}}$, $\mathcal{L}_{\text{IGC}}$ & Prototype contrastive loss and information gain contrast regularizer.\\
    $\mathcal{L}_{\text{pre}}$, $\alpha_1$, $\alpha_2$, $\alpha_3$ & Overall pre-training objective and loss weights.\\
    \midrule
    \multicolumn{2}{l}{\textit{Subgraph Selection and Information Bottleneck}}\\
    $\mathbf{p}$ & Query-conditioned prompt.\\
    $\mathbf{Z}_e^\star$, $\widetilde{\mathbf{Z}}_e$ & Frozen entity embedding, and query-conditioned entity representation.\\
    $\mathbf{W}_1$, $\mathbf{W}_2$ & Learnable projection matrices.\\
    $S_e(\mathbf{q})$, $g_e$, $\epsilon$ & Soft entity selection probability, Gumbel noise, and selection threshold.\\
    $I(\cdot~;~\cdot)$, $H(\cdot\mid\cdot)$ & Mutual information and conditional entropy.\\
    $\mathbf{S}_e(\mathbf{q})$, $\mathbf{L}$ & Soft selection vector and normalized graph Laplacian.\\
    $\beta$, $\beta_1$, $\beta_2$ & Information Bottleneck regularization coefficients.\\
    $\mathcal{L}_{\text{IB}}$, $\mathcal{L}_{\text{LIB}}$ & Information Bottleneck objective and its label-free variant.\\
    $\mathbf{G}_{\mathbf{q}}$, $\mathcal{L}_{\text{NCE}}$ & Pooled subgraph representation and contrastive lower-bound loss.\\
    $\mathcal{L}_{\text{size}}$, $\mathcal{L}_{\text{con}}$ & Subgraph size and connectivity regularization terms.\\
    $\mathcal{L}_{\text{ftn}}$ & Overall fine-tuning objective.\\
    \midrule
    \multicolumn{2}{l}{\textit{Path Reasoning and Prompt Construction}}\\
    $\mathbf{M}$ & Binary entity-document incidence matrix.\\
    $w_e$, $R_c$ & Entity importance weight and document relevance score.\\
    $\Pi_{\mathbf{q}}$, $\pi$ & Set of reasoning paths and a single relational path.\\
    $S_r(\pi)$, $\phi_{\text{path}}(e)$, $\omega(r)$ & Path relevance score, path-based entity centrality, and relation depth penalty.\\
    \bottomrule
    \end{tabular}%
}
\end{minipage}
\clearpage
\section{Algorithm and Complexity Analysis}
\label{app:algo}
The pre-training, fine-tuning, and inference (QA task) pipeline of \modelname~is illustrated in Algorithm~\ref{alg:pretrain}, Algorithm~\ref{alg:finetune}, and Algorithm~\ref{alg:qa}. We next analyze their computational complexity.

\begin{algorithm}[t]
\caption{Pre-training of the GFM Retriever $f_{\boldsymbol{\theta}}(\cdot)$.}
\label{alg:pretrain}
\KwIn{
The multi-domain indexed KG $\mathcal{G}=(\mathcal{V},\mathcal{R},\mathcal{E})$ over domains $\mathcal{D}$;
query-conditioned GFM $f_{\boldsymbol{\theta}}(\cdot)$ with $L$ layers;
hyper-parameters $\alpha_{1}$, $\alpha_{2}$, $\alpha_{3}$, $\tau$.}
\KwOut{Pre-trained GFM retriever $f_{\boldsymbol{\theta}}^{\star}(\cdot)$.}
Initialize model parameters $\boldsymbol{\theta}$\;
\textcolor{gray}{\tcp{\# Phase I: KG completion pre-training}}       
\For{each pre-training epoch}{
    \For{each triple $(e,r,?)$ or $(?,r,e)$ sampled from $\mathcal{G}$}{
        Construct pseudo-query $\mathbf{q}$ from triplets\;
        Initialize embeddings: $\mathbf{Z}_{e}^{(0)} \gets$ Eq.~\eqref{eq:init}\;
        \For{$l=0$ to $L-1$}{
            Message passing: $\mathbf{m}_v^{(l+1)} \gets$ Eq.~\eqref{eq:message}\;
            Entity update: $\mathbf{Z}_v^{(l+1)} \gets$ Eq.~\eqref{eq:update}\;
        }
        Predict entity relevance $\mathbb{P}_{\mathbf{q}}(e) \gets$ Eq.~\eqref{eq:relevance}\;
        Loss: $\mathcal{L}_{\text{BCE}} \gets$ Eq.~\eqref{eq:bce}, $\mathcal{L}_{\text{rank}} \gets$ Eq.~\eqref{eq:rank}\;
        $\mathcal{L}_{\text{Phase I}}=\alpha_{1}\mathcal{L}_{\text{BCE}}+(1-\alpha_{1})\mathcal{L}_{\text{rank}}$ (Eq.~\eqref{eq:pre})\;
    }
}
\textcolor{gray}{\tcp{\# Phase II: Continuous pre-training}}       
\For{each pre-training epoch}{
    Sample pairs $(\mathbf{q}_i,e_i)\in\mathcal{Q}_d$ from each domain $d\in\mathcal{D}$\;
    Compute domain prototype $\mathbf{c}_d$\;
    Compute prototype contrastive loss: $\mathcal{L}_{\text{proto}} \gets$ Eq.~\eqref{eq:proto}\;
    Compute information gain contrast $\mathcal{L}_{\text{IGC}} \gets$ Eq.~\eqref{eq:igc}\;
    $\mathcal{L}_{\text{Phase II}}=\alpha_{2}\mathcal{L}_{\text{proto}}+\alpha_{3}\mathcal{L}_{\text{IGC}}$ (Eq.~\eqref{eq:pre}).\;
}
\end{algorithm}

\subsection{Computational Complexity of the Pre-training Stage}
\label{app:complexity_pretrain}
The pre-training stage consists of two phases: 
\textbf{(1)} mixed-domain KG completion pre-training (Phase~I), and 
\textbf{(2)} prototype-driven continuous pre-training (Phase~II).
Unless otherwise stated, we focus on the asymptotic complexity with respect to graph size, and omit constant factors related to embedding dimensionality and the internal implementations of neural modules (\eg, $\textsc{Msg}(\cdot)$, $\textsc{Update}(\cdot)$, $\textsc{Agg}(\cdot)$, $\textsc{MLP}(\cdot)$, \etc).

\textbf{Phase I: Mixed-domain KG Completion Pre-training.}
Phase~I performs KG completion over the multi-domain indexed graph $\mathcal{G}=(\mathcal{V},\mathcal{R},\mathcal{E})$ using query-conditioned GFM inference. For each masked triple, a pseudo-query $\mathbf{q}$ is constructed.
\begin{itemize}[leftmargin=*]
    \item \textbf{Embedding initialization.}
    Initializing entity and relation embeddings for all $e\in\mathcal{V}$ and $r\in\mathcal{R}$ requires $\mathcal{O}(|\mathcal{V}|+|\mathcal{R}|)$.
    \item \textbf{Query-conditioned message passing.}
    For each GFM layer, messages are computed over all edges $(u,r,v)\in\mathcal{E}$ and aggregated at each node. The per-layer complexity is $\mathcal{O}(|\mathcal{E}|+|\mathcal{V}|)$, and stacking $L$ layers yields $\mathcal{O}(L(|\mathcal{E}|+|\mathcal{V}|))$.
    \item \textbf{Entity relevance scoring.}
    Computing $\mathbb{P}_{\mathbf{q}}(e)$ for all entities requires a forward pass of an MLP, which requires $\mathcal{O}(|\mathcal{V}|)$.
    \item \textbf{KG completion loss.}
    Both $\mathcal{L}_{\text{BCE}}$ and $\mathcal{L}_{\text{rank}}$ involve summations over positive and negative entity sets whose union equals $\mathcal{V}$. Hence, their computation costs $\mathcal{O}(|\mathcal{V}|)$.
\end{itemize}
Summing up, time complexity for Phase~I is $\mathcal{O}(L|\mathcal{E}| + L|\mathcal{V}| + |\mathcal{R}| + |\mathcal{V}|)$.
Since $|\mathcal{E}|\geqslant|\mathcal{V}|$ and $L\geqslant 1$, it can be simplified to $\mathcal{O}(L|\mathcal{E}| + |\mathcal{V}| + |\mathcal{R}|)$
with $\mathcal{O}(L|\mathcal{E}|)$ being the dominant term.
If masking over all triples in $\mathcal{E}$ once per epoch, the per-epoch complexity becomes:
\begin{equation}
    \mathcal{O}\big(L|\mathcal{E}|^2 + |\mathcal{E}||\mathcal{V}| + |\mathcal{E}||\mathcal{R}|\big),
\end{equation}
which is dominated by $\mathcal{O}(L|\mathcal{E}|^2)$.

\textbf{Phase II: Prototype-driven Continuous Pre-training.}
Phase~II performs semantic alignment using a subset of query-entity pairs $\bigcup_{d\in\mathcal{D}}\mathcal{Q}_d$ sampled from each domain.
\begin{itemize}[leftmargin=*]
    \item \textbf{Prototype construction.}
    Computing domain prototypes needs averaging entity representations, yielding $\mathcal{O}(\sum_{d\in\mathcal{D}}|\mathcal{Q}_d|)$.
    \item \textbf{Prototype contrastive loss.}
    For each pair, the softmax denominator sums over all domains, yielding $\mathcal{O}((\sum_{d\in\mathcal{D}}|\mathcal{Q}_d|)\cdot|\mathcal{D}|)$.
    \item \textbf{Information gain contrast regularization.}
    The KL-divergence is calculated over the entity distribution $\mathbb{P}_{\mathbf{q}}(e)\in\mathbb{R}^{|\mathcal{V}|\times1}$.
    For each query, this requires $\mathcal{O}(|\mathcal{V}|)$ time, yielding $\mathcal{O}((\sum_{d\in\mathcal{D}}|\mathcal{Q}_d|)\cdot|\mathcal{V}|)$.
    \item \textbf{GFM forward inference.}
    Since $\mathbb{P}_{\mathbf{q}}(e)$ depends on the query-conditioned GFM output, each query requires an $L$-layer message-passing forward pass, incurring $\mathcal{O}(L(|\mathcal{E}|+|\mathcal{V}|))$.
    Over all sampled queries, this contributes $\mathcal{O}((\sum_{d\in\mathcal{D}}|\mathcal{Q}_d|)\cdot L|\mathcal{E}|)$,
    which is the dominant cost in Phase~II.
\end{itemize}
Combining all, the per-epoch complexity of Phase~II is:
\begin{equation}
    \mathcal{O}\Big(
\big(\sum_{d\in\mathcal{D}}|\mathcal{Q}_d|\big)
\cdot (L|\mathcal{E}| + |\mathcal{V}| + |\mathcal{D}|)
\Big).
\end{equation}

\textbf{Overall Pre-training Complexity.}
Combining Phase~I and Phase~II, the total per-epoch time complexity of pre-training is:
\begin{equation}
    \mathcal{O}\big(L|\mathcal{E}|^{2} + |\mathcal{E}||\mathcal{V}| + |\mathcal{E}||\mathcal{R}|\big) + \mathcal{O}\Big(\sum_{d\in\mathcal{D}}|\mathcal{Q}_d| (L|\mathcal{E}| + |\mathcal{V}| + |\mathcal{D}|)\Big).
\end{equation}
Since Phase~II is performed on a subset and Phase~I dominates in large-scale knowledge graphs, the overall complexity is asymptotically governed by $\mathcal{O}(L|\mathcal{E}|^2)$, which highlights that the computational cost of pre-training is primarily determined by repeated query-conditioned message passing over the entire indexed graph, which motivates our later use of parameter-efficient fine-tuning and subgraph-level optimization in the downstream stages.

\begin{algorithm}[t]
\caption{Fine-tuning with the Subgraph Selector.}
\label{alg:finetune}
\KwIn{
The multi-domain indexed KG $\mathcal{G}=(\mathcal{V},\mathcal{R},\mathcal{E})$;
frozen entity embeddings $\mathbf{Z}_e^{\star}$ from $f_{\boldsymbol{\theta}}^{\star}(\cdot)$; training queries $\{(\mathbf{q},\mathcal{V}_{\mathbf{q}}^{+})\}$;
hyper-parameters $\alpha_1$, $\beta_1$, $\beta_2$, $\tau$, $\epsilon$.}
\KwOut{Trained subgraph selector producing $\mathcal{G}_{\mathbf{q}}$.}
\For{each mini-batch $\mathcal{B}$ of queries}{
    \For{each query $\mathbf{q}\in\mathcal{B}$}{
        Construct query prompt $\mathbf{p}=\operatorname{MLP}(\mathbf{q})$\;
        \For{each entity $e\in\mathcal{V}$}{
            \textcolor{gray}{\tcp{Learn query-conditioned subgraph}}    
            Query-conditioned entity: $\widetilde{\mathbf{Z}}_e \gets$ Eq.~\eqref{eq:query_e}\;
            Sample selection gate: $S_e(\mathbf{q}) \gets$ Eq.~\eqref{eq:prob}\;
        }
        Induce subgraph $\mathcal{G}_{\mathbf{q}}$ by thresholding $\epsilon$\;
        \textcolor{gray}{\tcp{Label-free IB regularization}} 
        Retrieval loss: $\mathcal{L}_{\text{BCE}} \gets$ Eq.~\eqref{eq:bce}, $\mathcal{L}_{\text{rank}} \gets$ Eq.~\eqref{eq:rank}\;
        Pooled subgraph representation: $\mathbf{G}_{\mathbf{q}}=\sum_{e\in\mathcal{V}_{\mathbf{q}}} S_e(\mathbf{q})\cdot \widetilde{\mathbf{Z}}_e$\;
        Lower bound of $I(\mathbf{q};\mathcal{G}_{\mathbf{q}})$: $\mathcal{L}_{\text{NCE}} \gets$ Eq.~\eqref{eq:bound_1}\;
        Penalties: $\mathcal{L}_{\text{size}} \gets$ Eq.~\eqref{eq:size_con}, $\mathcal{L}_{\text{con}} \gets$ Eq.~\eqref{eq:size_con}\;
        Combine the IB surrogate: $\lceil\mathcal{L}_{\text{LIB}}\rceil\gets$ Eq.~\eqref{eq:lib_t}\;         
        Calculate the final fine-tuning loss: $\mathcal{L}_{\text{ftn}}\gets$ Eq.~\eqref{eq:ftn}\;
    }
    Update selector parameters via gradient descent\;
}
\end{algorithm}

\subsection{Computational Complexity of the Fine-tuning Stage}
\label{app:complexity_finetune}

The fine-tuning stage optimizes the label-free IB objective (Eq.~\eqref{eq:label_free}) via tractable surrogate using Proposition~\ref{prop:lower} (Eq.~\eqref{eq:bound_1}) and Proposition~\ref{prop:upper} (Eq.~\eqref{eq:bound_2}), together with supervised retrieval loss (Eq.~\eqref{eq:bce} and Eq.~\eqref{eq:rank}).
Follow the same protocal in Appendix~\ref{app:complexity_pretrain}, we omit constant factors related to hidden dimensions and the internal implementations of neural modules (\eg, $\operatorname{MLP}(\cdot)$), focusing on asymptotic dependence on graph and batch sizes.

\textbf{Per-query Complexity (mini-batch $\mathcal{B}$).}
For a query $\mathbf{q}$, the fine-tuning stage consists of the following computation modules.
\begin{itemize}[leftmargin=*]
    \item \textbf{Query prompt construction.}
    Prompt $\mathbf{p}=\operatorname{MLP}(\mathbf{q})$ is computed once per query.
    Treating $\operatorname{MLP}(\cdot)$ as a constant-cost operator \textit{w.r.t.} graph size, this step requires $\mathcal{O}(1)$.
    \item \textbf{Query-conditioned entity embedding construction.}
    For every entity $e\in\mathcal{V}$, we form $\widetilde{\mathbf{Z}}_e=\big[\mathbf{Z}_e^{\star}\ \|\ \mathbf{W}_{1}\mathbf{p}\ \|\ \mathbf{W}_{2}\mathbf{q}\big]$, which is a per-entity operation. Thus, the complexity is $\mathcal{O}(|\mathcal{V}|)$.
    \item \textbf{Differentiable subgraph gate sampling.}
    For each $e\in\mathcal{V}$, the Gumbel-Sigmoid gate $S_e(\mathbf{q})$ in Eq.~\eqref{eq:prob} is computed once. This is also linear in the number of entities: $\mathcal{O}(|\mathcal{V}|)$.
    \item \textbf{Subgraph induction by thresholding.}
    Constructing $\mathcal{V}_{\mathbf{q}}$ requires one pass over all entities that yields $\mathcal{O}(|\mathcal{V}|)$.
    \item \textbf{Subgraph generation.} Inducing $\mathcal{R}_{\mathbf{q}}$ and $\mathcal{E}_{\mathbf{q}}$ from $\mathcal{V}_{\mathbf{q}}$ can be done by scanning edges, which is at most $\mathcal{O}(|\mathcal{E}|)$. We state explicitly below as an optional dominant term if full edge scanning is used.
    \item \textbf{Supervised retrieval loss.}
    Both $\mathcal{L}_{\text{BCE}}$ and $\mathcal{L}_{\text{rank}}$ sum over positives $\mathcal{V}_{\mathbf{q}}^{+}$ and negatives $\mathcal{V}_{\mathbf{q}}^{-}$.
    In the worst case, they are computed over all entities and yield $\mathcal{O}(|\mathcal{V}|)$, which matches the dependence of Eq.~\eqref{eq:bce} and Eq.~\eqref{eq:rank} on entity-level distributions.
    \item \textbf{Pooled subgraph representation.}
    The pooled $\mathbf{G}_{\mathbf{q}}$ requires iterating over the selected entity set $\mathcal{V}_{\mathbf{q}}$: $\mathcal{O}(|\mathcal{V}_{\mathbf{q}}|)\leqslant \mathcal{O}(|\mathcal{V}|)$.
    \item \textbf{Noise-contrastive estimation term.}
    Eq.~\eqref{eq:bound_1} computes a denominator summing over negative pairs $(\mathbf{G}_{\mathbf{q}'},\mathbf{q}')$ for $\mathbf{q}'\in\mathcal{B}$. Assuming $\mathbf{G}_{\mathbf{q}'}$ is precomputed for all $\mathbf{q}'\in\mathcal{B}$ (as in Algorithm~\ref{alg:finetune}), computing the softmax-normalized score per query costs $\mathcal{O}(|\mathcal{B}|)$.
    (If $\mathbf{G}_{\mathbf{q}'}$ is not cached, the cost would inherit $\mathcal{O}(|\mathcal{V}_{\mathbf{q}'}|)$ construction, which is accounted for by the pooled representation.)
    \item \textbf{Size penalty.}
    The size term $\mathcal{L}_{\textnormal{size}}$ costs $\mathcal{O}(|\mathcal{V}_{\mathbf{q}}|)\leqslant \mathcal{O}(|\mathcal{V}|)$.
    \item \textbf{Connectivity penalty.}
    The connectivity term $\mathcal{L}_{\textnormal{con}}$ involves multiplication by the normalized Laplacian $\mathbf{L}$.
    Computing $\mathbf{L}\mathbf{S}_e(\mathbf{q})$ can be implemented as the sparse matrix-vector multiplication with sparsity pattern given by $\mathcal{E}$, which costs $\mathcal{O}(|\mathcal{E}|)$, which is typically the dominant term among the regularizers.
\end{itemize}

\textbf{Per-batch Complexity.}
Summing up, a single query $\mathbf{q}$ gives:
\begin{align}
    \underbrace{\mathcal{O}(|\mathcal{V}|)}_{\text{Eq.~\eqref{eq:query_e}}} + \underbrace{\mathcal{O}(|\mathcal{V}|)}_{\text{Eq.~\eqref{eq:prob}}} + \underbrace{\mathcal{O}(|\mathcal{V}|)}_{\text{thresholding}} + \underbrace{\mathcal{O}(|\mathcal{V}|)}_{\mathcal{L}_{\text{BCE}}\text{,}~\mathcal{L}_{\text{rank}}} \notag \\
    + \underbrace{\mathcal{O}(|\mathcal{V}_{\mathbf{q}}|)}_{\mathbf{G}_{\mathbf{q}}\text{,}~\mathcal{L}_{\text{size}}} + \underbrace{\mathcal{O}(|\mathcal{B}|)}_{\mathcal{L}_{\text{NCE}}} + \underbrace{\mathcal{O}(|\mathcal{E}|)}_{\mathcal{L}_{\text{con}}}.
\end{align}
Using $|\mathcal{V}_{\mathbf{q}}|\leqslant |\mathcal{V}|$ and grouping terms, we obtain $\mathcal{O}(|\mathcal{V}| + |\mathcal{E}| + |\mathcal{B}|)$.

For a mini-batch $\mathcal{B}$, the total cost is $\mathcal{O}(|\mathcal{B}|\cdot(|\mathcal{V}| + |\mathcal{E}| + |\mathcal{B}|))$.
Expanding the product yields $\mathcal{O}(|\mathcal{B}||\mathcal{V}| + |\mathcal{B}||\mathcal{E}| + |\mathcal{B}|^2)$.
In typical large KGs, $|\mathcal{E}|$ dominates $|\mathcal{V}|$ and $|\mathcal{B}|\ll |\mathcal{E}|$.
Thus, the fine-tuning is asymptotically dominated by the connectivity regularizer, giving $\mathcal{O}(|\mathcal{B}||\mathcal{E}|)$, where the per-query dominant cost is $\mathcal{O}(|\mathcal{E}|)$.

Notably, if $\mathcal{E}_{\mathbf{q}}$ and $\mathcal{R}_{\mathbf{q}}$ are explicitly induced by scanning the full edge set each iteration, that additional cost is at most $\mathcal{O}(|\mathcal{E}|)$ per query and is already subsumed by the $\mathcal{O}(|\mathcal{E}|)$ term above.
In contrast, if one uses edge-indexing restricted to the selected node set $\mathcal{V}_{\mathbf{q}}$, the practical runtime can be reduced, but the worst-case asymptotic bound remains unchanged.

\subsection{Computational Complexity of the Inference Stage (QA)}
\label{app:complexity_inference}

Given a query $\mathbf{q}$ on the target knowledge graph $\mathcal{G}^T=(\mathcal{V}^T,\mathcal{R}^T,\mathcal{E}^T)$, the inference stage consists of
\textbf{(1)} subgraph retrieval via the trained selector (Eq.~\eqref{eq:prob}), and 
\textbf{(2)} relational paths induced in-context prompting, including entity-to-document mapping and truncated DFS path extraction.
We omit constant factors related to hidden dimensions and neural module internals.

\textbf{Step-wise Complexity Analysis.}
For a single query $\mathbf{q}$, Algorithm~\ref{alg:qa} includes the following modules.
\begin{itemize}[leftmargin=*]
    \item \textbf{Selection probability.}
    The selector computes $S_e(\mathbf{q})$ (Eq.~\eqref{eq:prob}), which requires one pass over all entities, hence $O(|\mathcal{V}^T|)$.
    \item \textbf{Subgraph induction by thresholding.}
    Constructing $\mathcal{V}_{\mathbf{q}}$ costs $\mathcal{O}(|\mathcal{V}^T|)$.
    If $\mathcal{E}_{\mathbf{q}}$ is explicitly induced from $\mathcal{V}_{\mathbf{q}}$ by scanning all edges in $\mathcal{E}^T$, the worst-case additional cost is $\mathcal{O}(|\mathcal{E}^T|)$.
    \item \textbf{Entity weighting.}
    For each entity $e\in\mathcal{V}_{\mathbf{q}}$, inference computes $w_e$ (Eq.~\eqref{eq:weight}), which is linear in the size of the selected entity set: $\mathcal{O}(|\mathcal{V}_{\mathbf{q}}|)$.
    (The term $\sum_c \mathbf{M}[e,c]$ can be precomputed per entity from the fixed mapping matrix $\mathbf{M}$.)
    \item \textbf{Document relevance computation.}
    The document scores are given by $R_c=\mathbf{M}^{\top}w_e$.
    Treating $\mathbf{M}\in\{0,1\}^{|\mathcal{E}_{\mathbf{q}}|\times|\mathcal{C}_{\mathbf{q}}|}$ as a sparse matrix, computing $\mathbf{M}^{\top}w_e$ is a sparse matrix-vector multiplication.
    Its time complexity is proportional to the number of non-zeros ($\mathrm{nnz}$) in $\mathbf{M}$ restricted to the selected entities as $\mathcal{O}(\mathrm{nnz}(\mathbf{M}_{\mathcal{V}_{\mathbf{q}}}))$, where $\mathbf{M}_{\mathcal{V}_{\mathbf{q}}}$ denotes the submatrix induced by $\mathcal{V}_{\mathbf{q}}$.
    In the worst case, this is upper-bounded by scanning all entity-to-chunk links in the query-restricted support, \ie, $\mathcal{O}(|\mathcal{E}_{\mathbf{q}}|  |\mathcal{C}_{\mathbf{q}}|)$,
    while in practice it is typically much smaller due to sparsity.
    \item \textbf{Top-$K$ document selection.}
    Selecting the top-$K$ documents from $\{R_c\}$ can be done in $\mathcal{O}(|\mathcal{C}_{\mathbf{q}}|)$ (\eg, linear-time selection),
    or $\mathcal{O}(|\mathcal{C}_{\mathbf{q}}|\log|\mathcal{C}_{\mathbf{q}}|)$ if full sorting is used.
    \item \textbf{Relational path extraction.}
    In the worst case, enumerating all paths ($\leqslant l$) may grow exponentially due to branching in DFS.
    Denoting the maximum degree in the induced subgraph $\mathcal{G}_{\mathbf{q}}$ by $\Delta$, the worst-case number of DFS expansions is $\mathcal{O}(\Delta^{l})$.
    \item \textbf{Path scoring.}
    For each extracted path $\pi$, inference computes $S_r(\pi)$ (Eq.~\eqref{eq:sr}), which costs $\mathcal{O}(k)\leqslant \mathcal{O}(l)$ per path, hence over all extracted paths requires $\mathcal{O}(|\Pi_{\mathbf{q}}|\cdot l)$.
    \item \textbf{Path ranking and subset selecting.}
    Ranking paths by $\{S_r(\pi)\}$ costs $\mathcal{O}(|\Pi_{\mathbf{q}}|\log|\Pi_{\mathbf{q}}|)$, or $\mathcal{O}(|\Pi_{\mathbf{q}}|)$ with top-$K$ selection.
    \item \textbf{LLM answer generation.}
    The answer generation $\hat{y}=\operatorname{LLM}(\cdot)$ depends on the LLM inference cost and prompt length. We omit it for brevity.
\end{itemize}

\begin{algorithm}[t]
\caption{Inference with In-context Prompting.}
\label{alg:qa}
\KwIn{
Query $\mathbf{q}$; target KG $\mathcal{G}^T$;
trained selector producing $S_e(\mathbf{q})$;
entity-to-document mapping matrix $\mathbf{M}$;
DFS search hop limit $l$; top-$K$ documents $\mathcal{C}_{\mathbf{q}}^{(K)}$.
}
\KwOut{Generated answer $\hat{y}$.}
\textcolor{gray}{\tcp{Subgraph retrieval}}
Selection probabilities $S_e(\mathbf{q}) \gets$ Eq.~\eqref{eq:prob} for all $e\in\mathcal{V}^T$\;
Induce subgraph $\mathcal{G}_{\mathbf{q}}=(\mathcal{V}_{\mathbf{q}},\mathcal{R}_{\mathbf{q}},\mathcal{E}_{\mathbf{q}})$\;
\textcolor{gray}{\tcp{Entity-to-document scoring}}
\For{each entity $e\in\mathcal{V}_{\mathbf{q}}$}{
    Compute importance weight $w_e \gets$ Eq.~\eqref{eq:weight}\;
}
Compute relevance $R_c=\mathbf{M}^{\top}w_e$ and select top-$K$ documents\;
\textcolor{gray}{\tcp{Reasoning path extraction.}}
Extract relational paths $\Pi_{\mathbf{q}}$ from $\mathcal{G}_{\mathbf{q}}$ via $l$-hop truncated DFS\;
\For{each path $\pi$}{
    Compute path score $S_r(\pi) \gets$ Eq.~\eqref{eq:sr}.
}
Select top-ranked paths $\Pi_{\mathbf{q}}^{\star}$\;
\textcolor{gray}{\tcp{In-context prompting and generation}}
Construct prompt from $\Pi_{\mathbf{q}}^{\star}$ and retrieved documents\;
Generate answer $\hat{y}=\operatorname{LLM}(\mathbf{q},\Pi_{\mathbf{q}}^{\star},\mathcal{C}_{\mathbf{q}}^{(K)})$\;
\end{algorithm}

\textbf{Overall Inference Complexity.}
Excluding the LLM generation cost, the per-query inference complexity is:
\begin{align}
    \mathcal{O}(|\mathcal{V}^T|) + \!\!\!\!\!\!\!\!\!\!\!\!\underbrace{\mathcal{O}(|\mathcal{E}^T|)}_{\text{(optional, if inducing $\mathcal{E}_{\mathbf{q}}$ by edge scan)}}\!\!\!\!\!\!\!\!\!\!\!\!\!\!  + \mathcal{O}(|\mathcal{V}_{\mathbf{q}}|) + \mathcal{O}\big(\mathrm{nnz}(\mathbf{M}_{\mathcal{V}_{\mathbf{q}}})\big) \notag \\
    + \mathcal{O}(|\mathcal{C}_{\mathbf{q}}|) + \mathcal{O}(\Delta^{l}) + \mathcal{O}(|\Pi_{\mathbf{q}}|\cdot l) + \mathcal{O}(|\Pi_{\mathbf{q}}|\log|\Pi_{\mathbf{q}}|).
\end{align}
As $|\mathcal{V}_{\mathbf{q}}|\leqslant|\mathcal{V}^T|$, we obtain the simplified bound:
\begin{equation}
    \mathcal{O}\Big(|\mathcal{V}^T| + |\mathcal{E}^T| +\mathrm{nnz}(\mathbf{M}_{\mathcal{V}_{\mathbf{q}}}) + \Delta^{l} + |\Pi_{\mathbf{q}}|\log|\Pi_{\mathbf{q}}|\Big).
\end{equation}
In typical settings, the dominant cost is the truncated DFS path extraction, which can be exponential in the hop budget $l$ in the worst case.

\section{Proofs}
\label{app:proof}
\subsection{Proof of Proposition~\ref{prop:expressivity}}
\label{app:proof_expressivity}
We first restate Proposition~\ref{prop:expressivity} for reference.
\begin{mathbox}
\textsc{Proposition~\ref{prop:expressivity}} (\textbf{\textsc{Multi-domain logical expressivity of query-conditioned GFM}})
\textit{Let $\mathcal{G}$ be a knowledge graph indexed from multiple domains. $\{\operatorname{Dom}_d(e)\}$ denotes unary predicates indicate domain attributes of entity $e$. For a query $q$, the query-conditioned GFM can learn a rule $\textnormal{\textsf{R}}(\textnormal{$\mathbf{q}$},e)$ if and only if $\textnormal{\textsf{R}}(\textnormal{$\mathbf{q}$},e)$ is a formula in graded modal logic over $\mathcal{G}$ with constant \textnormal{$\mathbf{q}$}, \ie, $\exists^{\geqslant N} e^\prime(\textnormal{\textsf{R}}(e^\prime,e)\wedge\operatorname{Dom}_d(e^\prime)\wedge\psi(e^\prime))$, where $\psi(\cdot)$ a subformula.}
\end{mathbox}

\begin{proof}
    Let $\mathcal{G} = (\mathcal{V}, \mathcal{R}, \mathcal{E})$ be the KG indexed from the multi-domain documents. $\mathcal{D}: \mathcal{V} \to \{0, 1\}^{|\mathbb{D}|}$ maps each entity node to its domain attributes. We denote the unary predicate $\operatorname{Dom}_d(v)$ as true if the $d$-th entry of $\mathcal{D}(v)$ is 1. We assume the query-conditioned GFM $f_{\boldsymbol{\theta}}(\cdot)$ follows the message passing scheme defined in Eq.~\eqref{eq:message} and Eq.~\eqref{eq:update}.
    To prove Proposition~\ref{prop:expressivity}, we establish the equivalence between the logical formulas learnable by this GFM and the Graded Modal Logic (CML)~\cite{de2000note, otto2019graded, qiu2024understanding} extended with domain predicates.

    \textbf{(1) Sufficiency (Constructive Proof).}
    We prove that for any formula $\varphi$ in the multi-domain CML logic, there exists a parameterization of the GFM that learns $\varphi$. We proceed by structural induction on the formula construction. The bad case is the atomic formulas (\eg, $\operatorname{Dom}_d(e)$) are directly contained in the initial entity embeddings $\mathbf{Z}_e^{(0)}$. A simple linear projection (1-layer GFM) can encode these embeddings, so that the base predicates are learnable.
    For the inductive step, assume a sub-formula $\psi(e^\prime)$ is learnable by a GFM with $l-1$ layers, such that the dimension $j$ of $\mathbf{Z}_{e^\prime}^{(l-1)}$ acts as an indication $\mathbf{Z}_{e^\prime}^{(l-1)}[j] = \mathbb{I}(\mathcal{G}, e^\prime \models \psi)$, where $\mathbb{I}(\cdot)$ is the indicator. 
    
    We next show the formula $\Phi(e) = \exists^{\geqslant N} e^\prime(\textnormal{\textsf{R}}(e^\prime,e)\wedge\operatorname{Dom}_d(e^\prime)\wedge\psi(e^\prime))$ is learnable by the $l$-th layer.
    The semantics of $\Phi(e)$ require counting neighbors $e^\prime$ connected by relation $\textnormal{\textsf{R}}$ that satisfy both the domain constraint $\operatorname{Dom}_d(\cdot)$ and the property $\psi$. We construct the $l$-th layer aggregation following Lemma C.3~\cite{qiu2024understanding}:

    \textbf{Message Construction.}
    Let the message function $\textsc{Msg}_{\textnormal{\textsf{R}}}^{(l)}(\cdot)$ select and combine the relevant features. Since $\operatorname{Dom}_d(e^\prime)$ is a unary predicate available in the input features (and preserved or propagated to $(l-1)$-th layer), let $\mathbf{Z}_{e^\prime}^{(l-1)}[k]$ correspond to $\operatorname{Dom}_d(e^\prime)$. We define the message to output a scalar value 1 only if both conditions are met, essentially implementing a logical ``AND'' ($\wedge$) via soft continuous approximation (\eg, using summing components or MLP approximation in the limit):
    \begin{equation}
        \mathbf{m}_{e', e} \doteq \mathbb{I}\big(\mathbf{Z}_{e^\prime}^{(l-1)}[k]=1 \wedge \mathbf{Z}_{e^\prime}^{(l-1)}[j]=1\big).
    \end{equation}
    Alternatively, in a standard linear GNN, $\mathbf{W}_{\textsc{Msg}}$ assigns a unit weight to the $\psi$ dimension conditioned on $e^\prime$ being masked or filtered by domain $d$. A more robust construction uses the universal approximation theorem of MLPs (contained in $\textsc{Msg}(\cdot)$) to map $\mathbf{Z}_{e^\prime}^{(l-1)}$ to 1 if the domain and $\psi$ bits are active, and 0 otherwise.

    \textbf{Aggregation.}
    The GFM sums the messages from neighbors:
    \begin{equation}
        S_e = \sum_{e^\prime \in \mathcal{N}_{\textnormal{\textsf{R}}}(e)} \mathbf{m}_{e^\prime, e} = \left|\{e^\prime \mid \textnormal{\textsf{R}}(e^\prime, e) \wedge  \operatorname{Dom}_d(e^\prime) \wedge  \psi(e')\}\right|.
    \end{equation}
    This sum $S_e$ exactly counts the number of valid logical witnesses in the specific domain.

    \textbf{Thresholding (Quantifier Realization).}
    To realize the quantifier $\exists^{\geqslant N}$, we set the bias term $\mathbf{b}^{(l)}$ and use a step-like activation function $\sigma$ (\eg, a steep Sigmoid). Assume $\mathbf{b} = -N + \delta$, then:
    \begin{equation}
        \mathbf{Z}_e^{(l)} = \sigma(S_e - N + \delta),\quad \delta>0.
    \end{equation}
    If $S_e \geqslant N$, the input to $\sigma$ is positive, outputting 1 (True). If $S_e < N$, the output is 0 (False).

    Thus, the GFM can structurally implement the formula $\Phi(e)$.

    \textbf{(2) Necessity (Expressivity Limit).}
    Deriving from Theorem C.2 and Theorem C.7 in~\citet{qiu2024understanding}, we prove any rule learnable by the query-conditioned GFM corresponds to a CML formula over $\mathcal{G}$. We invoke the established correspondence between MPNNs and the 1-dimensional Weisfeiler-Leman (1-WL) test~\cite{huang2023theory}. It is known that the discriminative power of standard MPNNs is bounded by 1-WL. Furthermore, the stable coloring of 1-WL on a graph is equivalent to the logical equivalence classes defined by CML with counting quantifiers~\cite{barcelo2020logical, morris2019weisfeiler}. In our specific case:

    \textbf{Graph Augmentation.}
    The multi-domain graph $\mathcal{G}$ can be viewed as a standard colored graph where the initial node colors $C(v)$ are tuples $(\mathbb{I}(v=\textnormal{$\mathbf{q}$}), \operatorname{Dom}_{d_1}(v), \cdots, \operatorname{Dom}_{d_k}(v), \cdots)$.

    \textbf{Bisimulation Invariance.}
    GFM is permutation invariant. If two entities $u$ and $v$ are indistinguishable by CML formulas (\ie, bisimilar on the augmented graph), they will receive identical message aggregates at every step and thus identical final embeddings.

    In conclusion, if the GFM learns a classifier that distinguishes entity $e$ based on $\textnormal{\textsf{R}}(\textnormal{$\mathbf{q}$}, e)$, there must exist a distinguishing structure in the computation tree. Since the computation tree is isomorphic to the unrolled logic formula in CML, the decision boundary learned by the GFM must be describable by a CML formula.
    Specifically, the ``multi-domain'' property simply expands the set of atomic unary predicates in the logic. If a rule relies on checking $\operatorname{Dom}_d$, this is valid in CML as checking an atomic proposition. If the rule cannot be expressed in CML (\eg, it requires counting cycles or checking connectivity beyond tree-unrolling), the 1-WL test would fail to distinguish the relevant subgraphs, and the GFM would consequently fail to learn the rule.
    Therefore, a rule is learnable only if it is expressible in CML with domain predicates.
\end{proof}

\subsection{Proof of Proposition~\ref{prop:error}}
\label{app:proof_error}
We first restate Proposition~\ref{prop:error} for reference.
\begin{mathbox}
\textsc{Proposition~\ref{prop:error}} (\textbf{\textsc{Error Bound of $\mathcal{L}_{\textnormal{LIB}}$ Approximation}})
\textit{Let $\mathcal{G}_\mathbf{q}$ be the subgraph, $y$ the ground-truth answer, and $\mathbf{q}$ the query. Assume the Markov property $\langle y \leftrightarrow \mathbf{q} \leftrightarrow \mathcal{G}_\mathbf{q} \rangle$ exists, the approximation error of $\mathcal{L}_{\textnormal{IB}}$ by $\mathcal{L}_{\textnormal{LIB}}$ is upper-bounded by the conditional entropy of $\mathbf{q}$ given $y$:
\begin{equation}
    \left | I(y;\mathcal{G}_\mathbf{q}) - I(\mathbf{q};\mathcal{G}_\mathbf{q}) \right| \leqslant H(\mathbf{q}\mid y),
\notag
\end{equation}
where $H(\mathbf{q}\mid y)$ is a data-dependent constant.}
\end{mathbox}

\begin{proof}
    We formally analyze the approximation error introduced by the label-free objective using the interaction information and the chain rule of mutual information. Let the random variables $y$, $\mathbf{q}$, and $\mathcal{G}_\mathbf{q}$ represent the ground-truth answer, the query, and the retrieved subgraph, respectively.
    Consider the joint mutual information $I(y, \mathbf{q}; \mathcal{G}_\mathbf{q})$. By symmetry, we can decompose this term in two equivalent ways based on the chain rule:
    \begin{align}
        &I(y,\mathbf{q};\mathcal{G}_\mathbf{q})=I(y;\mathcal{G}_\mathbf{q}) + I(\mathbf{q}\mid y;\mathcal{G}_\mathbf{q}),
        \label{eq:first}\\
        &I(y,\mathbf{q};\mathcal{G}_\mathbf{q})=I(\mathbf{q};\mathcal{G}_\mathbf{q}) + I(y\mid \mathbf{q};\mathcal{G}_\mathbf{q}).
        \label{eq:second}
    \end{align}
    Subtracting Eq.~\eqref{eq:second} from Eq.~\eqref{eq:first} allows us to isolate the discrepancy between the supervised objective $I(y;\mathcal{G}_\mathbf{q})$ and our proposed label-free proxy $I(\mathbf{q};\mathcal{G}_\mathbf{q})$:
    \begin{equation}
        I(y;\mathcal{G}_\mathbf{q})-I(\mathbf{q};\mathcal{G}_\mathbf{q})=I(y\mid \mathbf{q};\mathcal{G}_\mathbf{q})-I(\mathbf{q}\mid y;\mathcal{G}_\mathbf{q}).
    \label{eq:temp}
    \end{equation}
     To derive an upper bound, we recall the property that mutual information is bounded by the entropy of its variables. Specifically, for conditional mutual information, we have:
    \begin{equation}
        I(X; Y\mid Z) = H(X\mid Z) - H(X\mid Y,Z).
    \end{equation}
    Since entropy is non-negative and conditioning reduces entropy, or leaves it unchanged, we know that:
    \begin{equation}
        I(X; Y\mid Z) \leqslant H(Y\mid Z),
    \end{equation}
    where $H(\cdot)$ denotes the Shannon entropy. Applying this inequality to both terms on the RHS of Eq.~\eqref{eq:temp}:
    \begin{align}
        &I(y\mid \mathbf{q};\mathcal{G}_\mathbf{q}) \leqslant H(y\mid \mathbf{q}),
        \label{eq:temp1}\\
        &I(\mathbf{q}\mid y;\mathcal{G}_\mathbf{q}) \leqslant H(\mathbf{q} \mid y).
        \label{eq:temp2}
    \end{align}
    Substituting these bounds back into the absolute difference error:
    \begin{align}
        \left | I(y;\mathcal{G}_\mathbf{q}) - I(\mathbf{q};\mathcal{G}_\mathbf{q}) \right| & \leqslant \max \big( H(y\mid \mathbf{q}) + H(\mathbf{q}\mid y) \big) \notag \\
        &\leqslant H(y\mid \mathbf{q}) + H(\mathbf{q}\mid y).
    \end{align}
    In the context of retrieval where the query intends to specify the answer (\ie, $\mathbf{q}$ is not independent of $y$), if we assume the retrieved subgraph $\mathcal{G}_\mathbf{q}$ only contains information about $y$ via the semantics of $\mathbf{q}$ (Markov property $\langle y \!\leftrightarrow\! \mathbf{q} \!\leftrightarrow\! \mathcal{G}_\mathbf{q}\rangle$), the bound tightens to:
    \begin{equation}
        \left | I(y;\mathcal{G}_\mathbf{q}) - I(\mathbf{q};\mathcal{G}_\mathbf{q}) \right| \leqslant H(\mathbf{q}\mid y).
    \end{equation}
    While $H(\mathbf{q}\mid y)$ represents the intrinsic diversity of queries targeting the same answer, it is crucial to note that this term is determined solely by the data distribution $\mathbb{P}(\mathbf{q}, y)$ and remains constant with respect to the learnable parameters in the subgraph selector. Therefore, minimizing the upper bound is equivalent to minimizing the proxy gap up to a constant additive factor. The gradient directions for optimizing sufficiency are essentially unaffected by this term.

    In conclusion, Proposition~\ref{prop:error} suggests maximizing the mutual information with query $I(\mathcal{G}_\mathbf{q}; \mathbf{q})$ is an effective proxy for $I(\mathcal{G}_\mathbf{q}; y)$, provided that the query ambiguity $H(y\mid\mathcal{G}_\mathbf{q})$ is low. Since the query is designed to point towards the answer (low ambiguity), minimizing the label-free loss effectively minimizes the supervised loss within a controllable error margin.
\end{proof}

\subsection{Proof of Proposition~\ref{prop:lower}}
\label{app:proof_lower}
We first restate Proposition~\ref{prop:lower} for reference.
\begin{mathbox}
\textsc{Proposition~\ref{prop:lower}} (\textsc{\textbf{Lower Bound of $I(\mathbf{q};\mathcal{G}_\mathbf{q})$}})
\textit{Denote the average-pooled subgraph representation as $\mathbf{G}_{\mathbf{q}}=\sum_{e\in\mathcal{V}_\mathbf{q}} S_e(\mathbf{q})\cdot \widetilde{\mathbf{Z}}_e$. Deriving from~\textnormal{\citet{alemi2017deep}}, the lower bound is:
\begin{equation}
    I(\mathbf{q};\mathcal{G}_\mathbf{q}) \geqslant -\mathcal{L}_\textnormal{NCE}=\log\frac{\exp\big( \operatorname{sim}(\mathbf{G}_{\mathbf{q}},\mathbf{q})/\tau \big)}{\sum_{\mathbf{q}^\prime \in\mathcal{B}}\exp\big( \operatorname{sim}(\mathbf{G}_{\mathbf{q}^\prime},\mathbf{q}^\prime)/\tau \big)},
\notag
\end{equation}
where $(\mathbf{G}_{\mathbf{q}^\prime},\mathbf{q}^\prime)$ is the negative sample pair.}
\end{mathbox}

\begin{proof}
    The primary goal is to maximize the mutual information $I(\mathbf{q}; \mathcal{G}_\mathbf{q})$ to ensure the retrieved subgraph $\mathcal{G}_\mathbf{q}$ retains sufficient semantic information relevant to the query $\mathbf{q}$. Since the exact computation involves high-dimensional integration over the joint distribution, we derive a tractable lower bound via variational estimation.

    \textbf{Variational Lower Bound Derivation.}
    By definition, mutual information is the reduction in uncertainty of $\mathbf{q}$ given $\mathcal{G}_\mathbf{q}$: 
    \begin{equation}
        I(\mathbf{q}; \mathcal{G}_\mathbf{q}) = H(\mathbf{q}) - H(\mathbf{q} \mid \mathcal{G}_\mathbf{q}).
    \end{equation}
    Since the entropy of the query distribution $H(\mathbf{q})$ is constant with respect to the subgraph selector, maximizing mutual information is equivalent to minimizing conditional entropy $H(\mathbf{q} \mid \mathcal{G}_\mathbf{q})$.
    However, the true posterior distribution $\mathbb{P}(\mathbf{q} \mid \mathcal{G}_\mathbf{q})$ is unknown. We introduce a variational distribution $\mathbb{Q}(\mathbf{q} \mid \mathcal{G}_\mathbf{q})$ to approximate this posterior. Leveraging the non-negativity of the Kullback-Leibler divergence, $\operatorname{KL}(\mathbb{Q}(\mathbf{p}\mid\mathcal{G}_\mathbf{q})~||~\mathbb{Q}(\mathbf{q}\mid\mathcal{G}_\mathbf{q})) \geqslant 0$, we establish the following bound:
    \begin{align}
        I(\mathbf{q}; \mathcal{G}_\mathbf{q}) 
        &= \mathbb{E}_{\mathbb{P}(\mathbf{q}, \mathcal{G}_\mathbf{q})} \left[ \log \frac{\mathbb{P}(\mathbf{q} \mid \mathcal{G}_\mathbf{q})}{\mathbb{P}(\mathbf{q})} \right] \notag \\
        &= \mathbb{E}_{\mathbb{P}(\mathbf{q}, \mathcal{G}_\mathbf{q})} \left[ \log \mathbb{Q}(\mathbf{q} \mid \mathcal{G}_\mathbf{q}) \right] + \operatorname{KL}(\mathbb{P}~||~\mathbb{Q}) \notag \\
        &\geqslant \mathbb{E}_{\mathbb{P}(\mathbf{q}, \mathcal{G}_\mathbf{q})} \left[ \log \mathbb{Q}(\mathbf{q} \mid \mathcal{G}_\mathbf{q}) \right],
        \label{eq:var_bound}
    \end{align}
    which implies that maximizing the log-likelihood of the variational distribution is a valid proxy for mutual information.
    
    \textbf{Modeling Semantic Alignment.}
    To bridge the gap between the query and the graph structure, we model the variational distribution $\mathbb{Q}$ using an energy-based formulation. The probability is proportional to the semantic similarity between the pooled subgraph representation $\mathbf{G}_\mathbf{q}$ and the query embedding $\mathbf{q}$:
    \begin{equation}
        \mathbb{Q}(\mathbf{q} \mid \mathcal{G}_\mathbf{q}) = \frac{\exp(\operatorname{sim}(\mathbf{G}_\mathbf{q}, \mathbf{q}) / \tau)}{\Pi(\mathcal{G}_\mathbf{q})},
    \end{equation}
    where $\tau$ is a temperature hyper-parameter and $\Pi(\cdot)$ is the partition function. This formulation effectively encourages aligning the embedding spaces of relevant subgraphs and queries.
    
    \textbf{Tractable NCE Approximation.}
    Directly computing the partition $\Pi(\mathcal{G}_\mathbf{q})$ requires summing over the entire query space, which is computationally infeasible. Following the InfoNCE principle~\cite{oord2018representation, he2020momentum}, we approximate this term using a batch of negative samples $\mathcal{B}$. The learning problem is thus cast as a discrimination task: identifying the true query $\mathbf{q}$ from a set of distractors (negative samples):
    \begin{equation}
        \mathbb{E} \left[\log \mathbb{Q}(\mathbf{q} \mid \mathcal{G}_\mathbf{q}) 
        \right]\doteq \mathbb{E}_{\mathcal{B}} \left[ \log \frac{\exp(\operatorname{sim}(\mathbf{G}_\mathbf{q}, \mathbf{q}) / \tau)}{\sum_{\mathbf{q}' \in \mathcal{B}} \exp(\operatorname{sim}(\mathbf{G}_{\mathbf{q}'}, \mathbf{q}') / \tau)} \right] 
        = -\mathcal{L}_{\text{NCE}}.
    \end{equation}

    We conclude the proof.
\end{proof}

\subsection{Proof of Proposition~\ref{prop:upper}}
\label{app:proof_upper}
We first restate Proposition~\ref{prop:upper} for reference.
\begin{mathbox}
    \textsc{Proposition~\ref{prop:upper}} (\textbf{\textsc{Upper Bound of $I(\mathcal{G};\mathcal{G}_\mathbf{q})$}})
    \textit{Denote the subgraph size and connectivity penalizers as:
    \begin{equation}
        \mathcal{L}_\textnormal{size}=\sum\nolimits_{e\in\mathcal{V}}S_e(\mathbf{q}), \;\;\mathcal{L}_\textnormal{con}=\operatorname{Tr}\big( \mathbf{S}_e(\mathbf{q})^\top\mathbf{L}\mathbf{S}_e(\mathbf{q}) \big),
    \notag
    \end{equation}
    where $S_e(\mathbf{q})\in[\textnormal{0},\textnormal{1}]$ is entity selection probability, $\mathbf{S}_e(\mathbf{q}) \in \mathbb{R}^{|\mathcal{V}|}$ is the soft selection vector, $\mathbf{L}$ is the normalized graph Laplacian, and $\operatorname{Tr}(\cdot)$ calculates the trace of a matrix.
    The upper bound is:
    \begin{equation}
        I(\mathcal{G};\mathcal{G}_\mathbf{q}) \leqslant \beta_{\textnormal{1}} \mathcal{L}_\textnormal{size} + \beta_{\textnormal{2}}\mathcal{L}_\textnormal{con},
    \notag
    \end{equation}
    where $\beta_{\textnormal{1}}$ and $\beta_{\textnormal{2}}$ are hyper-parameters.}
\end{mathbox}

\begin{proof}
    This compression term prevents the subgraph selector from retrieving the entire graph and encourages the identification of a ``minimal'' subgraph. We prove that minimizing the two proposed structural penalizers is equivalent to minimizing an upper bound of this mutual information.

    \textbf{Variational Upper Bound Derivation.}
    Let $\mathbb{P}(\mathcal{G}_\mathbf{q} \mid \mathcal{G})$ denote the encoding distribution (\ie, the subgraph selection process). We construct an upper bound by introducing an arbitrary prior distribution $r(\mathcal{G}_\mathbf{q})$ over the space of possible subgraphs:
    \begin{align}
        I(\mathcal{G}; \mathcal{G}_\mathbf{q}) 
        &= \mathbb{E}_{\mathbb{P}(\mathcal{G})} \left[ \operatorname{KL}(\mathbb{P}(\mathcal{G}_\mathbf{q} \mid \mathcal{G})~||~\mathbb{P}(\mathcal{G}_\mathbf{q})) \right] \notag \\
        &\leqslant \mathbb{E}_{\mathbb{P}(\mathcal{G})} \left[ \operatorname{KL}(\mathbb{P}(\mathcal{G}_\mathbf{q} \mid \mathcal{G})~||~r(\mathcal{G}_\mathbf{q})) \right].
    \end{align}
    Minimizing KL divergence penalizes the learned subgraph distribution to conform to properties specified by the prior $r(\mathcal{G}_\mathbf{q})$.
    
    \textbf{Gibbs Prior for Graph Structure.}
    To inject the inductive biases that valid reasoning paths should be dense and connected, we explicitly define the prior $r(\mathcal{G}_\mathbf{q})$ as a Gibbs distribution (or Boltzmann distribution). The energy function $E(\mathcal{G}_\mathbf{q})$ encodes the cost of a subgraph configuration:
    \begin{align}
        & r(\mathcal{G}_\mathbf{q}) = \frac{1}{Z} \exp \left( -E(\mathcal{G}_\mathbf{q}) \right), \notag \\
        \text{where}\quad & E(\mathcal{G}_\mathbf{q}) = \beta_1 E_\text{size} + \beta_2 E_\text{con}.
    \end{align}
    The constant $Z$ denotes Zustandssumme (German for ``sum over states''). Its mathematical definition is the normalized sum of the exponential of the energy over all possible states:
    \begin{equation}
        Z = \sum_{\mathcal{G}' \in \Omega} \exp(-E(\mathcal{G}')).
    \end{equation}
    The specific forms of the energy terms are derived from two priors:

    \textbf{(1) Sparsity via Laplace Prior.}
    Assuming the edge selection follows a Laplace distribution centered at zero, its negative log-likelihood corresponds to the $L_1$-norm of the selection probabilities, yielding $E_\text{size} = \sum S_e(\mathbf{q})$, denoting as $\mathcal{L}_{\text{size}}$.
    
    \textbf{(2) Connectivity via GMRF Prior.}
    Assuming the selection variables follow a Gaussian Markov Random Field (GMRF) over the structure, the energy is defined by the graph Laplacian quadratic form, yielding $E_\text{con} = \operatorname{Tr}(\mathbf{S}_e(\mathbf{q})^\top\mathbf{L}\mathbf{S}_e(\mathbf{q}))$, denoting as $\mathcal{L}_{\text{con}}$.
    
    \textbf{Equivalence to Loss Minimization.}
    We substitute the Gibbs prior back into the KL divergence term. Note that the entropy of the encoder $H(\mathcal{G}_\mathbf{q} \mid \mathcal{G})$ is determined by the subgraph selector's confidence. Assuming a deterministic or low-variance selection process, this entropy term can be treated as a constant. The optimization objective simplifies as follows:
    \begin{align}
        \min \operatorname{KL}(\mathbb{P}(\cdot\mid\mathcal{G})~||~r(\cdot)) 
        &\Leftrightarrow \min \mathbb{E}_{\mathbb{P}(\mathcal{G}_\mathbf{q} \mid \mathcal{G})} \left[ -\log r(\mathcal{G}_\mathbf{q}) \right] \notag \\
        &\Leftrightarrow \min \mathbb{E}_{\mathbb{P}(\mathcal{G}_\mathbf{q} \mid \mathcal{G})} \left[ -\log \left( \frac{1}{Z} \exp(-E(\mathcal{G}_\mathbf{q})) \right) \right] \notag \\
        &\Leftrightarrow \min \mathbb{E}_{\mathbb{P}(\mathcal{G}_\mathbf{q} \mid \mathcal{G})} \left[ \beta_1 \mathcal{L}_{\text{size}} + \beta_2 \mathcal{L}_{\text{con}} + \log Z \right] \notag \\
        &\Leftrightarrow \min \left( \beta_1 \mathcal{L}_{\text{size}} + \beta_2 \mathcal{L}_{\text{con}} \right).
    \end{align}
    We conclude the proof.
\end{proof}
\refstepcounter{section}
\begin{table*}[!t]
  \centering
  \caption{Statistics of the general multi-hop QA datasets.}
  \vspace{-0.3cm}
  \resizebox{0.75\textwidth}{!}{
    \begin{tabular}{lrrrrrr}
    \toprule
    \textbf{Dataset} & $\#$ \textbf{Q-doc Pair} & $\#$ \textbf{Document} & $\#$ \textbf{KG} & $\#$ \textbf{Entity} & $\#$ \textbf{Relation} & $\#$ \textbf{Triple} \\
    \midrule
    \texttt{HotpotQA}~\cite{yang2018hotpotqa} & 20,000 & 204,822 & 20 & 1,930,362 & 967,218 & 6,393,342\\
    \texttt{MuSiQue}~\cite{trivedi2022musique} & 20,000 & 410,380 & 20 & 1,544,966 & 900,338 & 4,848,715 \\
    \texttt{2WikiMultiHopQA}~\cite{ho2020constructing} & 20,000 & 122,108 & 20 & 916,907 & 372,554 & 2,883,006 \\
    \midrule
    \textbf{Total} & 60,000 & 737,310 & 60 & 4,392,235 & 2,240,110 & 14,125,063 \\
    \bottomrule
    \end{tabular}%
    }
  \label{tab:dataset_general_qa}%
\end{table*}%
\begin{table*}[t]
\setlength{\tabcolsep}{5pt}
  \centering
  \caption{Statistics of the domain-specific QA datasets.}
  \vspace{-0.3cm}
  \resizebox{0.75\textwidth}{!}{
    \begin{tabular}{llrrrrr}
    \toprule
    \textbf{Dataset} & \textbf{Domain} & $\#$ \textbf{Test} & $\#$ \textbf{Document} & $\#$ \textbf{Entity} & $\#$ \textbf{Relation} & $\#$ \textbf{Triple} \\
    \midrule
    \texttt{HotpotQA}~\cite{yang2018hotpotqa}   & Multi-hop         & 1,000 & 9,221     & 87,768  & 45,112    & 279,112 \\
    \texttt{MuSiQue}~\cite{trivedi2022musique}    & Multi-hop         & 1,000 & 11,656    & 100,853 & 55,944    & 319,618 \\
    \texttt{2WikiMultiHopQA}~\cite{ho2020constructing}      & Multi-hop         & 1,000 & 6,119     & 48,779  & 20,748    & 160,950 \\
    \midrule
    \texttt{PubMedQA}~\cite{jin2019pubmedqa}   & Biomedical        & 2,450 & 5,932     & 42,389  & 20,952    & 149,782 \\
    \midrule
    \texttt{DelucionQA}~\cite{sadat2023delucionqa} & Customer Support  & 184   & 235       & 2,669   & 2,298     & 6,183   \\
    \texttt{TechQA}~\cite{castelli2020techqa}     & Customer Support  & 314   & 769       & 10,221  & 4,606     & 57,613  \\
    \texttt{ExpertQA}~\cite{malaviya2024expertqa}   & Customer Support  & 203   & 808       & 11,079  & 6,810     & 16,541  \\
    \texttt{EManual}~\cite{nandy2021question}    & Customer Support  & 132   & 102       & 695     & 586       & 1,329   \\
    \midrule
    \texttt{MS} \texttt{Marco}~\cite{nguyen2016ms}   & General Knowledge & 423   & 3,481     & 24,740  & 17,042    & 63,995  \\
    \texttt{HAGRID}~\cite{kamalloo2023hagrid}     & General Knowledge & 1,318 & 1,975     & 23,484  & 18,653    & 48,969  \\ 
    \bottomrule
    \end{tabular}%
    }
  \label{tab:dataset_domain_qa}%
\end{table*}%

\addtocounter{section}{-1}
\section{Experiment Details}
\label{app:exp_detail}
In this section, we provide additional experimental details.

\subsection{Datasets Details}
\label{app:exp_detail_data}
\subsubsection{General Multi-hop QA Datasets}
Following prior work~\cite{luo2025gfm}, we conduct experiments on \textbf{three} benchmark datasets for the multi-hop question answering task, namely \texttt{HotpotQA}~\cite{yang2018hotpotqa}, \texttt{MuSiQue}~\cite{trivedi2022musique}, and \texttt{2WikiMultiHopQA}~\cite{ho2020constructing}.
The statistics are listed in Table~\ref{tab:dataset_general_qa}. A brief description of each dataset is provided below.
\begin{itemize}[leftmargin=*]
    \item \texttt{HotpotQA}~\cite{yang2018hotpotqa}: a large-scale multi-hop question answering dataset comprising 97k examples. Each question is paired with up to two gold supporting documents and several distractors. Answering typically requires combining evidence from both documents, testing the ability to perform fine-grained reasoning.
    \item \texttt{MuSiQue}~\cite{trivedi2022musique}: contains 25k carefully curated multi-hop questions that require two to four reasoning steps. The questions are constructed to encourage compositional reasoning over multiple documents, making it a challenging benchmark for evaluating multi-step semantic understanding.
    \item \texttt{2WikiMultiHopQA}~\cite{ho2020constructing}: a Wikipedia-based multi-hop QA dataset with 192k questions. Each question is designed to require reasoning over two or four different Wikipedia articles. The dataset highlights diverse reasoning patterns across factual content, providing a comprehensive testbed.
\end{itemize}

\subsubsection{Domain-specific QA Datasets}
\label{app:domain_data}
To assess the generalization ability of \modelname~across diverse application scenarios, we further evaluate it on \textbf{seven} multi-hop QA benchmarks spanning multiple domains, including: \textbf{(1) Biomedical:} \texttt{PubMedQA}~\cite{jin2019pubmedqa}; \textbf{(2) Customer Support:} \texttt{DelucionQA}~\cite{sadat2023delucionqa}, \texttt{TechQA}~\cite{castelli2020techqa}, \texttt{ExpertQA}~\cite{malaviya2024expertqa}, and \texttt{EManual}~\cite{nandy2021question}; and \textbf{(3) General Knowledge:} \texttt{MS} \texttt{Marco}~\cite{nguyen2016ms} and \texttt{HAGRID}~\cite{kamalloo2023hagrid}. The statistics of the datasets are listed in Table~\ref{tab:dataset_domain_qa}.

\subsection{Baseline Details}
\label{app:exp_detail_baseline}
\begin{table*}[!t]
\renewcommand{\arraystretch}{0.9}
\setlength{\tabcolsep}{2.5pt}
  \centering
  \caption{Comparison of baselines. Filled ($\CIRCLE$), half-filled ($\RIGHTcircle$), and empty circles ($\ocircle$) indicate full, partial, and no support for each feature, respectively. \modelname~is the only method covering all the listed comparing dimensions.}
  \vspace{-0.3cm}
  \resizebox{\linewidth}{!}{
      \begin{tabular}{lccccccccccc}
    \toprule
    \multirow{2}[10]{*}{\textbf{Method}} & \multicolumn{6}{c}{\textbf{Retriever Design}} & \multicolumn{3}{c}{\textbf{Training and Generalization}} & \multicolumn{2}{c}{\textbf{Reasoning}} \\
    \cmidrule(l{0.5mm}r{1mm}){2-7}  \cmidrule(l{0.5mm}r{1mm}){8-10}  \cmidrule(l{0.5mm}){11-12}        & \makecell{external\\[-1.5pt]retrieve} & \makecell{dense\\[-1.5pt]neural\\[-1.5pt]retriver} & \makecell{graph-\\[-1.5pt]augmented\\[-1.5pt]retrieve} & \makecell{structured\\[-1.5pt]output} & \makecell{subgraph\\[-1.5pt]output} & \makecell{adaptive\\[-1.5pt]budget} & pre-train & \makecell{cross-\\[-1.5pt]domain} & fine-tune & \makecell{structure-\\[-1.5pt]aware\\[-1.5pt]prompt} & single-step \\[-1.5pt]
    \cmidrule(r{1mm}){1-1} \cmidrule(l{0.5mm}r{1mm}){2-7}  \cmidrule(l{0.5mm}r{1mm}){8-10}  \cmidrule(l{0.5mm}){11-12} 
    \texttt{BM25}  & $\CIRCLE$ & $\ocircle$ & $\ocircle$ & $\ocircle$ & $\ocircle$ & $\ocircle$ & $\ocircle$ & $\ocircle$ & $\ocircle$ & $\ocircle$ & $\ocircle$ \\
    \texttt{Contriever} & $\CIRCLE$ & $\CIRCLE$ & $\ocircle$ & $\ocircle$ & $\ocircle$ & $\ocircle$ & $\CIRCLE$ & $\RIGHTcircle$ & $\RIGHTcircle$ & $\ocircle$ & $\ocircle$ \\
    \texttt{GTR}   & $\CIRCLE$ & $\CIRCLE$ & $\ocircle$ & $\ocircle$ & $\ocircle$ & $\ocircle$ & $\CIRCLE$ & $\RIGHTcircle$ & $\CIRCLE$ & $\ocircle$ & $\ocircle$ \\
    \texttt{ColBERTv2} & $\CIRCLE$ & $\CIRCLE$ & $\ocircle$ & $\ocircle$ & $\ocircle$ & $\RIGHTcircle$ & $\ocircle$ & $\CIRCLE$ & $\CIRCLE$ & $\ocircle$ & $\ocircle$ \\
    \texttt{RAPTOR} & $\CIRCLE$ & $\CIRCLE$ & $\RIGHTcircle$ & $\CIRCLE$ & $\ocircle$ & $\RIGHTcircle$ & $\ocircle$ & $\ocircle$ & $\ocircle$ & $\RIGHTcircle$ & $\ocircle$ \\
    \texttt{Proposition} & $\CIRCLE$ & $\CIRCLE$ & $\CIRCLE$ & $\CIRCLE$ & $\CIRCLE$ & $\ocircle$ & $\ocircle$ & $\ocircle$ & $\ocircle$ & $\CIRCLE$ & $\ocircle$ \\
    \cmidrule(r{1mm}){1-1} \cmidrule(l{0.5mm}r{1mm}){2-7}  \cmidrule(l{0.5mm}r{1mm}){8-10}  \cmidrule(l{0.5mm}){11-12} 
    \texttt{GraphRAG} & $\CIRCLE$ & $\ocircle$ & $\CIRCLE$ & $\CIRCLE$ & $\ocircle$ & $\ocircle$ & $\ocircle$ & $\RIGHTcircle$ & $\ocircle$ & $\CIRCLE$ & $\CIRCLE$ \\
    \texttt{G-Retriever} & $\CIRCLE$ & $\CIRCLE$ & $\CIRCLE$ & $\CIRCLE$ & $\CIRCLE$ & $\RIGHTcircle$ & $\ocircle$ & $\RIGHTcircle$ & $\RIGHTcircle$ & $\CIRCLE$ & $\CIRCLE$ \\
    \texttt{LightRAG} & $\CIRCLE$ & $\CIRCLE$ & $\CIRCLE$ & $\CIRCLE$ & $\RIGHTcircle$ & $\RIGHTcircle$ & $\ocircle$ & $\RIGHTcircle$ & $\ocircle$ & $\ocircle$ & $\CIRCLE$ \\
    \texttt{HippoRAG} & $\CIRCLE$ & $\RIGHTcircle$ & $\CIRCLE$ & $\CIRCLE$ & $\RIGHTcircle$ & $\RIGHTcircle$ & $\ocircle$ & $\RIGHTcircle$ & $\ocircle$ & $\ocircle$ & $\CIRCLE$ \\
    \texttt{HippoRAG}~\texttt{2} & $\CIRCLE$ & $\RIGHTcircle$ & $\CIRCLE$ & $\CIRCLE$ & $\RIGHTcircle$ & $\RIGHTcircle$ & $\ocircle$ & $\RIGHTcircle$ & $\ocircle$ & $\ocircle$ & $\CIRCLE$ \\
    \texttt{SubgraphRAG} & $\CIRCLE$ & $\CIRCLE$ & $\CIRCLE$ & $\CIRCLE$ & $\CIRCLE$ & $\ocircle$ & $\ocircle$ & $\CIRCLE$ & $\RIGHTcircle$ & $\CIRCLE$ & $\CIRCLE$ \\
    \texttt{PropRAG} & $\CIRCLE$ & $\RIGHTcircle$ & $\CIRCLE$ & $\CIRCLE$ & $\CIRCLE$ & $\ocircle$ & $\ocircle$ & $\ocircle$ & $\ocircle$ & $\ocircle$ & $\CIRCLE$ \\
    \texttt{GFM-RAG} & $\CIRCLE$ & $\CIRCLE$ & $\CIRCLE$ & $\CIRCLE$ & $\ocircle$ & $\ocircle$ & $\CIRCLE$ & $\RIGHTcircle$ & $\CIRCLE$ & $\ocircle$ & $\CIRCLE$ \\
    \cmidrule(r{1mm}){1-1} \cmidrule(l{0.5mm}r{1mm}){2-7}  \cmidrule(l{0.5mm}r{1mm}){8-10}  \cmidrule(l{0.5mm}){11-12} 
    \texttt{FLARE} & $\CIRCLE$ & $\ocircle$ & $\ocircle$ & $\ocircle$ & $\ocircle$ & $\ocircle$ & $\ocircle$ & $\ocircle$ & $\ocircle$ & $\ocircle$ & $\ocircle$ \\
    \texttt{Adaptive-RAG} & $\CIRCLE$ & $\ocircle$ & $\ocircle$ & $\ocircle$ & $\ocircle$ & $\CIRCLE$ & $\ocircle$ & $\ocircle$ & $\RIGHTcircle$ & $\ocircle$ & $\ocircle$ \\
    \texttt{IRCoT} & $\CIRCLE$ & $\ocircle$ & $\ocircle$ & $\ocircle$ & $\ocircle$ & $\ocircle$ & $\ocircle$ & $\CIRCLE$ & $\ocircle$ & $\ocircle$ & $\ocircle$ \\
    \cmidrule(r{1mm}){1-1} \cmidrule(l{0.5mm}r{1mm}){2-7}  \cmidrule(l{0.5mm}r{1mm}){8-10}  \cmidrule(l{0.5mm}){11-12} 
    \textbf{\modelname} & $\CIRCLE$ & $\CIRCLE$ & $\CIRCLE$ & $\CIRCLE$ & $\CIRCLE$ & $\CIRCLE$ & $\CIRCLE$ & $\CIRCLE$ & $\CIRCLE$ & $\CIRCLE$ & $\CIRCLE$ \\
    \bottomrule
    \end{tabular}%
  }
  \label{tab:compare}%
\end{table*}%

We compare with \textbf{18} state-of-the-art baselines (and their combinations), categorized by \textbf{four} groups:
\textbf{(1) Base LLM:} \texttt{GPT-4o-mini}~\cite{hurst2024gpt};
\textbf{(2) Single-step RAGs:} including \texttt{BM25} \cite{robertson1994some}, \texttt{Contriever} \cite{izacard2022unsupervised}, \texttt{GTR} \cite{ni2022large}, \texttt{ColBERTv2} \cite{santhanam2022colbertv2}, \texttt{RAPTOR} \cite{sarthi2024raptor}, \texttt{Proposition} \cite{chen2024dense};
\textbf{(3) Graph-enhanced RAGs:} \texttt{GraphRAG} \cite{edge2024local}, \texttt{G-Retriever} \cite{he2024g}, \texttt{LightRAG} \cite{guo2024lightrag}, \texttt{HippoRAG} \cite{jimenez2024hipporag}, \texttt{HippoRAG} \texttt{2} \cite{gutierrez2025rag}, \texttt{SubgraphRAG} \cite{li2025simple}, \texttt{PropRAG} \cite{wang2025proprag}, \texttt{GFM-RAG} \cite{luo2025gfm};
\textbf{(4) Multi-step RAGs:} including \texttt{IRCoT} \cite{trivedi2023interleaving}, \texttt{FLARE} \cite{jiang2023active},  \texttt{Adaptive-RAG} \cite{jeong2024adaptive}.
A brief introduction is below.

\textbf{(1) Base LLM:}
uses a pre-trained large language model without fine-tuning or retrieval augmentation. It serves as the common generator for all retrieval-based methods, ensuring a fair comparison. We adopt \texttt{GPT-4o-mini}~\cite{hurst2024gpt}, which answers queries solely based on parametric knowledge.

\textbf{(2) Single-step RAGs:}
retrieve in a single pass, and are widely used due to their efficiency and general applicability. This category includes both sparse and dense retrievers:
\begin{itemize}[leftmargin=*]
    \item \texttt{BM25}~\cite{robertson1994some}: a classical probabilistic sparse retriever based on term frequency, document length normalization, and inverse document frequency, serving as a strong zero-shot baseline.
    \item \texttt{Contriever}~\cite{izacard2022unsupervised}: an unsupervised dense retriever trained via contrastive learning, exhibiting strong zero-shot and cross-domain generalization without labeled data.
    \item \texttt{GTR}~\cite{ni2022large}: a T5-based dual encoder that improves cross-domain retrieval by scaling model capacity with dot-product scoring, achieving high data efficiency.
    \item \texttt{ColBERTv2}~\cite{santhanam2022colbertv2}: a late-interaction retriever that performs token-level matching with compressed representations, balancing retrieval quality and memory efficiency.
    \item \texttt{RAPTOR}~\cite{sarthi2024raptor}: a hierarchical retrieval method that organizes documents into recursive summaries, enabling retrieval across multiple abstraction levels for complex queries.
    \item \texttt{Proposition}~\cite{chen2024dense}: a dense retriever that indexes text using fine-grained propositions as atomic units, improving retrieval precision and downstream QA under limited budgets.
\end{itemize}

\textbf{(3) Graph-enhanced RAGs:}
incorporate structured knowledge in the form of graphs to enhance retrieval and reasoning by explicitly modeling entities and relations, enabling multi-hop dependency modeling beyond flat text retrieval.
\begin{itemize}[leftmargin=*]
    \item \texttt{GraphRAG}~\cite{edge2024local}: constructs an entity graph from documents and generates community-level summaries with LLMs, enabling corpus-level question answering beyond standard retrieval.
    \item \texttt{G-Retriever}~\cite{he2024g}: formulates retrieval as a soft Prize-Collecting Steiner Tree problem to identify query-relevant subgraphs in large graphs, improving reasoning while reducing hallucination.
    \item \texttt{LightRAG}~\cite{guo2024lightrag}: integrates graph-aware indexing with flat text retrieval via a dual-level scheme, capturing both fine- and coarse-grained semantic dependencies efficiently.
    \item \texttt{HippoRAG}~\cite{jimenez2024hipporag}: combines knowledge graphs with Personalized PageRank to support efficient one-step retrieval inspired by hippocampal memory, with strong multi-hop QA performance.
    \item \texttt{HippoRAG}~\texttt{2}~\cite{gutierrez2025rag}: extends \texttt{HippoRAG} with deeper passage integration and improved LLM interaction, enhancing factual, associative, and sense-making memory.
    \item \texttt{SubgraphRAG}~\cite{li2025simple}: retrieves compact KG subgraphs using triple scoring and structural distance encoding, but relies on fixed heuristics that limit query-adaptive reasoning.
    \item \texttt{PropRAG}~\cite{wang2025proprag}: replaces triples with context-rich propositions and performs beam search over proposition paths, improving zero-shot multi-hop retrieval.
    \item \texttt{GFM-RAG}~\cite{luo2025gfm}: employs a graph foundation model as a retriever over large heterogeneous graphs, but lacks explicit cross-domain semantic alignment and does not fully exploit KG reasoning structure at inference.
\end{itemize}

\textbf{(4) Multi-step RAGs:}
iteratively interleave retrieval and reasoning to support complex multi-hop questions, gradually constructing reasoning chains. These methods are largely retriever-agnostic and can be combined with different RAG architectures.
\begin{itemize}[leftmargin=*]
    \item \texttt{IRCoT}~\cite{trivedi2023interleaving}: interleaves retrieval with Chain-of-Thought reasoning, where each reasoning step guides the next retrieval. This design reduces hallucination and improves multi-hop QA. Notably, \texttt{IRCoT} is retriever-agnostic and can be paired with existing RAG methods to enable multi-step reasoning.
    \item \texttt{FLARE}~\cite{jiang2023active}: performs forward-looking active retrieval during generation by anticipating low-confidence content and retrieving evidence to refine it, improving factual consistency in long-form generation.
    \item \texttt{Adaptive-RAG}~\cite{jeong2024adaptive}: dynamically selects between single-step and multi-step retrieval strategies based on predicted query complexity, balancing efficiency and reasoning capability across diverse queries.
\end{itemize}

\subsection{Experimental Setting Details}
\label{app:exp_detail_setting}
In this section, we introduce detailed experimental settings\footnote{Since standard benchmark datasets are adopted, and the same base LLM settings are used, the baseline results reported in this paper are taken directly from the original publications unless otherwise stated.}.

\textbf{Detailed Settings for Section~\ref{rq1:retrieve} (\textit{RQ1}).}
We evaluate entity- and document-level retrieval under a unified protocol. Baselines that output ranked lists are evaluated using Recall@2/5 on entities and documents (R@2/5$_\textsf{E}$, R@2/5$_\textsf{D}$). In contrast, \modelname~returns a query-conditioned subgraph $\mathcal{G}_{\mathbf{q}}$ and uses its node set $\mathcal{V}_{\mathbf{q}}$ as retrieved entities, where the retrieval size is determined by the learned selector rather than an explicit top-$K$ cutoff (we set $K=5$ for comparison).
Document retrieval is performed using the same entity-to-chunk inverted index across all methods. For multi-step baselines, we aggregate unique entities and documents across steps and evaluate Recall on the final ranked outputs. This setting decouples retrieval quality from downstream generation while maintaining consistent entity-document linking.

\textbf{Detailed Settings for Section~\ref{rq2:qa} (\textit{RQ2}).}
We evaluate end-to-end QA by feeding the retrieved evidence from each method into the same base LLM, \texttt{GPT-4o-mini}~\cite{hurst2024gpt}. Single-step RAGs use top retrieved document chunks under the same budget as \textit{\textbf{RQ1}}, while graph-enhanced methods serialize retrieved entities or subgraphs following their original designs. Multi-step methods such as \texttt{IRCoT} interleave retrieval and reasoning, and concatenate evidence across steps. For \modelname, the input is constructed from the query-specific minimal sufficient subgraph, including selected entities, relational paths, and linked document chunks. All methods share identical decoding settings and are evaluated without QA fine-tuning, using EM, F1, Precision, and Recall on test splits.

\textbf{Detailed Settings for Section~\ref{rq3:domain} (\textit{RQ3}).}
The retrieval performance is evaluated on all seven cross-domain QA datasets, while generative QA is reported on three datasets that provide free-form answer annotations. For retrieval, we use document-level Recall@5 (R@5$_\textsf{D}$), reflecting a practical retrieval budget and better capturing evidence coverage under domain shift. For generative QA, we adopt ROUGE-L instead of EM or F1, as cross-domain answers are often longer and paraphrased, where sequence-level overlap more reliably reflects semantic faithfulness.

\textbf{Detailed Settings for Section~\ref{rq4:efficiency} (\textit{RQ4}).}
We evaluate retriever-side efficiency by measuring end-to-end retrieval latency together with document-level Recall@5 (R@5$_\textsf{D}$). Latency is measured from query input to document scoring completion, excluding LLM generation. For multi-step methods (\eg, \texttt{IRCoT}-based), latency includes all iterative retrieval steps, while for graph-based methods, offline indexing is excluded and online traversal is counted. All results are averaged over 1,000 runs to ensure fair comparison across different retrieval paradigms.

\textbf{Detailed Settings for Section~\ref{rq5:case} (\textit{RQ5}).}
We conduct a case study on a multi-hop question from \texttt{HotpotQA}. The question asks for the government position held by the actress who portrayed ``Corliss Archer'' in the film ``Kiss and Tell''. The supporting evidence spans multiple documents, including the film ``Kiss and Tell (1945)'', the actress ``Shirley Temple'', and biographical descriptions of her later public service role. We apply \modelname~to this instance and visualize the query-conditioned subgraph produced by the subgraph selector. The visualization highlights the selected entities, relations, and explicit reasoning paths, and contrasts them with semantically related but query-irrelevant entities present in the original context. This case enables qualitative inspection of whether the retrieved subgraph is both minimal and sufficient for supporting multi-hop reasoning, with full metadata provided in Figure~\ref{fig:case_x}.

\textbf{Detailed Settings for Section~\ref{rq6:ablation} (\textit{RQ6}).}
For ablation studies, to analyze the contribution of the core components in \modelname, we construct \textbf{five} ablated variants (and their combinations) by removing the specific modules while keeping other settings unchanged. The variants are:
\begin{itemize}[leftmargin=*]
    \item \textbf{\modelname~(\textit{w/o Proto}):} Removes the prototype-based semantic alignment loss $\mathcal{L}_{\text{proto}}$ in Eq.~\eqref{eq:proto}, disabling cross-domain prototype supervision during retriever training.
    \item \textbf{\modelname~(\textit{w/o IGC}):} Removes the information gain contrast loss $\mathcal{L}_{\text{IGC}}$ in Eq.~\eqref{eq:igc}, preventing the retriever from being regularized toward informative semantic prototypes.
    \item \textbf{\modelname~(\textit{w/o Tune}):} Disables the fine-tuning stage defined in Eq.~\eqref{eq:ftn}, where the retriever is adapted to query-specific subgraph selection.
    \item \textbf{\modelname~(\textit{w/o SubG}):} Removes the entire subgraph selector module in Section~\ref{sec:method_subgraph}, reducing the model to entity-level retrieval without learning a minimal sufficient subgraph.
    \item \textbf{\modelname~(\textit{w/o Prompt}):} Removes the in-context relational prompting strategy described in Section~\ref{sec:method_prompt}. Retrieved entities and documents are still used, but relational paths are not explicitly reorganized into structured prompts for the LLM.
\end{itemize}

All variants are evaluated under the same settings.
For hyper-parameter sensitivity analysis, we conduct sensitivity analysis on \textbf{five} key weighting hyper-parameters:
\begin{itemize}[leftmargin=*]
    \item ${\alpha_{1}}$\textbf{:} controls weight of the binary cross-entropy retrieval $\mathcal{L}_{\text{BCE}}$ in Eq.~\eqref{eq:proto}, which helps relevance prediction for query-entity pairs.
    \item ${\alpha_{2}}$\textbf{:} controls the contribution of the ranking-based retrieval loss $\mathcal{L}_{\text{rank}}$ in Eq.~\eqref{eq:proto}, enforcing a margin between relevant and irrelevant entities to improve retrieval robustness.
    \item ${\alpha_{3}}$\textbf{:} weights the prototype-based semantic alignment loss $\mathcal{L}_{\text{proto}}$ in Eq.~\eqref{eq:proto}, regulating the strength of cross-domain semantic consistency during training.
    \item ${\beta_{1}}$\textbf{:} controls the minimality regularization term that penalizes the size of the selected subgraph in Eq.~\eqref{eq:ftn}.
    \item ${\beta_{2}}$\textbf{:} controls the connectivity regularization term that encourages structural coherence of the selected subgraph in Eq.~\eqref{eq:ftn}.
\end{itemize}
Each hyper-parameter is varied independently, and performance is reported using MRR on the validation set.

\begin{center}
\begin{tcolorbox}[
    float*=t,
    width=0.88\textwidth,
    enhanced,
    colback=white,
    colframe=black,
    boxrule=0.7pt,
    rounded corners,
    fonttitle=\bfseries,
    title=Metadata (Python Dict),
    boxsep=2pt,
    left=2pt,
    right=2pt,
    top=3pt,
    bottom=3pt,
    breakable
]
\begin{spacing}{0.88}
{
\small
\begin{Verbatim}[breaklines=true,breakanywhere=true]
{
  "_id": "5a8c7595554299585d9e36b6",
  "answer": "Chief of Protocol",
  "question": "What government position was held by the woman who portrayed Corliss Archer in the film Kiss and Tell?",
  "supporting_facts": [
    ["Kiss and Tell (1945 film)", 0],
    ["Shirley Temple", 0],
    ["Shirley Temple", 1],
  ],
  "context": [
    [
      "A Kiss for Corliss",
      [
        "A Kiss for Corliss is a 1949 American comedy film directed by Richard Wallace and written by Howard Dimsdale.",
        " It stars Shirley Temple in her final starring role as well as her final film appearance.",
        ' It is a sequel to the 1945 film "Kiss and Tell".',
        ' "A Kiss for Corliss" was retitled "Almost a Bride" before release and this title appears in the title sequence.',
        " The film was released on November 25, 1949, by United Artists.",
      ],
    ],
    [
      "Lord High Treasurer",
      [
        "The post of Lord High Treasurer or Lord Treasurer was an English government position and has been a British government position since the Acts of Union of 1707.",
        " A holder of the post would be the third-highest-ranked Great Officer of State, below the Lord High Steward and the Lord High Chancellor.",
      ],
    ],
    [
      "Meet Corliss Archer (TV series)",
      [
        "Meet Corliss Archer is an American television sitcom that aired on CBS (July 13, 1951 - August 10, 1951) and in syndication via the Ziv Company from April to December 1954.",
        " The program was an adaptation of the radio series of the same name, which was based on a series of short stories by F. Hugh Herbert.",
      ],
    ],
    [
      "Village accountant",
      [
        'The Village Accountant (variously known as "Patwari", "Talati", "Patel", "Karnam", "Adhikari", "Shanbogaru","Patnaik" etc.) is an administrative government position found in rural parts of the Indian sub-continent.',
        ' The office and the officeholder are called the "patwari" in Telangana, Bengal, North India and in Pakistan while in Sindh it is called "tapedar".',
        ' The position is known as the "karnam" in Andhra Pradesh, "patnaik" in Orissa or "adhikari" in Tamil Nadu, while it is commonly known as the "talati" in Karnataka, Gujarat and Maharashtra.',
        ' The position was known as the "kulkarni" in Northern Karnataka and Maharashtra.',
        ' The position was known as the "shanbogaru" in South Karnataka.',
      ],
    ],
    [
      "Joseph Kalite",
      [
        "Joseph Kalite (died 24 January 2014) was a Central African politician.",
        " As a government minister he either held the housing or health portfolio.",
        " Kalite, a Muslim, was reported to be killed by anti-balaka outside the Central Mosque in the capital Bangui during the Central African Republic conflict.",
        " He was killed with machetes on the day in Bangui after interim president Catherine Samba-Panza took power.",
        " At the time of the attack Kalite held no government position, nor did he under the Séléka rule.",
        " He was reported to have supported the rule of Séléka leader Michel Djotodia.",
      ],
    ],
    [
      "Charles Craft",
      [
        "Charles Craft (May 9, 1902 – September 19, 1968) was an English-born American film and television editor.",
        " Born in the county of Hampshire in England on May 9, 1902, Craft would enter the film industry in Hollywood in 1927.",
        ' The first film he edited was the Universal Pictures silent film, "Painting the Town".',
        " Over the next 25 years, Craft would edit 90 feature-length films.",
        ' In the early 1950s he would switch his focus to the small screen, his first show being "Racket Squad", from 1951–53, for which he was the main editor, editing 93 of the 98 episodes.',
        ' He would work on several other series during the 1950s, including "Meet Corliss Archer" (1954), "Science Fiction Theatre" (1955–56), and "Highway Patrol" (1955–57).',
        ' In the late 1950s and early 1960s he was one of the main editors on "Sea Hunt", starring Lloyd Bridges, editing over half of the episodes.',
        ' His final film work would be editing "Flipper\'s New Adventure" (1964, the sequel to 1963\'s "Flipper".',
        " When the film was made into a television series, Craft would begin the editing duties on that show, editing the first 28 episodes before he retired in 1966.",
        " Craft died on September 19, 1968 in Los Angeles, California.",
      ],
    ],
    [
      "Meet Corliss Archer",
      [
        "Meet Corliss Archer, a program from radio's Golden Age, ran from January 7, 1943 to September 30, 1956.",
        ' Although it was CBS\'s answer to NBC\'s popular "A Date with Judy", it was also broadcast by NBC in 1948 as a summer replacement for "The Bob Hope Show".',
        " From October 3, 1952 to June 26, 1953, it aired on ABC, finally returning to CBS.",
        " Despite the program's long run, fewer than 24 episodes are known to exist.",
      ],
    ],
    [
      "Janet Waldo",
      [
        "Janet Marie Waldo (February 4, 1920 – June 12, 2016) was an American radio and voice actress.",
        ' She is best known in animation for voicing Judy Jetson, Nancy in "Shazzan", Penelope Pitstop, and Josie in "Josie and the Pussycats", and on radio as the title character in "Meet Corliss Archer".',
      ],
    ],
    [
      "Kiss and Tell (1945 film)",
      [
        "Kiss and Tell is a 1945 American comedy film starring then 17-year-old Shirley Temple as Corliss Archer.",
        " In the film, two teenage girls cause their respective parents much concern when they start to become interested in boys.",
        " The parents' bickering about which girl is the worse influence causes more problems than it solves.",
      ],
    ],
    [
      "Secretary of State for Constitutional Affairs",
      [
        "The office of Secretary of State for Constitutional Affairs was a British Government position, created in 2003.",
        " Certain functions of the Lord Chancellor which related to the Lord Chancellor's Department were transferred to the Secretary of State.",
        " At a later date further functions were also transferred to the Secretary of State for Constitutional Affairs from the First Secretary of State, a position within the government held by the Deputy Prime Minister.",
      ],
    ],
  ],
  "type": "bridge",
  "level": "hard",
}
\end{Verbatim}
}
\end{spacing}
\captionof{figure}{Metadata of the case study from \texttt{HotpotQA}.}
\label{fig:case_x}
\vspace{-0.75cm}
\end{tcolorbox}
\end{center}
\section{Implementation Details}
\label{app:imp_detail}
In this section, we provide the implementation details.

\begin{table}[!t]
  \centering
  \caption{Hyper-parameter configurations.}
  \vspace{-0.3cm}
  \resizebox{\columnwidth}{!}{
    \begin{tabular}{lll}
    \toprule
    \textbf{Stage} & \textbf{Hyper-parameter} & \textbf{Value} \\
    \midrule
    \multirow{4}[0]{*}{\makecell[l]{KG Index\\Construction}} & OpenIE model & \texttt{GPT-4o-mini} \\
     & Entity resolution model & \texttt{ColBERTv2} \\
     & Text embedding dimension & 768 \\
     & Similarity threshold $\tau$ & 0.8 \\
    \midrule
    \multirow{6}[0]{*}{\makecell[l]{GFM\\Architecture}} & Number of layers & 6\\
     & Hidden dimension & 512 \\
     & Sentence encoder & \texttt{all-mpnet-v2} \\
     & Message function $\textsc{Msg}(\cdot)$ & $\operatorname{DistMult}(\cdot)$ \\
     & Aggregation function $\textsc{Agg}(\cdot)$ & $\textsc{Sum}(\cdot)$ \\
     & Relation update $g(\cdot)$ & 2-layer MLP\\
    \midrule
    \multirow{6}[0]{*}{\makecell[l]{Pre-training\\(Phase I)}} & Loss weight $\alpha_1$ & 0.3 \\
          & Optimizer & AdamW \\
          & Learning rate & 5e-4 \\
          & Batch size & 4 \\
          & Number of training steps & 30,000 \\
          & Number of negative samples & 128 \\
    \midrule
    \multirow{7}[0]{*}{\makecell[l]{Pre-training\\(Phase II)}} & Loss weight $\alpha_2$ & 0.1 \\
          & Loss weight $\alpha_3$ & 0.1 \\
          & Temperature $\tau$& 0.1 \\
          & Optimizer & AdamW \\
          & Learning rate & 5e-4 \\
          & Batch size & 4 \\
          & Epochs of continual training & 20 \\
    \midrule
    \multirow{11}[0]{*}{Fine-tuning} & Text embedding dimension & 768 \\
          & Hidden dimension & 1024 \\
          & Query prompt dimension & 512 \\
          & Subgraph selector dimension & 256 \\
          & Regularization weight $\beta_1$ & 0.01 \\
          & Regularization weight $\beta_2$ & 0.1 \\
          & Gumbel-Sigmoid temperature $\tau$ & 0.5 \\
          & InfoNCE temperature $\tau$ & 0.07 \\
          & Optimizer & AdamW \\
          & Learning rate & 5e-4 \\
          & Batch size & 8 \\
          & Number of fine-tuning epochs & 20 \\
    \midrule
    \multirow{3}[0]{*}{Inference (QA)} & Number of top-$K$ documents & 5 \\
          & Maximum DFS depth $l$ & 3 \\
          & Random seed & 1024 \\
    \bottomrule
    \end{tabular}%
  }
  \label{tab:para}%
\end{table}%

\subsection{Implementation Details of \modelname}
\label{app:imp_detail_ours}
\subsubsection{Pre-training Data}
We follow the pre-training data processing pipeline released by~\citet{luo2025gfm} to prepare large-scale pre-training data.
Specifically, a total of 60K pre-training instances are sampled from the training splits of \texttt{HotpotQA}, \texttt{MuSiQue}, and \texttt{2WikiMultiHopQA}, and merge their associated candidate passages to form the document corpus.
To build KG, the OpenIE-based extraction strategy is adopted, using \texttt{GPT-4o-mini} with prompts described in~\citet{jimenez2024hipporag} to extract entities, relations, and triples from the corpus.
Entity resolution is then performed by computing pairwise embedding similarities $s(e_i, e_j)$ with \texttt{ColBERTv2}~\cite{santhanam2022colbertv2}:
\begin{equation}
    s(e_i, e_j) = \textsc{TextEmb}(e_i)^\top \textsc{TextEmb}(e_j),
\label{eq:similarity}
\end{equation}
where a new equivalence relation $(e_i, \equiv, e_j)$ is added when their embedding similarity exceeds a threshold $\tau=0.8$.
To control the scale of each index, the sampled data are partitioned into subsets of approximately 1K questions and 10K documents.
This results in 60 knowledge graph indexes paired with corresponding question-document sets, which are used for large-scale pre-training.

\subsubsection{Model Configurations}
\modelname~is implemented as a query-conditioned GNN with 6 message-passing layers with a hidden size of 512. Entity and relation embeddings are initialized from the frozen \texttt{all-mpnet-v2} sentence encoder~\cite{allmpnet2021sbert} in 768 dimensions. Each layer performs relation-aware message passing using a bilinear interaction function followed by sum aggregation, while relation states are refined through a lightweight 2-layer MLP to capture context-dependent transformations. The pre-trained \modelname~contains approximately 8.1M trainable parameters and serves as a generalized and cross-domain retriever that produces query-conditioned entity embeddings for subgraph selection. For document ranking, we return the top-5 document corpus.

\subsubsection{Hyper-paramters}
We provide hyper-parameter configurations of pre-training, fine-tuning, and inference stages in Table~\ref{tab:para}. 

\subsection{Implementation Details of Baselines}
We provide the baseline methods implementations as follows.
\begin{itemize}[leftmargin=*]
    \item \texttt{GPT-4o-mini}: \url{https://platform.openai.com}.
    \item \texttt{BM25}: \url{https://github.com/jxmorris12/bm25_pt}.
    \item \texttt{Contriever}: \url{https://github.com/facebookresearch/contriever}.
    \item \texttt{GTR}: \url{https://github.com/google-research/t5x_retrieval#dense-retrieval-models}.
    \item \texttt{ColBERTv2}: \url{https://github.com/stanford-futuredata/ColBERT}.
    \item \texttt{RAPTOR}: \url{https://github.com/parthsarthi03/raptor}.
    \item \texttt{Proposition}: \url{https://github.com/chentong0/factoid-wiki}.
    \item \texttt{GraphRAG}: \url{https://github.com/microsoft/graphrag}.
    \item \texttt{G-Retriever}: \url{https://github.com/XiaoxinHe/G-Retriever}.
    \item \texttt{LightRAG}: \url{https://github.com/HKUDS/LightRAG}.
    \item \texttt{HippoRAG}: \url{https://github.com/OSU-NLP-Group/HippoRAG}.
    \item \texttt{HippoRAG}~\texttt{2}: \url{https://github.com/OSU-NLP-Group/HippoRAG}.
    \item \texttt{SubgraphRAG}: \url{https://github.com/Graph-COM/SubgraphRAG}.
    \item \texttt{PropRAG}: \url{https://github.com/ReLink-Inc/PropRAG}.
    \item \texttt{GFM-RAG}: \url{https://github.com/RManLuo/gfm-rag}.
    \item \texttt{IRCoT}: \url{https://github.com/stonybrooknlp/ircot}.
    \item \texttt{FLARE}: \url{https://github.com/jzbjyb/FLARE}.
    \item \texttt{Adaptive-RAG}: \url{https://github.com/starsuzi/Adaptive-RAG}.
\end{itemize}

For all baselines, if the same standard benchmark datasets and the same base LLM settings are used, the results reported in this paper are taken directly from the original
publications. When such settings are unavailable, we adopt the hyper-parameter configurations suggested in the original publications or provided by official implementations to ensure reasonable baseline performance.
To maintain a fair and controlled comparison, we keep other factors such as hidden size, number of layers, and architectural settings consistent with those used in our proposed framework.

\subsection{Hardware and Software Configurations}
\label{app:config}
We conduct the experiments in a Linux-based server environment with the following hardware and software configurations.
\begin{itemize}[leftmargin=*]
    \item \textbf{Operating System:} Ubuntu 20.04 LTS.
    \item \textbf{CPU:} 10 Intel(R) Xeon(R) Platinum 8358 CPU @ 2.60GHz with 1TB DDR4 of Memory.
    \item \textbf{GPU:} 8 NVIDIA Tesla A100 SMX4 with 80GB of Memory.
    \item \textbf{Software:} CUDA 12.4, Python 3.8.12, PyTorch\footnote{\url{https://github.com/pytorch/pytorch}.} 1.9.1, PyTorch Geometric\footnote{\url{https://github.com/pyg-team/pytorch_geometric}.} 2.0.1.
\end{itemize}

\section{Additional Results}
\label{app:additional_res}
\subsection{Impact of Domain-Specific Fine-Tuning}
\label{app:additional_res_1}
To further investigate and compare the cross-domain adaptability of retrievers, we fine-tune the most relevant baseline \texttt{GFM-RAG} and our \modelname~on the training split of its corresponding target dataset and evaluate on the same \textbf{seven} cross-domain benchmarks used in Section~\ref{rq3:domain} (\textit{\textbf{RQ3}}). Table ~\ref{tab:add_res_retrieve} reports document retrieval performance using R@5$_\textsf{D}$, and Table~\ref{tab:add_res_qa} reports generative QA performance using ROUGE-L on the three datasets that support free-form answers. Results marked with ``*'' denote performance after domain-specific fine-tuning, while unmarked results correspond to the zero-shot setting in consistency with that in \textit{\textbf{RQ3}}.

Results show that domain-specific fine-tuning yields clear improvements for both \texttt{GFM-RAG} and \modelname~across retrieval and generative QA, indicating that both benefit from in-domain supervision. However, \modelname~consistently remains stronger than \texttt{GFM-RAG} in both settings. Specifically, for cross-domain document retrieval, \modelname~outperforms \texttt{GFM-RAG} before fine-tuning and retains this advantage after fine-tuning. The same trend holds for the generative QA, where \modelname~achieves higher ROUGE-L both zero-shot and fine-tuned. Overall, the results suggest that \modelname~not only generalizes better for cross-domain retrieval, but also provides a consistently stronger adaptation starting point under domain-specific fine-tuning.
\begin{table}[!t]
\setlength{\tabcolsep}{3.5pt}
  \centering
  \caption{Impact of domain-specific fine-tuning on cross-domain document retrieval (R@5$_\textsf{D}$ \%). Results marked with ``*'' denote the fine-tuned performance.}
  \vspace{-0.3cm}
  \resizebox{\columnwidth}{!}{
    \begin{tabular}{lcccc}
    \toprule
    \textbf{Dataset} & \texttt{PropRAG} & \texttt{HippoRAG}~\texttt{2} & \texttt{GFM-RAG} & \textbf{\modelname} \\
    \midrule
    \texttt{PubMedQA} & 55.2  & 60.1  & 58.5 ~~\;\; \underline{63.6}* & 61.9 ~~\;\; \textbf{64.2}* \\
    \texttt{DelucionQA} & 65.9  & 68.2  & 70.8 ~~\;\; \underline{82.7}* & 72.3 ~~\;\; \textbf{84.3}* \\
    \texttt{TechQA} & 44.1  & 47.5  & 46.6 ~~\;\; 49.5* & \underline{50.3} ~~\;\; \textbf{51.6}* \\
    \texttt{ExpertQA} & 60.9  & 63.5  & 62.7 ~~\;\; 60.8* & \underline{65.3} ~~\;\; \textbf{65.9}* \\
    \texttt{Emanual} & 58.3  & 56.2  & 60.6 ~~\;\; \textbf{75.9}* & 61.9 ~~\;\; \underline{71.7}* \\
    \texttt{MS}~\texttt{Marco} & 65.2  & 68.4  & 71.0 ~~\;\; \underline{77.5}* & 71.6 ~~\;\; \textbf{78.0}* \\
    \texttt{HAGRID} & 78.8  & 81.6  & 84.7 ~~\;\; \underline{86.6}* & 86.2 ~~\;\; \textbf{87.5}* \\
    \bottomrule
    \end{tabular}%
  }
  \label{tab:add_res_retrieve}%
\end{table}%

\begin{table}[!t]
\setlength{\tabcolsep}{3.5pt}
  \centering
  \caption{Impact of domain-specific fine-tuning on cross-domain generative QA (ROUGE-L  \%). Results marked with * denote the fine-tuned performance.}
  \vspace{-0.3cm}
  \resizebox{\columnwidth}{!}{
    \begin{tabular}{lcccc}
    \toprule
    \textbf{Dataset} & \texttt{PropRAG} & \texttt{HippoRAG}~\texttt{2} & \texttt{GFM-RAG} & \textbf{\modelname} \\
    \midrule
    \texttt{ExpertQA}\textcolor{white}{on} & 27.1  & 29.0    & 27.6 ~~\;\; 27.3* & \underline{31.9} ~~\;\; \textbf{32.3}* \\
    \texttt{Emanual} & 43.3  & 46.1  & 47.6 ~~\;\; \underline{50.9}* & 49.8 ~~\;\; \textbf{51.6}* \\
    \texttt{MS Marco} & 33.8  & 33.1  & 30.2 ~~\;\; \underline{36.5}* & 35.6 ~~\;\; \textbf{37.1}* \\
    \bottomrule
    \end{tabular}%
  }
  \label{tab:add_res_qa}%
\end{table}%

\begin{figure*}[t]
    \centering
    \includegraphics[width=1\linewidth]{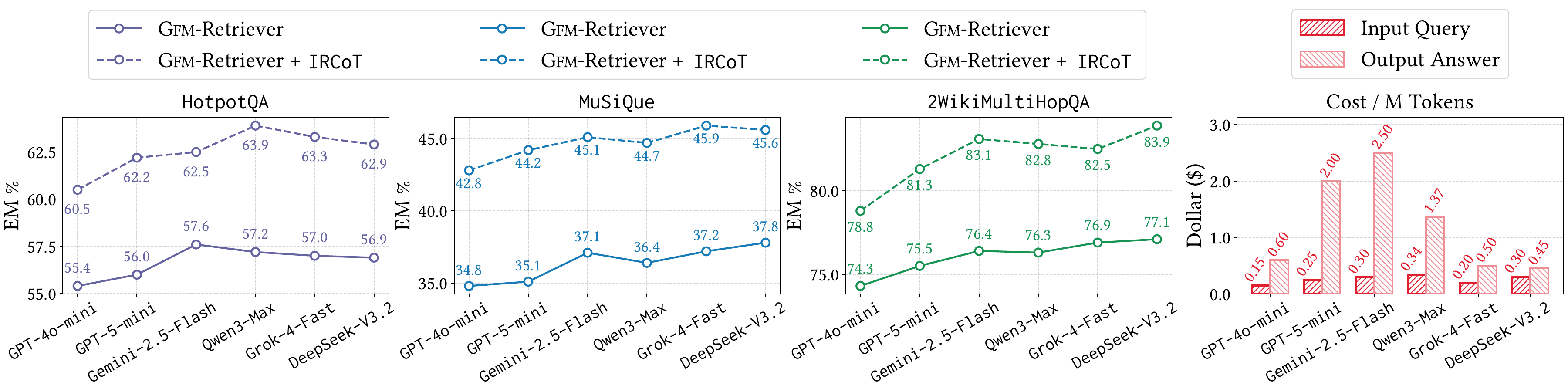}
    \vspace{-0.5cm}
    \caption{QA performance of \modelname~with different base LLMs.}
    \label{fig:llm}
\end{figure*}
\begin{figure*}[!t]
    \centering
    \includegraphics[width=1\linewidth]{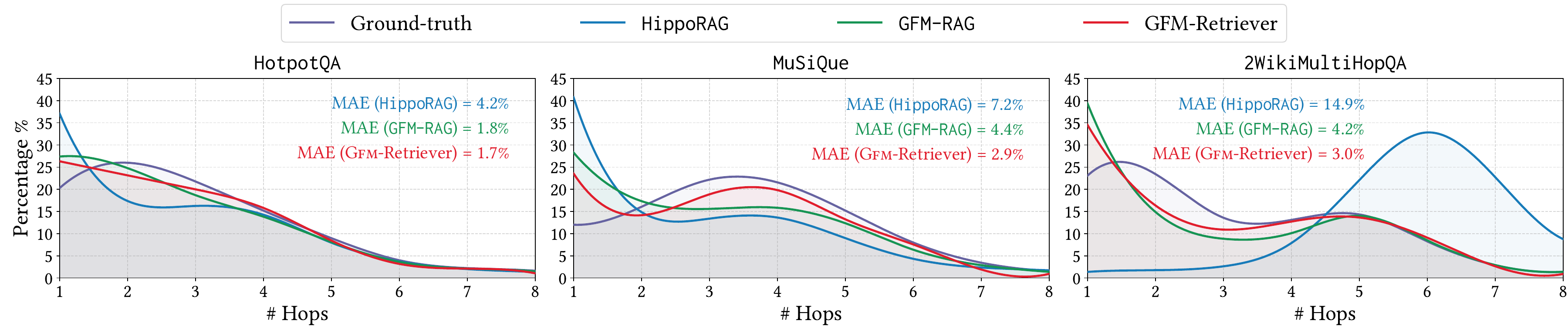}
    \vspace{-0.5cm}
    \caption{Distribution of predicted reasoning hop counts.}
    \label{fig:hop}
\end{figure*}

\subsubsection{Impact of Different Base LLMs}
To evaluate the modularity of \modelname~with respect to the underlying generator, we compare its QA performance when paired with different base LLMs. Since our retriever is LLM-agnostic, we keep the retrieved subgraphs fixed and only replace the generation LLMs. We test \textbf{six} representative LLMs with varying capacities and costs: \texttt{GPT-4o-mini}, \texttt{GPT-5-mini}, \texttt{Gemini-2.5-Flash}, \texttt{Qwen3-Max}, \texttt{Grok-4-Fast}, and \texttt{DeepSeek-V3.2}. We report Exact Match (EM \%) on \textbf{three} multi-hop QA benchmarks, with and without integrating the multi-step reasoning framework \texttt{IRCoT}. Additionally, we compare their token-level monetary cost for both input and output tokens.

Figure~\ref{fig:llm} shows that stronger LLMs consistently improve answer generation while the retrieval component remains unchanged. For all three datasets, EM steadily increases as model capability grows, demonstrating that better generators can more effectively exploit the structured evidence provided by \modelname. When combined with \texttt{IRCoT}, the performance gap between LLMs further widens, indicating that multi-step reasoning benefits more from stronger base LLMs. Notably, even lightweight models achieve competitive performance when guided by the structured subgraph evidence, highlighting the robustness of the retrieval stage. The cost analysis reveals a clear quality-cost trade-off: while larger models such as \texttt{Qwen3-Max} yield the best accuracy, smaller models (\eg, \texttt{GPT-4o-mini}) offer a favorable efficiency-performance balance. Overall, the results confirm that \modelname~is fully pluggable with diverse LLM backbones, and improvements in generation quality can be directly translated into better end-to-end QA without modifying the retriever.

\subsubsection{Distribution of Predicted Reasoning Hops}
To further examine whether the retrieved structures reflect the intrinsic multi-hop complexity of questions, we analyze the distribution of reasoning hop counts implied by different methods. For each question in the test sets of \texttt{HotpotQA}, \texttt{MuSiQue}, and \texttt{2WikiMultiHopQA}, we estimate the number of hops in the ground-truth reasoning chain based on annotated supporting facts. We then compute the hop count of the reasoning paths induced by different retrieval methods and compare their distributions against the ground-truth.

Figure~\ref{fig:hop} shows that \modelname~behaves consistently closer to the ground-truth than those of \texttt{HippoRAG} and \texttt{GFM-RAG}, which is reflected by the lowest mean absolute error (MAE) across all three datasets. In contrast, \texttt{HippoRAG} exhibits a clear bias toward longer reasoning chains, especially on \texttt{2WikiMultiHopQA}, where it overestimates high-hop paths and deviates substantially from the ground-truth distribution. While \texttt{GFM-RAG} better tracks the general trend, its hop distribution remains less aligned than that of \modelname.
The results indicate that \modelname~not only retrieves relevant evidence but also captures the sufficient structural depth required for each query. Instead of relying on shallow or deep reasoning paths, it tends to recover subgraphs that match the reasoning requirements, further supporting the subgraph selector effectively models minimal yet sufficient multi-hop structures.
\section{Discussions}
\label{app:discuss}
In this section, we address these \textit{\textbf{``Key Questions''}} to provide deeper insights into the motivation, theories, and design choices.

\paragraph{\textbf{Q1. Why can \modelname~be regarded as a foundation model rather than a task-specific retriever?}}
\modelname~is pre-trained on large-scale, heterogeneous knowledge graphs constructed from multi-domain corpora, rather than on a single downstream task. It learns generalizable structural reasoning patterns through KG completion, then transfers to structurally different tasks, including subgraph selection and document retrieval under domain shift. Proposition~\ref{prop:expressivity} further shows that a query-conditioned GFM can express graded modal logic over multi-domain KGs, providing formal justification that the model captures domain-aware multi-hop reasoning rules instead of dataset-specific heuristics.

\paragraph{\textbf{Q2. What is the source of the retriever's transferability across domains?}}
Transferability mainly stems from the Prototype-driven Continuous Pre-training (Phase II).
Standard GNNs often overfit to domain-specific structural patterns. In contrast, Phase II training explicitly aligns entity embeddings with domain semantic prototypes ($\mathbf{c}_d$ in Eq.~\eqref{eq:proto}), which regularizes the retriever to learn general domain knowledge while distinguishing domain-invariant semantic differences. The IGC regularizer further suppresses spurious correlations. This ensures the message-passing mechanism captures fundamental reasoning patterns (\eg, transitivity, composition) that hold true regardless of the specific domain.

\paragraph{\textbf{Q3. Why is the Information Bottleneck (IB) formulation termed "label-free"?}}
The term ``label-free'' refers to the unavailability of the ground-truth answer $y$ during the retrieval optimization phase.
We therefore propose a surrogate objective by replacing $y$ with the query $\mathbf{q}$, optimizing $I(\mathbf{q}; \mathcal{G}_\mathbf{q})$ instead. Proposition~\ref{prop:error} theoretically guarantees that the error introduced by this approximation is bounded by the conditional entropy $H(\mathbf{q}\mid y)$, which allows us to optimize the subgraph selector with only query-data pairs.

\paragraph{\textbf{Q4. Why employ Parameter-Efficient Fine-Tuning (PEFT) instead of full fine-tuning?}}
\textbf{(1)} Preserving Generalization: Full fine-tuning on a specific dataset risks catastrophic forgetting. By freezing the GFM backbone and only training the lightweight projection matrices and subgraph selector, we retain the cross-domain reasoning capabilities while adapting to the specific format of the retrieval task.
\textbf{(2)} Task Discrepancy Management: The pre-training task (KG completion) and downstream task (subgraph selection) differ. PEFT acts as a bridge, optimizing the decision boundary without distorting the semantic space that encodes entities.

\paragraph{\textbf{Q5. What is the theoretical performance bottleneck, and how is it mitigated?}}
Theoretically, the bottleneck lies in inference, specifically the path extraction process. While subgraph selection is efficient ($\mathcal{O}(|\mathcal{V}^T|)$), enumerating reasoning paths via DFS has a worst-case exponential complexity $\mathcal{O}(\Delta^l)$, where $\Delta$ is the node degree and $l$ is its depth. We address this by implementing a truncated DFS with a strict hop limit ($l\leqslant 3$) and beam search scoring.

\paragraph{\textbf{Q6. Is the performance gain due to the Subgraph Selector or the In-context Prompter?}}
Both components are essential. The ablation study (Figure~\ref{fig:ablation}) disentangles their contributions.
Removing the subgraph selector (reverting to standard retrieval) causes the most significant drop in retrieval, which proves that the content quality is primarily driven by the retriever's selection capability.
Removing the in-context prompter causes a notable drop in QA but has little effect on retrieval, which indicates that while the selector finds the ``ingredients'', the prompter provides the ``recipe'', explicitly organizing structures to guide the LLM's reasoning process.

\paragraph{\textbf{Q7. How does \modelname~systematically differ from related works?}}
While following standard evaluation protocols (datasets and metrics), \modelname~fundamentally differs in its overall paradigm. We illustrate their comparison in Table~\ref{tab:compare}.

\paragraph{\textbf{Q8. Does constructing the KG index introduce prohibitive costs for real-world applications?}}
While the KG construction incurs an initial overhead, it is a one-time and offline investment. Once indexed, the graph supports infinite queries with high efficiency. As the KG construction is not our main contribution, and as we follow the same experiment protocols in~\citet{jimenez2024hipporag, li2025simple} and~\citet{luo2025gfm}, we directly perform evaluations on their released KGs, which ensures the fair comparisons.
\section{Additional Related Work}
\label{app:related_work}
\subsection{Retrieval-Augmented Generation}
\label{app:related_work_rag}
Retrieval-augmented generation (RAG) is a general paradigm that enhances generative large models by enabling explicit access to external knowledge during inference. Instead of relying solely on information stored in model parameters, RAG systems retrieve relevant evidence conditioned on the input and incorporate it into the generation process. This design has demonstrated strong effectiveness across diverse real-world applications, including financial analysis~\cite{wang2025omnieval}, healthcare systems~\cite{xu2024ram}, legal reasoning tasks~\cite{wiratunga2024cbr}, and education~\cite{miladi2024leveraging}. With the rapid development of LLMs, retrieval augmentation has become increasingly important, as purely parametric models often suffer from hallucinated outputs, outdated knowledge, and limited adaptability. Prior surveys indicate that grounding generation in retrieved evidence substantially improves faithfulness, reduces privacy risks, and enhances robustness under distributional shifts~\cite{zhao2023survey, huang2023survey}. Consequently, RAG has emerged as a widely adopted strategy for extending the knowledge capacity of LLMs without requiring full model retraining.



Most retrieval-augmented generation methods adopt a ``retrieve-then-generate'' paradigm. Retrieved text units are then incorporated into the generator as additional context through concatenation, attention-based fusion, or encoder-decoder conditioning~\cite{fan2024survey}. Compared with keyword-based retrieval, dense retrievers provide higher semantic recall and better alignment with downstream generation. Building on this framework, later RAG variants explore tighter retrieval-generation coupling by operating at different granularities, including sentence-level~\cite{karpukhin2020dense, guu2020retrieval}, chunk-level~\cite{lewis2020retrieval, zhang2025sage}, and token-level~\cite{borgeaud2022improving, xu2024theory}, enabling dynamic evidence integration during decoding. Other approaches further constrain generation to improve grounding and verifiability. Despite architectural differences, a shared principle across RAG systems is to treat retrieval as a modular and updatable component, decoupling external knowledge access from the generator’s parametric capacity.

However, most existing RAGs retrieve unstructured text and treat evidence as an unordered set of contexts. Although effective for many knowledge-intensive tasks, this formulation largely ignores relational dependencies and higher-order structure among retrieved items, motivating graph-based extensions such as GraphRAG.

\subsection{Graph-based Retrieval-Augmented Generation}
\label{app:related_work_graphrag}
Graphs provide an expressive representation for structured relationships and have been widely applied in knowledge representation, social network analysis, and biomedical modeling~\cite{ma2021deep, wu2023survey}. Motivated by the limitations of text-only retrieval, Graph-based retrieval-augmented generation (GraphRAG) has emerged as an extension of RAG that incorporates graph-structured knowledge into the retrieval process~\cite{zhang2025survey}. In GraphRAG, concepts are typically represented as nodes, while edges encode semantic, logical, or relations, enabling retrieval to exploit explicit structure rather than relying solely on unstructured textual similarity~\cite{edge2024local, peng2024graph, zhou2025depth}. Recent studies explore this paradigm in settings involving structured graphs such as knowledge graphs and molecular graphs~\cite{han2024retrieval, peng2024graph}, highlighting its potential for relation-aware retrieval and reasoning.

Beyond retrieving from pre-existing graphs, a growing body of work focuses on constructing graphs from unstructured text. By organizing textual information into a graph, GraphRAGs enable structured access to implicit relational knowledge, benefiting tasks such as document summarization~\cite{edge2024local}, planning~\cite{lin2024graph}, and multi-step reasoning~\cite{han2025reasoning}. Within this line of research, graphs act as intermediate representations that facilitate abstraction, aggregation, and traversal-based evidence selection, complementing dense text retrieval. Existing methods can be grouped into several recurring design patterns. Hierarchical graph construction approaches organize knowledge through tree structures or community-level abstractions, enabling coarse-to-fine retrieval over large corpora, as exemplified by \texttt{RAPTOR}~\cite{sarthi2024raptor} and Microsoft’s \texttt{GraphRAG}~\cite{edge2024local}. Neural graph retrieval methods integrate graph neural encoders with specialized objectives to support multi-hop reasoning and scalable retrieval, including \texttt{G-Retriever}~\cite{he2024g}, \texttt{LightRAG}~\cite{guo2024lightrag}, and \texttt{GFM-RAG}~\cite{luo2025gfm}. Other approaches emphasize dynamic graph construction and adaptive traversal tightly coupled with large language models, allowing retrieval policies and graph structures to evolve during inference, such as \texttt{DALK}~\cite{li2024dalk}, \texttt{KGP}~\cite{wang2024knowledge}, and \texttt{ToG}~\cite{sun2023think}. In addition, \texttt{HippoRAG}~\cite{jimenez2024hipporag} adopts a memory-inspired formulation based on Personalized PageRank to enable efficient single-step multi-hop retrieval, demonstrating that graph-based propagation mechanisms can substantially improve reasoning efficiency.

Despite recent progress, many GraphRAG methods retrieve over-dense graphs, often introducing redundant structure and yielding contexts that are not task-sufficient. Moreover, reliance on fixed graph construction limits generalization across tasks, highlighting the open challenge of identifying compact, task-relevant subgraphs that retain essential relational information.

\subsection{Graph Foundation Models}
\label{app:related_work_gfm}
Graph foundation models aim to learn transferable representations from large-scale graph data via self-supervised pre-training~\cite{liu2022graph, xie2022self}. By leveraging objectives such as attribute corruption, structural reconstruction, and contrastive learning, these models seek to capture generalizable structural and semantic patterns that can be reused across tasks. A broad line of work has explored this paradigm, including \texttt{GCC}~\cite{qiu2020gcc}, \texttt{GPT-GNN}~\cite{hu2020gpt}, \texttt{GPPT}~\cite{sun2022gppt}, \texttt{L2P-GNN}~\cite{lu2021learning}, and subsequent variants~\cite{xu2023better,jiang2021contrastive,yang2022self,chen2022pre,cao2023pre,yin2023train}. After pre-training, the learned representations are typically adapted through fine-tuning or prompting. However, most existing approaches implicitly assume limited distributional discrepancy between pre-training and downstream graphs, often focusing on homogeneous subgraphs~\cite{huang2022few} or narrowly scoped domains~\cite{zhang2021motif}. As a result, their transferability degrades under domain shift.

To alleviate this issue, recent research has investigated cross-domain and multi-domain GFMs, aiming to reduce distribution gaps and improve generalization. One line of work focuses on extracting domain-invariant representations to support transfer across dissimilar domains, as explored in \texttt{GCOPE}~\cite{zhao2024all}, \texttt{OMOG}~\cite{liu2024one}, \texttt{PGPRec}~\cite{yi2023contrastive}, \texttt{UniAug}~\cite{tang2024cross}, and \texttt{UniGraph}~\cite{he2024unigraph}. Another direction incorporates language models to assist cross-domain alignment by mapping graphs into textual space, such as \texttt{OFA}~\cite{liu2024one} and \texttt{HiGPT}~\cite{tang2024higpt}, though such methods are restricted to text-attributed graphs~\cite{wang2024can,tan2024walklm}. Complementary strategies introduce structural or token-level mechanisms to bridge domains, including virtual nodes in \texttt{GCOPE}~\cite{zhao2024all} and domain tokens in \texttt{MDGPT}~\cite{yu2024text}, yet they overlook semantic consistency over heterogeneous domains.

There are some attempts to use GFMs as retrievers in GraphRAG, exemplified by \texttt{GFM-RAG}~\cite{luo2025gfm}. However, GFM pre-training acts as a semantic black box and does not explicitly align semantics across domains, causing retrieval to rely on implicit representation similarity and limiting robustness under domain shift.

\subsection{Critical Subgraph Mining}
\label{app:related_work_subgraph}
Graph structures contain rich relational semantics, yet only a small fraction of them are truly decisive for downstream tasks. Though structures have long been recognized as fundamental to graph representation learning~\cite{yang2018node, alsentzer2020subgraph}, systematic approaches for identifying task-relevant and data-dependent subgraphs remain underexplored.

Early efforts mainly relied on predefined structure decomposition, most notably graph kernels~\cite{kriege2020survey}, which extract specific patterns such as random walks~\cite{gartner2003graph}, shortest paths~\cite{borgwardt2005shortest}, or rooted subtrees~\cite{shervashidze2011weisfeiler}. These methods benefit from strong inductive bias and can achieve competitive results in certain domains. However, the structural patterns encoded by kernel functions are typically handcrafted and fixed, resulting in limited flexibility and poor generalization when task requirements or graph semantics vary. 
Another prominent direction focuses on small recurring local patterns, especially motifs, to preserve fine-grained structural properties~\cite{monti2018motifnet, lee2019graph, peng2020motif}. \texttt{NEST}~\cite{yang2018node} further combines multiple motifs to capture richer subgraph-level interactions. Despite their effectiveness, motif-based approaches fundamentally rely on exhaustive enumeration of small substructures, which requires domain expertise and limits scalability to larger or irregular patterns.

To move beyond rigid neighborhoods, several studies investigate more flexible subgraph modeling strategies~\cite{sun2021sugar, li2023adaptive, lee2025pre, li2025disentangling}. \texttt{SGN}~\cite{xuan2019subgraph} expands the structural feature space via manually specified subgraph selection rules, but such heuristics become less informative as subgraphs grow larger. \texttt{NCAT}~\cite{zuo2021neighbor} assigns learnable weights to combinatorial neighbor patterns, though it still operates under fixed-size constraints. From an information-theoretic perspective, \texttt{SIB}~\cite{yu2020graph, yu2021recognizing} learns edge-wise assignments to identify bottleneck subgraphs, but it does not explicitly preserve subgraph semantics. \texttt{CAL}~\cite{sui2022causal} further explores causal subgraph discovery by separating causal signals from shortcut correlations through attention-based estimation.

Subgraphs are rarely used for retrieval. \texttt{SubgraphRAG}~\cite{li2025simple} retrieves query-centered subgraphs via triple-level scoring under a size budget, but focuses on coverage rather than task-specific sufficiency, resulting in non-minimal retrieved structures and leaving critical subgraph identification unaddressed.
\section{Limitations}
\label{app:limitation}
Despite the promising performance of the proposed \modelname~in cross-domain retrieval and multi-hop reasoning, we acknowledge several limitations that define the scope of our current work and point towards future research directions

\textbf{Dependency on Knowledge Graph Construction.} 
Our framework primarily focuses on the retrieval and reasoning stages, operating on the premise of a pre-constructed knowledge graph. While we utilize advanced OpenIE techniques driven by LLMs for indexing, the upper bound of retrieval performance is naturally constrained by the completeness of this upstream knowledge graph construction. In scenarios where the raw corpora are extremely noisy or the extracted triples suffer from severe sparsity, the semantic expressiveness of the GFM could be theoretically limited by the quality of the underlying structure.

\textbf{Trade-off in Path Extraction Depth.}
To ensure real-time inference efficiency, we employ a truncated Depth-First Search (DFS) strategy with a pre-defined hop limit ($l \leqslant 3$) to extract reasoning paths from the selected subgraph. While our empirical analysis (Appendix~\ref{app:complexity_inference}) confirms that this setting covers the reasoning complexity of most standard benchmarks, it may introduce constraints when handling hyper-complex queries that require extremely long-range dependencies. Extending the path extraction mechanism to support deeper traversal without incurring exponential computational overhead remains a challenge for future optimization.

\textbf{Adaptability to Knowledge Updates.}
\modelname~leverages a ``pretrain-then-finetune'' paradigm, which assumes a relatively static knowledge snapshot during the inference phase. Although the Graph Foundation Model exhibits strong generalization across domains, adapting the indexed knowledge graph to real-time data streams or highly temporal facts without full re-indexing poses a practical challenge. We envision future work incorporating continual learning to address this limitation in time-sensitive applications.

\section{Future Directions}
\label{app:future}
Building upon the structural retrieval paradigm established, we envision several promising avenues to further advance the synergy between GFMs and LLMs:

\textbf{Dynamic and Continual GraphRAG.}
Real-world knowledge is inherently temporal and evolving. A natural extension of this work is to integrate continual graph learning mechanisms into the GFM backbone. Future research could explore efficient inductive updates that enable the retriever to dynamically integrate new entities and relations from streaming data without requiring full-scale re-indexing. This would enable the system to handle time-sensitive queries in scenarios such as breaking news analysis or financial monitoring.

\textbf{Alignment via Graph Instruction Tuning.}
While the proposed \modelname~currently bridges graphs and LLMs via in-context prompting, tighter integration could yield further performance gains. We plan to explore parameter-efficient joint training strategies that align the representation space of the GFM directly with the LLM. By treating subgraph retrieval as a latent step in a unified instruction-tuning pipeline, we aim to construct integrated architectures capable of understanding complex structural instructions natively.

\textbf{GFM-driven Autonomous Agents.}
Beyond question answering tasks, the ability to retrieve minimal and sufficient subgraphs serves as a critical capability for autonomous agents navigating complex environments. We intend to adapt \modelname~as a ``cognitive map'' module for agents, where the retrieved subgraph represents the state space or action constraints. This could empower agents to perform multi-step planning and reasoning in knowledge-intensive environments, such as code repository navigation or scientific discovery.

\end{document}